\documentclass{CSML}
\pdfoutput=1

\usepackage{lastpage}
\newcommand{\dOi}{13(4:4)2017}
\lmcsheading{}{1--\pageref{LastPage}}{}{}%
{Nov.~29, 2016}{Oct.~30, 2017}{}


\ACMCCS{[\textbf{Theory of computation}]: Database theory---Logic and databases.}

\usepackage{hyperref}
\hypersetup{hidelinks}

\usepackage{amsmath}
\usepackage{amsfonts}
\usepackage{amssymb}
\usepackage{latexsym}
\usepackage{alltt}
\usepackage{enumitem}
\usepackage{tikz}
\usepackage{bm}
\usepackage{tabularx}
\usepackage{mathrsfs}
\usepackage{graphicx}
\usepackage{subfig}
\usepackage{caption}
\usepackage{array}
\usepackage{pifont}

\newtheoremstyle{thmC}%
  {6pt}
  {6pt}
  {\itshape}
  {}
  {\bfseries}
  {{\bfseries .}}
  {5pt plus 1pt minus 1pt}
  {\thmname{#1} \thmnumber{#2} \thmnote{\normalfont#3}}

\theoremstyle{thmC}

\newtheorem{lemC}[thm]{Lemma}

\newcommand{\PFD}{\protect\reflectbox{P}}

\newcommand{\cO}{\ensuremath{D}}  
\newcommand{\cE}{\ensuremath{\mathcal{E}}}  
\newcommand{\cL}{\ensuremath{\mathcal{L}}}  
\newcommand{\gA}{\ensuremath{\mathfrak{A}}}  
\newcommand{\gB}{\ensuremath{\mathfrak{B}}}  

\newcommand{\br}{\ensuremath{\bar{r}}}  
\newcommand{\bt}{\ensuremath{\bar{t}}}
\newcommand{\bs}{\ensuremath{\bar{s}}}

\newcommand{\ba}{\ensuremath{\bar{a}}}
\newcommand{\bb}{\ensuremath{\bar{b}}}
\newcommand{\bc}{\ensuremath{\bar{c}}}
\newcommand{\bdd}{\ensuremath{\bar{d}}}

\newcommand{\baf}{\ensuremath{\bar{f}}}
\newcommand{\bag}{\ensuremath{\bar{g}}}
\newcommand{\bah}{\ensuremath{\bar{h}}}

\newcommand{\xO}{\mathtt{0}}
\newcommand{\xI}{\mathtt{1}}

\newcommand{\s}{\ensuremath{\texttt{s}}}    
\newcommand{\sO}{\ensuremath{\texttt{s0}}}  
\newcommand{\sI}{\ensuremath{\texttt{s1}}}
\renewcommand{\t}{\ensuremath{\texttt{t}}}  
\newcommand{\tO}{\ensuremath{\texttt{t0}}}
\newcommand{\tI}{\ensuremath{\texttt{t1}}}
\newcommand{\lps}{\ensuremath{\Lambda_{\pi, \s}}}
\newcommand{\mpt}{\ensuremath{M_{\pi, \t}}}

\newcommand{\GF}{\ensuremath{\mathcal{G\:\!\!F}}}
\newcommand{\GC}{\ensuremath{\mathcal{GC}^2}}
\newcommand{\GCD}{\ensuremath{\mathcal{GC}^2\mathcal{D}}}
\newcommand{\ALCQI}{\ensuremath{\mathcal{ALC\!QI}}}
\newcommand{\GCDK}{\ensuremath{\mathcal{GC}^2\mathcal{DK}}}

\newcommand{\tp}[2]{\ensuremath{\mathrm{tp}_{#1}[#2]}}
\newcommand{\tpA}[1]{\tp{\gA}{#1}}

\newcommand{\tpStart}[1]{\ensuremath{\mathrm{tp}_1(#1)}}
\newcommand{\tpEnd}[1]{\ensuremath{\mathrm{tp}_2(#1)}}
\renewcommand{\sp}[3]{\ensuremath{\mathrm{sp}_{#1}^{#2}[#3]}}
\newcommand{\spA}[2]{\sp{#1}{\gA}{#2}}
\newcommand{\tl}[3]{\ensuremath{\mathrm{tl}_{#1}^{#2}[#3]}}
\newcommand{\tlA}[2]{\tl{#1}{\gA}{#2}}

\newcommand{\x}{\ensuremath{\mathtt{x}}}
\newcommand{\y}{\ensuremath{\mathtt{y}}}
\newcommand{\bitb}{\ensuremath{\mathtt{b}}}

\newcommand{\branch}[2]{\ensuremath{\mathrm{branch}_{#2}\langle #1 \rangle}}
\newcommand{\fan}[1]{\ensuremath{\mathrm{fan}\langle #1 \rangle}}
\newcommand{\isth}[1]{\ensuremath{\mathrm{isth}\langle #1 \rangle}}

\newcommand{\dst}[1]{\ensuremath{\mathrm{dst}( #1 )}}
\newcommand{\frev}{\ensuremath{f^{-1}}}

\newcommand{\seq}{\ensuremath{\bar{\mathbf r}_\kappa}}

\renewcommand{\phi}{\varphi}

\newcommand{\im}{\ensuremath{1 \le i \le m}}

\renewcommand{\bf}{\ensuremath{\mathbf{f}}}
\newcommand{\bg}{\ensuremath{\mathbf{g}}}
\newcommand{\bv}{\ensuremath{\bm{v}}}
\newcommand{\bu}{\ensuremath{\bm{u}}}
\newcommand{\bw}{\ensuremath{\bm{w}}}
\newcommand{\bO}{\ensuremath{\bm{0}}}
\newcommand{\bC}{\ensuremath{\bm{C}}}

\newcommand{\hy}{\ensuremath{\hat{y}}}
\newcommand{\hz}{\ensuremath{\hat{z}}}

\newcommand{\p}{\ensuremath{\textsc{PTime}}}
\newcommand{\np}{\ensuremath{\textsc{NPTime}}}
\newcommand{\dexp}{\ensuremath{\textsc{ExpTime}}}
\newcommand{\nexp}{\ensuremath{\textsc{NExpTime}}}

\newcommand{\Ccom}{\mbox{CC}}
\newcommand{\CcomS}{\mbox{CCS}}
\newcommand{\CcomT}{\mbox{CCT}}

\newcommand{\ignore}[1]{}

\newcommand{\sizeOf}[1]{\ensuremath{|#1|}}

\usetikzlibrary{shapes.arrows}

\def\vertex(#1, #2); {
  \draw[fill] (#1, #2) circle[radius=3pt];
}

\def\triangleUp(#1, #2); {
  \draw[line width=.5pt] (#1 - .065, #2 + .11) -- (#1 - .7, #2 + 1.5)
                         -- (#1 + .7, #2 + 1.5) -- (#1 + .065, #2 + .11);
}

\def\triangleDown(#1, #2); {
  \draw[line width=.5pt] (#1 - .065, #2 - .11) -- (#1 - .7, #2 - 1.5)
                         -- (#1 + .7, #2 - 1.5) -- (#1 + .065, #2 - .11);
}

\def\triangleRight(#1, #2); {
  \draw[line width=.5pt] (#1 + .065, #2 + .11) -- (#1 + 1, #2 + 1.2)
                         -- (#1 + 1, #2 - 1.2) -- (#1 + .065, #2 - .11);
}

\def\triangleBigRight(#1, #2); {
  \draw[line width=.5pt] (#1 + .065, #2 + .11) -- (#1 + 1.25, #2 + 1.2)
                         -- (#1 + 1.25, #2 - 1.2) -- (#1 + .065, #2 - .11);
}

\def\triangleLeft(#1, #2); {
  \draw[line width=.5pt] (#1 - .065, #2 + .11) -- (#1 - 1, #2 + 1.2)
                         -- (#1 - 1, #2 - 1.2) -- (#1 - .065, #2 - .11);
}

\def\triangleBigLeft(#1, #2); {
  \draw[line width=.5pt] (#1 - .065, #2 + .11) -- (#1 - 1.25, #2 + 1.2)
                         -- (#1 - 1.25, #2 - 1.2) -- (#1 - .065, #2 - .11);
}

\def\triangleRotate(#1, #2, #3); {
  \draw[line width=.5pt, rotate around={-#3:(#1,#2)}] (#1 - .065, #2 + .11) -- (#1 - .7, #2 + 1.5)
                         -- (#1 + .7, #2 + 1.5) -- (#1 + .065, #2 + .11);

}

\def\sandglass[(#1, #2), #3, #4]; {
  \vertex(#1, #2);
  \triangleUp(#1, #2);
  \triangleDown(#1, #2);
  
  \node at (#1, #2 + 1) {#3};
  \node at (#1, #2 - 1) {#4};
}

\def\halfSandRight[(#1, #2), #3]; {
  \vertex(#1, #2);
  \triangleRight(#1, #2);
  \node at (#1 + .6, #2) {#3};
}

\def\arcRight[(#1, #2), #3]; {
  \vertex(#1, #2);
  \node at (#1 + .6, #2) {#3};
  \draw (#1+1.2, #2) arc (0:360:17pt);
}

\def\doubleArrow[(#1, #2), (#3, #4), #5, #6]; {
  \draw[line width=.5pt, ->] (#1 + .12, #2 + .05) -- (#3 - .12, #4 + .05);
  \draw[line width=.5pt, <-] (#1 + .12, #2 - .05) -- (#3 - .12, #4 - .05);
  \node at ({#1 + #3)/2}, {(#2+#4)/2 + .25}) {#5};
  \node[below=-8pt] at ({#1 + #3)/2}, {(#2+#4)/2 - .3}) {#6};
}

\def\arrowCollinear[(#1, #2), (#3, #4)]; {
  \draw[line witdh=.4pt, ->] ({(#1 + #3)/2 - .2, (#2 + #4)/2}) --
                             ({(#1 + #3)/2, (#2 + #4)/2});
}

\numberwithin{equation}{section}

\begin{document}

\title[Adding Path-Functional Dependencies to $\GC$]{Adding Path-Functional Dependencies to the Guarded Two-Variable Fragment with Counting}

\author{Georgios Kourtis}
\address{School of Computer Science, The University of Manchester, UK}
\email{\{kourtisg,ipratt\}@cs.man.ac.uk}

\author{Ian Pratt-Hartmann}




\begin{abstract}
  \noindent The satisfiability and finite satisfiability problems for the two-variable
  guarded fragment of first-order logic with counting quantifiers, a database, and path-functional
  dependencies are both $\dexp$-complete.
\end{abstract}

\maketitle

\section{Introduction}
\label{sec:introduction}
In the theory of information systems, a {\em path-functional dependency} is a data 
constraint stating that individuals yielding identical values on application of various sequences of functions must themselves be identical. For example, in a database of customers, we may wish to impose the
condition that no two individuals have the same customer ID:
\begin{align}
& \forall x \forall y (\mathrm{custID}(x) = \mathrm{custID}(y) \rightarrow x = y).
\label{eq:pfd1}
\end{align}
Alternatively,
we may wish to impose the
condition that no two customers have both the same name and postal codes:
\begin{multline}
\forall x \forall y (\mathrm{name}(x) = \mathrm{name}(y) \, \wedge\\ \mathrm{postCode}(\mathrm{address}(x)) = 
 \mathrm{postCode}(\mathrm{address}(y)) \rightarrow x = y).
\label{eq:pfd2}
\end{multline}
The formula in~\eqref{eq:pfd1} has a single equality on the left-hand side of the implication; and therefore
we speak of the condition it expresses as a {\em unary} path-functional dependency. By contrast,
\eqref{eq:pfd2} has a conjunction of two equalities on the left-hand side of the implication; and
we speak of the condition it expresses as a {\em binary} path-functional dependency. 

Path-functional dependencies were introduced by Weddell~\cite{weddell89} (see also~\cite{iw94,bp00}), and are of particular interest when combined with data-integrity constraints represented using formulas of some \textit{description logics}~\cite{tw05a,tw05b}. In this connection,
Toman and Weddell~\cite{tw08} report various results on the decidability of the satisfiability problem for
description logics extended with various forms of path-functional constraints. The purpose of the present paper is to extend those results to a more general logical setting, namely that of the two-variable guarded fragment with counting quantifiers.

The \textit{guarded fragment} of first-order logic, denoted $\GF$, was introduced by Andr\'eka \emph{et~al.} \cite{andreka1998modal} as a generalization of modal logic, in an attempt to explain the latter's good computational behaviour (see e.g.~\cite{vardi1996modal,gradel1999modal}).
Roughly speaking, $\GF$ is the fragment of first-order logic in which all quantification is relativized
(guarded) by atoms featuring all variables free in the context of quantification. Thus, for example, in the
guarded formula 
\begin{equation*}
\forall x (\mathrm{customer}(x) \rightarrow \exists y (\mathrm{lastName}(x,y) \wedge \mathrm{string}(y))),
\end{equation*}
stating that every customer has a last name (which is of type string),
the quantifier $\exists y $ is relativized by the binary atom $\mbox{lastName}(x,y)$.
This fragment has the so-called finite model property: if a $\GF$-formula is satisfied in some structure, it is satisfied in some finite structure. The problem of determining whether a given
$\GF$-formula is satisfiable is 2-$\dexp$-complete \cite{gradel1999restraining}, dropping
to $\dexp$-complete if the number of variables in formulas is less than some fixed limit $k \geq 2$. Of particular interest in this regard is the case $k=2$, since this fragment remains 
$\dexp$-complete even when we add so-called \textit{counting quantifiers}
(`there exist at most/at least/exactly $m$ $x$ such that \dots'). The resulting logic, denoted $\GC$,
allows us to state, for example, that 
customers are uniquely identified by their 
customer identification numbers:
\begin{equation*}
\forall x (\mathrm{string}(x) \rightarrow \exists_{\leq 1} y (\mathrm{custID}(y,x) \wedge \mathrm{customer}(y))).
\end{equation*}
Indeed, $\GC$ subsumes the well-known description logic $\ALCQI$.

The logic $\GC$ lacks the finite model property. Nevertheless,  
the satisfiability and finite satisfiability problems for this fragment are both decidable, and in fact $\dexp$-complete~\cite{kazakov2004polynomial, pratt07}. 

From a logical point of view, a \textit{database} is simply a finite set of ground literals---i.e.~atomic statements or negated atomic statements featuring only constants as arguments: 
\begin{align*}
& \mathrm{custID}(\#111001, 5030875), \quad
 \mathrm{firstName}(\#111001, \mbox{`Fred'}), \quad \dots
\end{align*}
We assume that such data are subject to \emph{constraints} 
presented in the form of a logical theory---as it might be, a collection of formulas of $\GC$, for instance.
Those constraints may imply the existence of individuals not mentioned in the database: indeed, since 
$\GC$ lacks the finite model property, they may consistently imply the existence of infinitely many such individuals. We denote by $\GCD$ the logic that results from adding a database to $\GC$; and we denote by 
$\GCDK$ the logic that results from adding unary and binary path-functional dependencies to $\GCD$. (A formal definition
is given in Sec.~\ref{sec:preliminaries}.)

The \textit{satisfiability problem} for $\GCDK$ asks the following:
given a database $\Delta$, a $\GC$-formula $\phi$ and a finite set of unary and binary path-functional dependencies $K$, are the data in $\Delta$ consistent with the constraints expressed in $\phi$ and $K$?
The \textit{finite satisfiability problem} for $\GCDK$ asks the same question, but subject to the assumption of a finite universe. We show that both of these problems are $\dexp$-complete, the
same as the corresponding problems for $\GC$ alone. That is: adding databases and unary or binary path-functional dependencies to $\GC$ does not change the complexity class of the  satisfiability and finite satisfiability problems.

Our strategy will be to reduce the satisfiability and finite satisfiability
problems for $\GCDK$ to the corresponding problems for 
$\GCD$, by showing how unary and binary path-functional dependencies can be systematically eliminated. The case of unary path-functional dependencies is simple, and is dealt with in the preliminary material of Sec.~\ref{sec:preliminaries}. Binary path-functional dependencies, however, present a greater challenge, and occupy the whole of Sections~\ref{sec:decompositions}--\ref{sec:gcdk}.
Our point of departure here is the familiar observation that models of $\GC$-formulas, when 
viewed as graphs, can be assumed to contain no `short' cycles. We use this observation to characterize putative violations of binary path-functional dependencies in terms of 
the occurrence of certain acyclic subgraphs, which can then be forbidden by writing additional $\GC$-formulas.
The remainder of the paper, Sections~\ref{sec:types}--\ref{sec:main_result}, is devoted to establishing that the satisfiability and finite satisfiability problems for
$\GCD$ are both in $\dexp$. The argument here 
proceeds by reduction to linear programming feasibility, closely following the proof given in~\cite{pratt07} that the finite satisfiability of $\GC$ is in $\dexp$. We show how to modify that 
proof so as to accommodate the presence of a database.

\section{Preliminaries}
\label{sec:preliminaries}

\subsection{The logics $\GC$, $\GCD$ and $\GCDK$}
A \textit{literal} is an atomic formula or the negation of an atomic formula; a literal is \textit{ground} if all arguments occurring in it are constants. If $\phi$ is any formula, the \textit{length} of $\phi$, denoted $\|\phi\|$, is the number of symbols it contains. In addition to the usual boolean connectives, quantifiers $\forall$, $\exists$,
and equality predicate, we employ the \textit{counting  quantifiers} $\exists_{\leq C}$, $\exists_{\geq C}$ and $\exists_{= C}$ (for all $C \geq 0$). A \textit{sentence} is a formula with no free variables.
In most of this paper, we consider \textit{two-variable} formulas---that is, those whose only variables (free or bound) are $x$ and $y$. Furthermore, we restrict attention to signatures comprised of individual constants or predicates of arity 1 or 2; in particular, function symbols are not allowed. (Function symbols were used in Section~\ref{sec:introduction}
for the sake of presentation, but we
introduce different notation for path-functional dependencies in the sequel.) We call literals involving only symbols from some signature $\sigma$ $\sigma$-\textit{literals}.

An atomic formula of the form $p(x, y)$ or $p(y, x)$ is called a \emph{guard}.
The guarded two-variable fragment with counting, denoted $\GC$ is the smallest set of formulas satisfying the following conditions:
\begin{itemize}
  \item unary or binary literals with all arguments in $\{x,y\}$ are formulas, with $=$ allowed as a binary predicate;
  \item the set of formulas is closed under Boolean combinations;
  \item if $\phi$ is a $\GC$-formula with at most one free variable and $u$ is a variable (i.e.~either
        $x$ or $y$), then $\forall u.\phi$ and $\exists u.\phi$ are formulas.
  \item if $\phi$ is a formula, $\alpha$ is a guard, $u$ is a variable and $C$ is a bit-string, then $\exists_{\leq C}u(\alpha \land \phi)$, $\exists_{\geq C}u(\alpha \land \phi)$ and $\exists_{= C}u(\alpha \land \phi)$ are formulas.
\end{itemize}

We read $\exists_{\leq C} u. \phi$ as `there exist at least $C$ $u$ such that $\phi$', and similarly for
$\exists_{\geq C}$ and $\exists_{= C}$. The formal semantics is as expected. To improve readability, we write 
$\exists_{\leq 0} u(\alpha \land \neg \phi)$ as $\forall u(\alpha \rightarrow \phi)$, and $\exists_{\geq 1}u(\alpha \land \phi)$ as $\exists u(\alpha \land \phi)$. Where no confusion results, we equivocate between bit-strings and the integers they represent, it being understood that the size of the expression $\exists_{= C}$ is approximately {$\log C$}
(and not $C$). Observe that 
formulas of the form $\exists_{\leq C}x . p(x)$ are not in $\GC$, because counting quantifiers must be guarded by  
atoms featuring both variables. (This is not a gratuitous restriction: adding such formulas to $\GC$ renders the
fragment $\nexp$-time hard.) In the sequel, formulas that are obviously logically equivalent to a $\GC$-formula will typically be counted as $\GC$-formulas by courtesy.

A \emph{database} is a set $\Delta$ of
ground (function-free) literals. We call $\Delta$ \emph{consistent} if it includes no pair of ground literals 
$\{\alpha,\neg \alpha\}$. Where the signature
$\sigma$ is clear from context, we call $\Delta$
\emph{complete} if for any ground $\sigma$-literal $\lambda \not \in \Delta$, $\Delta \cup \{ \lambda \}$ is inconsistent. 
Given a database $\Delta$,
a \emph{completion} for $\Delta$ is any complete set $\Delta^\ast \supseteq \Delta$
of ground (function-free) $\sigma$-literals. It is obvious that every consistent database has a
consistent completion.  
Databases are interpreted as sets of atomic formulas in the expected way. For convenience, we shall
employ the unique names assumption throughout: distinct individual constants are interpreted as distinct individuals. Hence, when a particular interpretation is clear from context, we sometimes equivocate
between constants and their denotations to streamline the presentation. To further reduce notational clutter,
we allow ourselves to treat a database $\Delta$ as a single conjunctive formula---i.e., writing
$\Delta$ in place of $\bigwedge \Delta$.

The \textit{two-variable guarded fragment with counting and databases}, denoted $\GCD$, is defined exactly as for $\GC$,
except that we have the additional syntax rule
\begin{itemize}
  \item ground unary or binary literals are formulas.
  \end{itemize}
Although $\GCD$ allows formulas in which ground literals appear within the scope of quantifiers, in practice, it is much more natural to separate out the `$\GC$-part' from the
`database-part': thus all $\GCD$-formulas encountered in the sequel will have the form $\phi \wedge \Delta$, where $\phi$ is a $\GC$-formula, and $\Delta$ a database. Notice, however, that, in $\GCD$, it is forbidden to mix variables and
individual constants in atoms: thus, for example, $p(x,c)$ is not a $\GCD$-formula. This
is an essential restriction; if mixing variables and constants is allowed, it is
easy to encode a grid of exponential size and, as a result, runs of an exponential-time
Turing machine, leading to an \textsc{NExpTime} lower bound.
If $\gA$ is any structure, we call those elements interpreting an individual constant {\em active}, and we refer 
to the set of all such elements as the {\em active domain} (or sometimes, informally, as \textit{the database}). 
Elements which are not active are called {\em passive}.

We assume that any signature $\sigma$ features  a (possibly empty) distinguished subset of binary predicates, which we refer to as \emph{key predicates}. Key predicates are always interpreted as the graphs of irreflexive functions. That is, if $f$ is a key predicate, then in any structure $\gA$ interpreting $f$, $\gA \models \forall x \exists_{\leq 1} y \, f(x, y)$ and $\gA \models \forall x \exists y (f(x, y) \wedge x \neq y)$. (We remark that this formula is in $\GC$.) 
We use the (possibly subscripted or otherwise decorated) letters $f$, $g$, $h$ to range over key predicates,
and we use $\baf$, $\bag$, $\bah$ for words over the alphabet of key predicates---i.e., finite sequences of key predicates.
Warning: $\baf_1$, $\baf_2$ etc.~will always be taken to denote whole sequences of key predicates, not the individual elements of some such sequence $\baf$. 

If $\gA$ is a structure, $\bar{f}$ a word $f_0 \cdots f_{k-1}$ over the alphabet of key predicates interpreted by $\gA$,
and $a \in A$, we write $\baf^\gA(a)$
to denote the result of successively applying the corresponding functions in $\baf$ to $a$.
More formally, $\baf^\gA(a) = a_k$ where
$a_0 = a$ and, for all $i$ ($0 \leq i < k$),
$a_{i+1}$ is the unique $b \in A$ such that $\gA \models f_i(a_i,b)$.
An $n$-ary {\em path-functional dependency} $\kappa$ is an expression
\begin{equation}
\PFD[\baf_1, \dots, \baf_n]
\label{eq:kc}
\end{equation}
where, for all $i$ ($1 \leq i \leq n$), $\baf_i$ is a word over the alphabet of key predicates.
Thus, $\PFD[\ \ ]$ is simply a piece of logical syntax that allows us to construct a formula 
from an $n$-tuple of sequences of key predicates.
For any interpretation $\gA$, we take the dependency~\eqref{eq:kc} to be {\em satisfied} in $\gA$, and write 
$\gA \models \kappa$, if, for all $a, b \in A$, 
\begin{center}
	$\baf_i^\gA(a) =  \baf_i^\gA(b)$ for all $i$ ($1 \leq i \leq n$) implies $a = b$.
\end{center}
That is to say, $\kappa$ is satisfied in $\gA$ if any two elements of $A$ which agree on the result of applying each of the function sequences corresponding to the words $\baf_1$ to $\baf_n$
are in fact identical. In this paper we shall consider only unary and binary path-functional dependencies, i.e.~$n = 1 \text{ or }2$.
The \textit{two-variable guarded fragment with counting, databases and \textup{(}unary or binary\textup{)} path-functional dependencies}, denoted $\GCDK$, is defined exactly as for $\GCD$,
except that we have the additional syntax rule
\begin{itemize}
  \item all unary and binary path-functional dependencies are formulas.
\end{itemize}
Again, to reduce notational clutter, 
we treat a set of path-functional dependencies
$K$ as a single conjunctive formula---i.e., writing $K$ in place of $\bigwedge K$.
Although $\GCDK$ allows formulas in which path-functional dependencies
appear within the scope of quantifiers, in practice, it is much more natural to separate out the `$\GC$-part' from the
database and the path-functional dependencies: thus all $\GCDK$-formulas encountered in the sequel will have the form $\phi \wedge \Delta \wedge K$, where $\phi$ is a sentence of $\GC$, $\Delta$ is database, and $K$ a set of path-functional dependencies.

Note that, although key predicates give us the ability to \textit{express} functional dependencies, they remain, syntactically speaking, predicates, not function-symbols. In fact in the logics considered here, there are no function-symbols. In particular the expressions~\eqref{eq:pfd1} and~\eqref{eq:pfd2} are for expository purposes only: in $\GCDK$, we should write, respectively, the unary and binary path-functional dependencies
\begin{equation*}
\PFD[\mathrm{custID}], \qquad \PFD[\mathrm{name}, \mathrm{postCode}~\mathrm{address}], 
\end{equation*}
where 
$\mbox{custID}$, $\mbox{postCode}$ and $\mbox{address}$ are binary (key) predicates.

Taking key predicates to denote total (as opposed to partial) functions represents no essential 
restriction, as partial functions can always be encoded using total functions interpreted over expanded domains featuring `dummy' objects. The restriction to \textit{irreflexive} functions, though not so easily
eliminable, is in most cases perfectly natural (the telephone number of a person is not a person), and greatly simplifies much of the ensuing argumentation, obviating the constant need to consider various special cases.

If $\cL$ is any of the languages $\GC$, $\GCD$ or $\GCDK$, a sentence $\phi$ of $\cL$ is (\textit{finitely})
\textit{satisfiable} if, for some (finite) structure $\gA$ interpreting its signature, $\gA \models \phi$.
We define the ({\em finite}) {\em satisfiability} problem for any of the logics
in the expected way: given a formula $\phi$ of $\cL$, return Y if $\phi$ is (finitely) satisfiable; N otherwise. It was shown in~\cite{kazakov2004polynomial} that 
the satisfiability problem for $\GC$ is \dexp-complete, and in~\cite{pratt07} that the finite satisfiability problem for $\GC$ is \dexp-complete. In this paper, we show that the same complexity bounds 
 apply to $\GCD$ and $\GCDK$. 
  
The following lemma assures us that we may confine attention to $\GCDK$ formulas of the standard form 
$\phi \wedge \Delta \wedge K$.
\begin{lem}
\label{lma:splitting}
Let $\psi$ be a $\GCDK$ formula. We can compute, in time bounded by exponential function of $\sizeOf{\psi}$, 
a set $\Psi$ of $\GCDK$ formulas with the following properties:
(i) each formula in $\Psi$ is bounded in size by
a polynomial function of $\sizeOf{\psi}$ and has the form $\phi \wedge \Delta \wedge K$, 
where $\phi$ is a $\GC$-formula, $\Delta$ a complete database and $K$ a set of path-functional dependencies; (ii) $\phi$ is (finitely) satisfiable if and only if some member 
of $\Psi$ is.
\end{lem}

\begin{proof}[Sketch proof]
Consider any mapping $\theta$ from the set of path-functional dependencies and ground atoms occurring in
$\psi$ to the logical constants $\{\top, \bot\}$, and let $K = \Delta = \emptyset$.
For every path-functional dependency $\kappa$ occurring in
$\psi$, if $\theta(\kappa) = \top$, add $\kappa$ to $K$; otherwise, add to $\Delta$ a sequence of literals (possibly with fresh individual constants) encoding a violation of $\kappa$ in the obvious way.  
Let $\psi'$ be the result of
replacing each path-functional dependency $\kappa$ in $\psi$ by $\theta(\kappa)$.
For every ground atom $\alpha$ occurring in
$\psi'$, if $\theta(\alpha) = \top$, add $\alpha$ to $\Delta$; otherwise, add $\neg \alpha$ to $\Delta$.
Let $\phi$ be the result of
replacing each ground atom $\alpha$ in $\psi'$ by $\theta(\alpha)$.
Let $\Psi$ be the set of
of all formulas obtained in this way, for all possible mappings $\theta$.
\end{proof}

The logic $\GC$ sometimes surprises us with its expressive power. The next lemma provides a simple example, that 
will prove useful in the sequel. Let $\dst{x, y, z}$ be an abbreviation
for the formula $x \neq y \land y \neq z \land x \neq z$, stating that $x$, $y$ and $z$ are distinct.
\begin{lem}
	\label{lem:rewrite}
	Let \begin{align*}
	\phi_1(x) &\;:=\; \exists y \exists z \big( \dst{x, y, z} \land \alpha(x, y) \land \beta(y, z) \big), \\
	\phi_2(y) &\;:=\; \exists x \exists z \big( \dst{x, y, z} \land \alpha(x, y) \land \beta(y, z) \big),
	\end{align*}
	be formulas over $\sigma$, where $\alpha(x, y)$ and $\beta(y, z)$ are $\GC$-formulas. 
	Then, we can compute, in polynomial time, $\GC$-formulas
	$\phi_1^\ast(x)$ and $\phi_2^\ast(y)$, which are logically equivalent to $\phi_1(x)$ and $\phi_2(y)$
	respectively.
\end{lem}

\begin{proof}
Let
\begin{gather*}
\hspace{-50pt}
\phi_1^\ast(x) \;:=\; \exists y \big( x \neq y \land \alpha(x, y) \land \lnot \beta(y, x) \land \exists x  ( x \neq y \land \beta(y, x)) \big) \\
\hspace{65pt} \lor \, \exists y \big( x \neq y \land \alpha(x, y) \land \beta(y, x) \land \exists_{\geq 2} x ( x \neq y \land \beta(y, x)) \big);
\\[8pt] 
\hspace{-50pt}
\phi_2^\ast(y) \;:=\; \exists x \big(x \neq y \land \alpha(x, y) \land \lnot \beta(y, x)  \land \exists x (x \neq y \land \beta(y, x)) \big) \\
\hspace{65pt} \lor \, \exists x \big(x \neq y \land \alpha(x, y) \land \beta(y, x) \land \exists_{\geq 2} x (x \neq y \land \beta(y, x)) \big).
\tag*{\qEd}
\end{gather*}
\def\popQED{}
\end{proof}

\subsection{Graphs of structures}  
If $f$ is a key predicate, we introduce a new binary predicate $f^{-1}$, referred to as the {\em converse} of $f$, and subject to the  requirement that, for any structure $\gA$, $\gA \models \forall x \forall y \big( f(x, y) \leftrightarrow f^{-1}(y, x) \big)$. (We remark that this formula is in $\GC$.) There is no requirement that $f^{-1}$ be functional, though of course it must be irreflexive. 
We refer to any key predicate $f$ or its converse $f^{-1}$ as a {\em graph predicate}. We use the (possibly subscripted or otherwise decorated) letters $r$, $s$, $t$ to range over graph predicates. If $r = f^{-1}$, we
take  $r^{-1}$ to denote $f$. In the sequel, we fix a signature $\sigma$ such that, for any key predicate $f$ in $\sigma$, $f^{-1}$ is
also in $\sigma$. That is: the set of graph predicates of $\sigma$ is closed under converse.
We use letters $\br$, $\bs$, $\bt$ for words over the alphabet of graph predicates---i.e., finite sequences of graph predicates. 
If $\br = r_1 \cdots r_\ell$, we write $\br^{-1}$ for the word $r_\ell^{-1} \cdots r_1^{-1}$. 

Let $\gA$ be a structure interpreting $\sigma$ over domain $A$, and let $E = E_1 \cup E_2$ 
be the set of unordered pairs of elements of $A$, given by
\begin{align*}
E_1 = & \{(a,b) \in A^2 \mid \gA \models r(a,b) \text{ for some graph predicate $r$ of $\sigma$}\}\\
E_2 = & \{(a,b) \in A^2 \mid \text{$a$ and $b$ are distinct and both active}\}.  
\end{align*}
Then
$G = (A,E)$ is a (possibly infinite) graph, and we refer to $G$ as the {\em graph of} $\gA$. The edges of the graph are essentially the union of the interpretations of graph predicates, but with the database totally connected.
By a \textit{cycle} in $\gA$, we mean a finite sequence of distinct elements $\ba = a_0 \cdots a_{\ell-1}$, with $\ell \geq 3$, such that, writing $a_\ell = a_0$,
 $(a_i,a_{i+1}) \in E$ for all $i$ ($0 \leq i < \ell$). The {\em length} of $\ba$ is $\ell$.
A cycle in $\gA$ is said to be \textit{active} if all its elements are active
(i.e.~are interpretations of constants), \textit{passive} if all its elements are
passive, and {\em mixed} if it contains both passive and active elements. We call $\gA$ $\ell$-\textit{quasi-acyclic} if
all passive or mixed cycles in $\gA$ have length greater than $\ell$.  

The following lemma is a kind of pumping lemma for interpretations for the logic $\GCD$. It 
states that, if a $\GCD$-formula $\phi$ is satisfied in a (finite) structure $\gA$, then $\phi$ is also satisfied in a (finite) structure $\gB$ containing no short passive or mixed cycles in $\gB$. 
\begin{lemC}[{\cite[Lemma 13]{pratt2009data}}]
	\label{lem:big_cycles}
	Let $\phi$ be a $\GCD$-formula.
	Suppose $\gA \models \phi$, with $D \subseteq A$ 
	the set of active elements, and let $\ell > 0$. Then there exists an $\ell$-quasi-acyclic  
	model $\gB \models \phi$ such that $D \subseteq B$ and $\gA|_{D} = \gB|_{D}$.
	Moreover, if $\gA$ is finite, then we can ensure that $\gB$ is finite.
\end{lemC}
\noindent 
We remark 
that the proof of Lemma~\ref{lem:big_cycles} makes essential use of guardedness. 

Having discussed the graphs defined by $\GCD$-structures, we turn now to configurations of small collections of elements
of those structures. Let $\gA$ be a structure interpreting $\sigma$ over domain $A$, let $\ba = a_0 \cdots a_{\ell}$ ($\ell \geq 0$)
be a word over the alphabet $A$, and let
$\br = r_0, \dots r_{\ell-1}$ be a word over the alphabet of graph predicates of $\sigma$.
The pair $[\ba,\br]$ is a {\em walk} of \textit{length} $\ell$ if, for all $i$ ($0 \leq i < \ell$) $\gA \models r_i[a_i,a_{i+1}]$. Alternatively,
we say that $\ba$ is an $\br$-{\em walk}; and where $\br$ is clear from context,
we speak of {\em the walk} $\ba$. There is no requirement that the $a_i$ all be distinct, though, of course, the irreflexivity of $r_i$ means that $a_i \neq a_{i+1}$ ($0 \leq i < \ell$); however, if the $a_i$ \textit{are} all distinct, then we speak of the {\em path} $\ba$.
Now let $\ba'= a_0, \dots, a_{\ell-1}$, with 
$\br = r_0, \dots r_{\ell-1}$ be as before. The pair $(\ba',\br)$ is a {\em tour} of \textit{length} $\ell$ if, for all $i$ ($0 \leq i < \ell$) $\gA \models r_i[a_i,a_{i+1}]$ where addition in indices is performed modulo $\ell$. 
Again, we speak of $\ba'$ as being an $\br$-\textit{tour}, etc. Intuitively, we identify
the tour $(\ba',\br)$ with the walk $[\ba'{}a_0,\br]$, starting and ending at $a_0$. Note that $(\epsilon,\epsilon)$ (corresponding to $\ell=0$) is a tour, namely, the empty tour
starting (and ending) at any element $a$. 
A path $\bar{a}$ is said to be {\em active} if all its elements are active, {\em passive} if all its elements are passive, and {\em mixed} otherwise; similarly for (non-empty) tours. In the case of the empty tour,
we count it as being active if we think of it as beginning at an active element, and passive otherwise: 
this slight informality should cause no confusion in practice. 
Notice that a tour of length at least 3 in which all elements are distinct is a cycle in $\gA$.

Consider now any walk $[\ba,\baf]$ in $\gA$, where
$\ba = a_0 \cdots a_\ell$. Noting that the $a_i$ are not necessarily distinct, let $V = \{a_i \mid 0 \leq i \leq  \ell\}$, and let us say that
$a$ and $b$ in $V$ are \textit{neighbours} if, for some $i$ ($0 \leq i < \ell$), either $a= a_{i}$ and $b= a_{i+1}$, or 
$b= a_{i}$ and $a= a_{i+1}$. Let $E$ be the set of unordered pairs $(a,b)$
from $V$ such that $a$ and $b$ are neighbours. Then $G = (V,E)$ 
is a graph, which we refer to as the {\em locus} of the walk $[\ba,\baf]$. 
If $(\ba,\baf)$ is a non-empty tour in $\gA$ beginning at $a_0$, then we take the locus of $(\ba,\baf)$ to be the locus of the walk $[\ba a_0,\baf]$.
Thus loci are static records of all the steps taken during a walk or tour, but with information about the order of those steps suppressed. Fig.~\ref{fig:locus} shows a possible locus of a tour 
$(a_0a_1a_2a_3a_4a_3a_5a_3a_4a_3a_2a_1,f_0\cdots f_{11})$ of length 12; the element $a_3$ is encountered four times in this tour. It is easy to see that the locus of a walk (or tour) in $\gA$ is a subgraph of the graph of $\gA$, though of course it will in general not
be an induced subgraph. We call a walk (or tour) \textit{acyclic} if its locus contains no cycles. The tour whose locus is depicted in Fig.~\ref{fig:locus} is acyclic.
It should be obvious that, if $\gA$ is an 
$\ell$-quasi-acyclic structure, and $(\ba,\baf)$ a passive tour in $\gA$ of length at most $\ell$, then 
$(\ba,\baf)$ must be acyclic. 
\ignore{Moreover, bearing in mind that every active element has an edge to every other in the graph of $\gA$, if $[a\ba b,\baf]$ is a walk with $a$ and $b$ 
active elements and $\ba$ a non-empty sequence of at most $\ell-2$ passive elements, then $[a\ba b,\baf]$ is acyclic 
and $a = b$.}

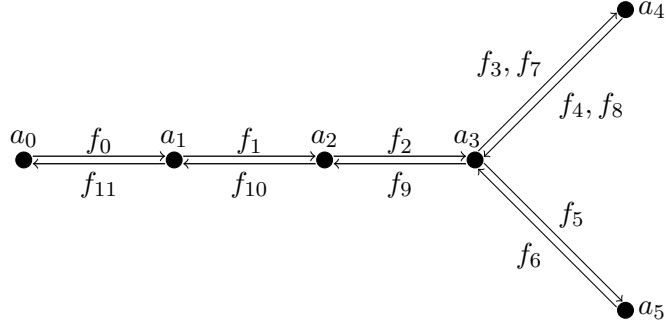
\begin{figure}[ht]
\centering
\begin{tikzpicture}

\vertex(0, 2);
\vertex(2, 2);
\vertex(4, 2);

\node[above=1pt] at (0,2) {$a_0$};
\node[above=1pt] at (2,2) {$a_1$};
\node[above=1pt] at (4,2) {$a_2$};
\node[above=1pt] at (5.9,2) {$a_3$};
\node[right=1pt] at (8,4) {$a_4$};
\node[right=1pt] at (8,0) {$a_5$};

\doubleArrow[(0, 2), (2, 2), $f_0$, $f_{11}$];	
\doubleArrow[(2, 2), (4, 2), $f_1$, $f_{10}$];	
\doubleArrow[(4, 2), (6, 2), $f_2$, $f_{9}$];

\tikzset{
  vertex/.style={circle, draw=black, fill=black, inner sep=0pt, minimum size=6pt}
}

\node[vertex] (A) at (6, 2) {};
\node[vertex] (B) at (8, 4) {};
\node[vertex] (C) at (8, 0) {};

\pgfmathsetmacro{\shift}{.23ex}

\path[->, transform canvas={xshift=-\shift,yshift=\shift}] (A) edge node[midway, above left=-2pt] {$f_3,f_7$} (B);
\path[->, transform canvas={xshift=\shift,yshift=-\shift}] (B) edge node[midway, below right=-2pt] {$f_4,f_8$} (A);

\path[->, transform canvas={xshift=\shift,yshift=\shift}] (A) edge node[midway, above right=-2pt] {$f_5$} (C);
\path[->, transform canvas={xshift=-\shift,yshift=-\shift}] (C) edge node[midway, below left=-2pt] {$f_6$} (A);
\end{tikzpicture}
\caption{A tour $(a_0a_1a_2a_3a_4a_3a_5a_3a_4a_3a_2a_1,f_0\cdots f_{11})$ and its locus.}
\label{fig:locus}
\end{figure}
\subsection{Path-functional dependencies and their violations}
With these preliminaries behind us, we turn our attention to the analysis of path-functional dependencies in particular structures. We start, for simplicity, with the unary case.
Our goal will be to encode a given unary path-functional dependency as a certain $\GC$
formula.
Let $\kappa$ be a unary path-functional dependency $\PFD[\baf]$, where
$\baf = f_0 \cdots f_{k-1}$.
If $0 \leq i \leq k$, denote 
the prefix $f_0 \cdots f_{i-1}$ by $\baf_i$, and let $\kappa_i = \PFD[\baf_i]$; we 
call $\kappa_i$ a \textit{prefix} of $\kappa$. 
It is easily seen that $\kappa$ entails each of its prefixes---i.e., $\gA \models \kappa \rightarrow \kappa_i$.
Moreover, the
empty unary path-functional dependency, $\PFD[\epsilon]$, is trivially valid.
A set $K$ of unary path-functional dependencies is \textit{prefix-closed}, if for any $\kappa \in K$, every prefix of $\kappa$ is in $K$.  Given any set $K$ of unary path-functional dependencies, we may---without affecting satisfiability---ensure that it is prefix closed by adding at most linearly many additional unary path-functional dependencies.

Suppose now that $\kappa$ is violated in some $\sigma$-structure
$\gA$, and let $a, b \in A$ be a violating pair for $\kappa$, i.e.~$a \neq b$ but the $\baf$-walks starting at $a$ and $b$ 
end in the same element.
Define the sequence of elements
$\ba = a_0 \cdots a_k$ where $a_0 = a$ and, for all $i$ ($0 \leq i < k$), $\gA \models f_i(a_i,a_{i+1})$;
and define $\bb = b_0 \cdots b_k$ where $b_0 = b$ and, for all $i$ ($0 \leq i < k$), $\gA \models f_i(b_i,b_{i+1})$. Noting that $a_k = b_k$,
let $i$ be the smallest number $0 \leq i \leq k$ such that $a_i = b_i$, i.e. the point at which $\ba$ and $\bb$ converge. Clearly $i > 0$, since $a$
and $b$ are distinct. 
We say that the violation of $\kappa$ at $a, b$ is \emph{critical} if $i = k$, i.e.~if $\ba$ and $\bb$ converge only at their last element. It is then obvious that, if $K$ is prefix closed, and some $\kappa \in K$ is violated, then some
$\kappa' \in K$ is critically violated. 

Checking for critical violations of unary path-functional dependencies is simple. Bearing in mind that
$\PFD[\epsilon]$ is trivial valid, suppose
$\kappa= \PFD[\baf f]$.
For any prefix $\baf'$ of $\baf$ (proper or improper), let $p_{\baf'}$ be a fresh unary predicate, and let $P_\kappa$ be the  set $\GC$-formulas
\begin{equation*}
\big\{ \forall x. p_\epsilon (x) \} \cup
\big\{\forall x 
   \big( p_{\baf' \! f'}(x) \leftrightarrow 
             \exists y (p_{\baf'} (y) \land f'(y, x)) \big)
\mid \text{$\baf' \! f'$ a prefix of $\baf$} \}. 
\end{equation*}
In the presence of $P_\kappa$, we may read 
$p_{\baf'}(x)$ as stating that $x$ is at the end of an $\baf'$-walk in $\gA$.
Thus, in a model $\gA \models P_\kappa$, $\kappa= \PFD[\baf f]$ has a critical violation just in case
the formula 
\begin{align*}
\psi_\kappa(y):= \exists x \exists z ( \dst{x, y, z}  \land 
(p_{\baf}(x) \land f(x,y)) \land (p_{\baf}(z) \land f(z,y)).
\end{align*}
is satisfied in $\gA$. (In particular, if $\kappa$ has a critical violation at $a$, $b$, $\psi_\kappa$ will be satisfied at the 
(common) final elements of the $\baf$-walks starting at $a$ and $b$.)
By Lemma~\ref{lem:rewrite}, $\psi_\kappa(y)$ may be equivalently written as a 
$\GC$-formula $\psi_\kappa^*(y)$. Thus, if $K$ is a prefix-closed set of unary path-functional dependencies, 
a $\GCDK$-formula $\phi \wedge K \wedge \Delta$ is satisfiable over some domain $A$ if and only if
the $\GCD$-formula
\begin{align*}
\phi \wedge \bigwedge \bigcup \{P_\kappa \mid \kappa \in K\}  \wedge \bigwedge_{\kappa \in K} \forall y \, \neg \psi^*_\kappa(y) \wedge \Delta
\end{align*}
is satisfiable over the same domain. This provides a polynomial time reduction from the (finite) satisfiability problem for $\GCDK$ restricted to {\em unary}
path-functional dependencies to the corresponding problem for $\GCD$. Since we show in the sequel that the latter problem is in $\dexp$, so is the former.

We now turn our attention to the more difficult case of binary path-functional dependencies.
Let $\kappa$ be a binary path-functional dependency $\PFD[\baf, \bag]$, where
$\baf = f_0 \cdots f_{k-1}$.
For all $i$ ($0 \leq i < k$), denote by $\baf_i$, 
the prefix $f_1 \cdots f_{i-1}$ of $\baf$, and let $\kappa_i = \PFD[\baf_i, \bag]$. 
We call $\kappa_i$ a \textit{left-prefix} of $\kappa$.
Thus $\kappa$ entails each of its left-prefixes; moreover, $\PFD[\epsilon, \bag]$ is trivially valid. 
If $K$ is a set of binary path-functional dependencies, say that $K$ is \textit{left-prefix-closed}, if for $\kappa \in K$, any left-prefix of $\kappa$ is in $K$. Any set $K$ of binary path-functional dependencies may---without affecting satisfiability---be made left-prefix-closed by adding at most linearly many additional binary path-functional dependencies.
We could instead have spoken of right-prefix-closed sets of functional dependencies, defined in the obvious way; the choice between these notions is completely arbitrary.

Suppose now that $\kappa = \PFD[\baf,\bag]$ is violated in some $\sigma$-structure
$\gA$, and let $a, b \in A$ be a violating pair for $\PFD[\baf, \bag]$, i.e.~$a \neq b$ but the 
$\baf$-walks starting at $a$ and $b$ 
 end in the same element, and moreover 
the $\bag$-walks starting at $a$ and $b$ end in the same element. Writing $\baf = f_0 \cdots f_{k-1}$, 
define the sequence of elements
$\ba = a_0 \cdots a_k$ where $a_0 = a$ and, for all $i$ ($0 \leq i < k$), $\gA \models f_i(a_i,a_{i+1})$.  Similarly
define $\bb = b_0 \cdots b_k$ where $b_0 = b$ and, for all $i$ ($0 \leq i < k$), $\gA \models f_i(b_i,b_{i+1})$. We
call the violation of $\kappa$ at $a, b$ \emph{critical} if $a_i \neq b_i$ for all $i$ ($0 \leq i < k$).
It is again obvious that, if $K$ is a left-prefix-closed set of binary path-functional dependencies, and some $\kappa \in K$ is violated, then some $\kappa' \in K$ is critically violated. Thus, as with
unary path-functional dependencies, so with their binary counterparts, we may confine our attention to critical violations.

The difficulty is that critical violations of binary path-functional dependencies cannot be straightforwardly expressed 
using $\GC$-formulas as in the unary case. To understand the problem, consider a binary path-functional dependency 
$\kappa = \PFD[\baf f,\bag]$, which has a
critical violation at $a, b$. Writing $\baf = f_0 \cdots f_{k-1}$ and $f_k = f$, define
$\ba = a_0 \cdots a_{k+1}$ where $a_0 = a$ and, for all $i$ ($0 \leq i \leq k$), $\gA \models f_i(a_i,a_{i+1})$, and
define $\bb = b_0 \cdots b_{k+1}$ where $b_0 = b$ and, for all $i$ ($0 \leq i \leq k$), $\gA \models f_i(b_i,b_{i+1})$. 
Thus, $a_{k+1} = b_{k+1}$.
Writing $\bag = g_0 \cdots g_{\ell-1}$,
define $\ba' = a'_0 \cdots a'_{\ell}$ where $a'_0 = a$ and, for all $i$ ($0 \leq i < \ell$), 
$\gA \models g_i(a'_i,a'_{i+1})$; and define $\bb'$ similarly, but starting with $b'_0 = b$. 
Thus, $a'_{\ell} = b'_{\ell}$. Referring to Fig.~\ref{fig:violation_critical}, it follows that
\begin{align*}
& a_0 a_1 \cdots a_{k} a_{k+1} b_{k} \cdots b_1 b_0 b'_1 \cdots b'_{\ell} a'_{\ell-1} \cdots a'_1
\end{align*}
is an $\baf f f^{-1} \baf^{-1} \bag \bag^{-1}$-tour in $\gA$. The problem is how to characterize the existence of such tours with only $\GC$-formulas and conditions on the database at our disposal.

Let us consider the tour of Fig.~\ref{fig:violation_critical} more closely.
Since key predicates are by assumption irreflexive, we
know that neighbouring elements in this tour are distinct. Furthermore, since the violation of $\kappa$ is by assumption {critical},
we also know that $a_k \neq b_k$. That is, the elements $a_{k}$, $a_{k+1}$ and $b_{k}$ are all distinct. 
\begin{figure}
	\centering
	\begin{tikzpicture}
	
	\vertex(0, 0);
	\vertex(0, 3);
	\vertex(3, 0);
	\vertex(3, 3);
	
	\vertex(2, 3);
	\vertex(3, 2);
	
	\draw[line width=.5pt, ->] (0, 2.85) -- (0, .15);
	\draw[line width=.5pt, ->] (.15, 3) -- (1.85, 3);
	\draw[line width=.5pt, ->] (2.85, 0) -- (.15, 0);
	\draw[line width=.5pt, ->] (3, 0.15) -- (3, 1.85);
	\draw[line width=.5pt, ->] (2.15, 3) -- (2.85, 3);
	\draw[line width=.5pt, ->] (3, 2.15) -- (3, 2.85);
	
	\node[above left=1pt] at (0, 3) {$a= a_0$};
	\node[below right] at (3, 0) {$b = b_0$};
	\node[above=1pt] at (2, 3) {};
	\node[below=1pt] at (2, 3) {$c= a_k$};
	\node[above right] at (3, 3) {$d= a_{k+1} = b_{k+1}$};
	\node[left=1pt] at (3, 2) {$e= b_k$};
	\node[right=1pt] at (3, 2) {};
	
	\node[above] at (1, 3) {$\baf$};
	\node[below] at (1.5, 0) {$\bag$};
	\node[left]  at (0, 1.5) {$\bag$};
	\node[right] at (3, 1) {$\baf$};
	\node[above] at (2.5, 3) {$f$};
	\node[right=1pt] at (3, 2.5) {$f$};
	
	\end{tikzpicture}
	
	\caption{A critical violating tour of $\PFD[\baf f, \bag]$. The sequences $a_0 \cdots a_k$
		and $b_0 \cdots b_k$ are
		$\baf$-walks; the elements
		$c$, $d$ and $e$ are \emph{distinct}.}
	\label{fig:violation_critical}
\end{figure}
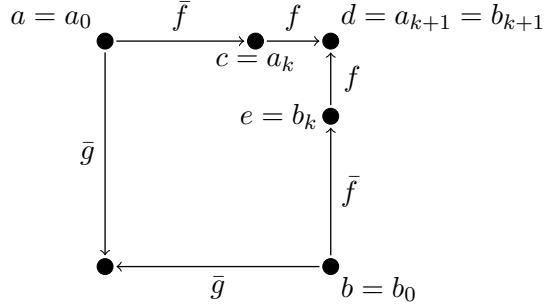
\ignore{\begin{figure}[ht]
\centering
\begin{tikzpicture}

\vertex(0, 0);
\vertex(0, 3);
\vertex(3, 0);
\vertex(3, 3);
\vertex(4.25, 4.25);
\vertex(-1.25, -1.25);

\draw[line width=.5pt, ->] (0, 2.85) -- (0, .15);
\draw[line width=.5pt, ->] (.15, 3) -- (2.85, 3);
\draw[line width=.5pt, ->] (2.85, 0) -- (.15, 0);
\draw[line width=.5pt, ->] (3, 0.15) -- (3, 2.85);
\draw[line width=.5pt, ->] (3.1, 3.1) -- (4.15, 4.15);
\draw[line width=.5pt, ->] (-.1, -.1) -- (-1.15, -1.15);

\node[above left] at (0, 3) {$a$};
\node[below right] at (3, 0) {$b$};

\node[above] at (1.5, 3) {$\baf'$};
\node[below] at (1.5, 0) {$\bag'$};
\node[left]  at (0, 1.5) {$\bag'$};
\node[right] at (3, 1.5) {$\baf'$};

\node[below=2pt] at (3.75, 3.75) {$\baf''$};
\node[above=1pt] at (-.75, -.75) {$\bag''$};

\end{tikzpicture}
\caption{A violating pair $a, b$ for a path-functional dependency
$\PFD[\baf, \bag]$. Note that the arrows denote sequences
of functional predicates and not single (directed) edges. In that context, $\baf = \baf' \! \baf''$ 
and $\bag = \bag' \bag''$. Note, also, that $\baf''$ and $\bag''$ can be empty. The pair $a, b$ is
critical just in case $\baf'' = \epsilon$.}
\label{fig:violation}
\end{figure}}
We can make this observation---which is fundamental to the
entire development of Section~\ref{sec:decompositions}, \ref{sec:violations} and~\ref{sec:tour_predicates}---work for us if we rotate the tour
so that it starts at $a_{k}$.
Write $c = a_{k}$, $d = a_{k+1} = b_{k+1}$ and $e = b_{k}$, so that
$\gA \models f(c,d)$ and $\gA \models f(e,d)$, with
$c$, $d$ and $e$ all distinct. This yields a tour
\begin{equation*}
(cde\bar{e}, f f^{-1} \baf^{-1} \bag \bag^{-1}\baf f)
\end{equation*}
for some sequence of elements $\bar{e}$. For brevity,
we shall always write
$\seq$ in the sequel to denote the `rotated' sequence of graph predicates 
\begin{equation*}
\seq =  f f^{-1} \baf^{-1} \bag \bag^{-1}\baf f
\end{equation*}
obtained from the binary path-functional dependency $\PFD[\baf f, \bag]$.
Thus, $\kappa$ is critically violated in a structure $\gA$ if and only if $\gA$ contains
an $\seq$-tour whose first three elements are distinct. 

Note that the diagram of Fig.~\ref{fig:violation_critical} depicts an $\seq$-tour, and not its locus. In 
particular, there is no assumption that the sequences $\ba$, $\bb$, $\ba'$ and $\bb'$ do not contain repeated elements, and
no assumption that they are disjoint. Indeed, we recall Lemma~\ref{lem:big_cycles}, which allows us to restrict attention to
$\ell$-quasi-acyclic structures for various $\ell$. If $\ell > |\seq|$, then all violations of $\kappa$ 
involving only passive elements will yield acyclic $\seq$-tours. This forms our main line of attack: Sections~\ref{sec:decompositions} and \ref{sec:violations} are concerned
with the classification of $\seq$-tours; 
Section~\ref{sec:tour_predicates} uses this classification to encode the non-existence of critical violations
using $\GC$-formulas and conditions on the database; finally, Section~\ref{sec:gcdk} assembles all these observations to
yield a reduction of the (finite) satisfiability problem for $\GCDK$ to the corresponding problem for $\GCD$.

\section{Decompositions of Walks and Tours}
\label{sec:decompositions}
We have seen that critical violations of binary path-functional dependencies in a structure
correspond to the existence of certain tours in the graph of that structure.
This section presents the basic tools we require in the sequel for decomposing walks and tours in structures.
None of the reasoning involved goes beyond elementary graph theory. Lemmas~\ref{lma:spine}--\ref{lma:retrace}
concern acyclic walks and tours;  Lemmas~\ref{lem:simple_db1}--\ref{lem:isthmus} concern tours in $\ell$-quasi-acyclic structures.

\begin{lem}
Let $\gA$ be a structure and $[\ba, \bt]$ be an acyclic
walk in $\gA$. 
Then, for some $m \geq1$, the sequences $\ba$ and $\bt$ can be decomposed 
\begin{align*}
\ba \: &= \: \bb_0 b_0 \cdots \bb_{m-1} b_{m-1} \bb_{m} b_m\\
\bt \: &= \: \br_0 r_0 \cdots  \br_{m-1} r_{m-1} \br_{m}
\end{align*}
such that: (i) the sequences $\bb_1, \dots, \bb_m$ are pairwise disjoint; (ii)
for all $j$ ($0 \leq j \leq m$), 
$(\bb_j, \br_j)$ is a tour starting at $b_j$; and
(iii) $[b_0 \cdots b_k, r_0 \cdots r_{m-1}]$ is a path in $\gA$ (Fig.~\ref{fig:backbone}).
\label{lma:spine}
\end{lem}
\begin{proof}
Let $\ba = a_0 \cdots a_\ell$ ($\ell > 0$).
Let $\iota(0) = 0$ and define $b_0 = a_{\iota(0)} = a_0$. Suppose now that $\iota(i)$ and $b_i$ have been defined,
with $i \leq \ell$.
If $b_i \neq a_\ell$, then let $\iota(i+1)$ be the largest number $j$ ($\iota(i) < j \leq \ell$) such that
$a_{j-1} = b_{i}$, and define $b_{i+1} = a_{\iota(i+1)}$ and $\bb_{i} = a_{\iota(i)} \cdots a_{\iota(i+1)-2}$. 
Likewise, define $r_{i} = t_{\iota(i+1)-1}$ and $\br_{i} = t_{\iota(i)} \cdots t_{\iota(i+1)-2}$.
Then $\gA \models r_i(b_i, b_{i+1})$ and  $(\bb_{i},\br_i)$ is a possibly empty acyclic tour starting at $b_i$. 
When, eventually, $b_{i} = a_\ell$, define $m = i$, $\bb_{m} = a_{\iota(i)} \cdots a_{\ell-1}$
and $\br_{m} = f_{\iota(i)} \cdots f_{\ell-1}$. The disjointness of the
$\bb_i$ is immediate.
\end{proof}

We call the path $[b_0 \cdots b_{k},r_0 \cdots r_{k-1}]$ in the decomposition of $[\ba, \bt]$ given by Lemma~\ref{lma:spine} the \emph{spine} of $[\ba, \bt]$. Where the sequence
$\bt$ is of no interest, we simply say that $b_0 \ldots b_{k}$ is the \textit{spine} of $\ba$.

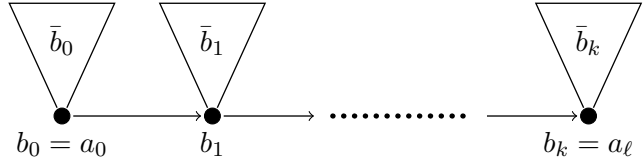
\begin{figure}

\centering
\begin{tikzpicture}

\vertex(0, 0);
\vertex(2, 0);
\vertex(7, 0);

\triangleUp(0, 0);
\triangleUp(2, 0);
\triangleUp(7, 0);

\draw[dotted, line width=2pt, line cap=round, dash pattern=on 0pt off 2\pgflinewidth] (3.6, 0) -- (5.4, 0);

\draw[line width=.5pt, ->] (.15, 0) -- (1.85, 0);
\draw[line width=.5pt, ->] (2.15, 0) -- (3.35, 0);
\draw[line width=.5pt, ->] (5.65, 0) -- (6.85, 0);

\node[below=2pt] at (0, 0) {$b_0 = a_0$};
\node[below=2pt] at (2, 0) {$b_1$};
\node[below=2pt] at (7, 0) {$b_k = a_\ell$};

\node at (0, 1) {$\bb_0$};
\node at (2, 1) {$\bb_1$};
\node at (7, 1) {$\bb_k$};

\node[above=-2pt] at (1, 0) {~};
\node[above=-2pt] at (2.8, 0) {~};
\node[above=-2pt] at (6.3, 0) {~};

\end{tikzpicture}
\caption{A decomposition of a walk $a_0 \cdots a_\ell$ into its spine
$b_0 \cdots b_k$, with a (possibly empty)  acyclic subtour at each $b_i$, $0 \leq i \leq k$.}
\label{fig:backbone}
\end{figure}

\begin{lem}
Let $\gA$ be a structure and $\ba = a_0, \ldots, a_{\ell - 1}$ a non-empty acyclic tour in $\gA$. Let $a$ be any element of $\ba$ not equal to $a_0$, let $i$, $j$
($0 < i \leq j < \ell$) be the smallest and the largest index
respectively such that $a_i = a_j = a$, and let 
$a_\ell = a_0$. Then $a_{j+1} = a_{i-1}$.
\label{lma:retrace}
\end{lem}

\begin{proof}
Observe first that the words $a_0, \ldots, a_{i - 1}$ and $a_{j+1}, \ldots, a_{\ell}$ are non-empty, since
$a_i = a_j \neq a_0 = a_\ell$; thus, 
these words define walks in $\gA$.
Let $\bb$ be the spine of the walk $a_0, \ldots, a_{i - 1}$, and let
$\bb'$ be the spine of the walk $a_{j+1}, \ldots, a_{\ell}$. 
Thus,
$\bb$ is a path from $a_0$ to $a_{i-1}$ and 
$\bb'$ a path from $a_{j+1}$ to $a_\ell$. 
Let $\bb''c$ be the shortest prefix of $\bb'$ that intersects $\bb$,
and let $\bb'''$ be the suffix of $\bb$ beginning with the element $c$.
Then, if $a_{i-1} \neq a_{j+1}$, $\bb''' a_i \bb''$ is a cycle, contrary to hypothesis.
(See Fig.~\ref{fig:simple-helper} for an illustration.)
\end{proof}

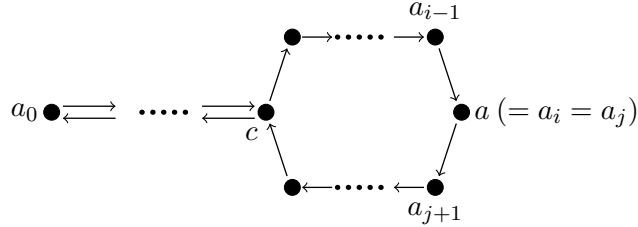
\begin{figure}
\centering
\begin{tikzpicture}

\vertex(0, 0);
\node[left=1pt] at (0, 0) {$a_0$};
\draw[line width=.5pt, ->] (.15, .08) -- (.85, .08);
\draw[line width=.5pt, <-] (.15, -.08) -- (.85, -.08);
\draw[dotted, line width=2pt, line cap=round, dash pattern=on 0pt off 2\pgflinewidth] (1.2, 0) -- (1.8, 0);
\draw[line width=.5pt, ->] (2, .08) -- (2.7, .08);
\draw[line width=.5pt, <-] (2, -.08) -- (2.7, -.08);
\vertex(2.85, 0); \node[below=2pt] at (2.85, 0) {$c \quad$};

\vertex(5.45, 0);
\node[right=1pt] at (5.45, 0) {$a \: (= a_i = a_j)$};

\vertex(5.1, -1);
\node[below=2pt] at (5.1, -1) {$a_{j+1}$};
\draw[line width=.5pt, ->] (4.95, -1) -- (4.55, -1);
\vertex(3.2, -1);
\draw[line width=.5pt, ->] (3.73, -1) -- (3.33, -1);
\draw[dotted, line width=2pt, line cap=round, dash pattern=on 0pt off 2\pgflinewidth] (4.38, -1) -- (3.76, -1);
\draw[line width=.5pt, ->] (5.38, -.15) -- (5.15, -.85);
\draw[line width=.5pt, ->] (3.13, -.85) -- (2.9, -.15);

\vertex(5.1, 1);
\node[above=2pt] at (5.1, 1) {$a_{i-1}$};
\draw[line width=.5pt, <-] (4.95, 1) -- (4.55, 1);
\vertex(3.2, 1);
\draw[line width=.5pt, <-] (3.73, 1) -- (3.33, 1);
\draw[dotted, line width=2pt, line cap=round, dash pattern=on 0pt off 2\pgflinewidth] (4.38, 1) -- (3.76, 1);
\draw[line width=.5pt, <-] (5.38, .15) -- (5.15, .85);
\draw[line width=.5pt, <-] (3.13, .85) -- (2.9, .15);

\end{tikzpicture}
\caption{A cycle when $a_{i-1} \neq a_{j+1}$, as described in the proof of Lemma~\ref{lma:retrace}; $\bb'''$ is
the path from $c$ to $a_{i-1}$ and $\bb''c$ is the path from $a_{j+1}$ to $c$.}
\label{fig:simple-helper}
\end{figure}

The following lemmas concern tours in $\ell$-quasi-acyclic structures. We start with the observation that, if a tour of length $\ell$ in such a structure exits the database at some 
point, then it must re-enter at the same point. Recall that a mixed tour is one featuring both active and passive elements; and remember that, in the graph of a structure $\gA$, the database is totally connected. 
\begin{lem}
\label{lem:simple_db1}
Let $\gA$ be an $\ell$-quasi-acyclic structure and $\ba = a_0 \cdots a_{\ell - 1}$
a tour in $\gA$. Suppose $i$ ($0 \leq i < \ell$) is such that $a_0 \cdots a_i$ are all active, but $a_{i+1}$ passive. 
Writing $a_\ell = a_0$, let
$j$ be the smallest index $j$ ($i < j < \ell$) such that $a_{j+1}$ is active. Then $a_{i+1} = a_{j}$
and $a_{j+1} = a_{i}$. 
\end{lem}
\begin{proof}
By construction, $a_{i+1} \cdots a_j$ is a passive walk from $a_{i+1}$ to $a_j$.
Let the spine of this walk be $\bb$; note that the sequence $\bb$ is non-empty, and its 
initial and final elements are, respectively, $a_{i+1}$ and
$a_j$.
Suppose first that $a_{i} \neq a_{j+1}$. Since these elements are active, they are joined by an 
edge in the graph of $\gA$, whence $a_i \bb a_{j+1}$ is a mixed cycle in $\gA$ of length (at least 3 and) at most $\ell$. 
Since $\gA$ is $\ell$-quasi-acyclic this is a contradiction, whence $a_{i} = a_{j+1}$.
Suppose, now that $a_{i+1} \neq a_{j}$, whence $|\bb| \geq 2$. Then, since $a_{i} = a_{j+1}$, it follows that $a_{i} \bb$ is a 
mixed cycle of length (at least 3 and) at most $\ell$. Again, this is a contradiction, whence $a_{i+1} = a_{j}$.
\end{proof}
The next lemma is a similar to the last, but with `enter' and `exit' transposed. The reasoning is essentially 
identical.
\begin{lem}
\label{lem:simple_db2}
Let $\gA$ be an $\ell$-quasi-acyclic structure and $\ba = a_0 \cdots a_{\ell - 1}$
a tour in $\gA$. Suppose $i$ ($0 \leq i < \ell$) is such that $a_0 \cdots a_i$ are all passive, but $a_{i+1}$ active. 
Writing $a_\ell = a_0$, let
$j$ be the largest index $j$ ($i < j < \ell$) such that $a_j$ is active and 
$a_{j+1}$ is passive. Then $a_{i+1} = a_{j}$
and $a_{j+1} = a_{i}$. 
\end{lem}

The final two lemmas of this section will play an important role in
Sec.~\ref{sec:gcdk}.

\begin{lem}
\label{lem:isthmus}
Let $\gA$ be an $\ell$-quasi-acyclic structure, and
$(\ba, \bt)$ a mixed tour in $\gA$ of length $\ell$, beginning with some passive element $a$. Then $(\ba,\bt)$ can be decomposed as $(\bc \ba^* \bb, \bs \bt^* \br)$
such that, for some active element $b$:
(i) $[\bb a; \br]$ is an acyclic walk from $b$ to $a$ in which $b$ occurs exactly once; (ii) $[\bc b; \bs]$ is an acyclic walk from $a$ to $b$ in which $b$ occurs exactly once; and
(iii) $(\ba^*; \bt^*)$ is a tour beginning at $b$. 
\end{lem}
\begin{proof}
Write $\ba = a_0 \cdots a_{\ell -1}$,
let $b = a_i$ be the first active element of $\ba$, and let $j$ ($i \leq j < \ell$) 
be the greatest index such that $a_j = b$. Since $a$ is passive, $i > 0$.
Write $\bc = a_0 \cdots a_{i-1}$, $\ba^* = a_{i} \cdots a_{j-1}$ and
$\bb = a_j \cdots a_{\ell-1}$. Let $\bs$, $\bt^*$ and $\br$ be the corresponding decomposition of $\bt$. Since $(\ba, \bt)$ is a mixed tour, it is acyclic; so therefore are
the walks $[\bb a; \br]$ and $[\bc b; \bs]$ which it contains.
\end{proof}

\begin{lem}
\label{lem:dbcycle_decomp}
Suppose $\gA$ is an $\ell$-quasi-acyclic structure containing a tour, $(\ba^*;\bt^*)$, with $|t^*| = \ell$, and suppose the initial element of $\ba^*$ is active.
Then there exist decompositions 
$\ba^* = \bb_0 b_0  \cdots \bb_{m-1} b_{m-1} \bb_m$ and 
$\bt^* = \br_0 r_0 \cdots \br_{m-1} r_{m-1} \br_m$,
such that: (i) $(b_0 \cdots b_{m-1}, r_0 \cdots r_{m-1})$ is an active tour; (ii)
each  
$(\bb_i; \br_i)$ is an acyclic tour beginning at $b_i$ ($1 \leq i < m$),
and $(\bb_m, \br_m)$ is an acyclic tour beginning at $b_0$ (Fig.~\ref{fig:dbcycle}).
\end{lem}

\begin{proof}
Write $\ba^* = a_0, \ldots, a_{\ell - 1}$, with $a_0$ active. As usual, let $a_\ell = a_0$.
Let $\iota(0) = 0$ and define $b_0 = a_{\iota(0)} = a_0$. Suppose that $\iota(i)$ and $b_i$ have been defined,
with $b_i$ active and $\iota(i) \leq \ell$.
Let $\iota(i+1)$ be the smallest number $j$ ($\iota(i) < j \leq \ell$) such that
$a_{j}$ is active and distinct from $b_i$, assuming first that such an element exists.
Note that, if $\iota(i+1) \neq \iota(i) +1$, then $a_{\iota(i)+1}$ is passive, whence, by Lemma~\ref{lem:simple_db1},
the next active element in the tour $\ba$ after $a_{\iota(i)+1}$ is $a_{\iota(i)} = b_i$. By applying this argument repeatedly, we
see that $a_{\iota(i+1)-1} = b_i$. Define $b_{i+1} = a_{\iota(i+1)}$ and $\bb_{i} = a_{\iota(i)} \cdots a_{\iota(i+1)-2}$. 
Likewise, define $r_{i} = t_{\iota(i+1)-1}$ and $\br_{i} = t_{\iota(i)} \cdots t_{\iota(i+1)-2}$. 
Thus $(\bb_{i},\br_{i})$ is a (possibly empty) mixed---hence acyclic---tour
starting at $b_i$, and $\gA \models r_i(b_i,b_{i+1})$.
Now suppose that there is no number $j$ ($\iota(i) < j \leq \ell$) such that
$a_{j}$ is active and distinct from $b_i$. Set $m = \iota(i)$, and
define $\bb_{m} = a_{\iota(i)} \cdots a_{\ell-1}$. 
Likewise, define $\br_{m} = t_{m} \cdots t_{\ell-1}$.
If $\bb_{m}$ is non-empty, it leaves the database, and, whenever it does so, 
it re-enters at the same point, by Lemma~\ref{lem:simple_db1}.
Hence $(\bb_{m},\br_m)$ is a (possibly empty) mixed---hence acyclic---tour
starting at $b_m$ whence $b_m = a_\ell = a_0$. 
\end{proof}

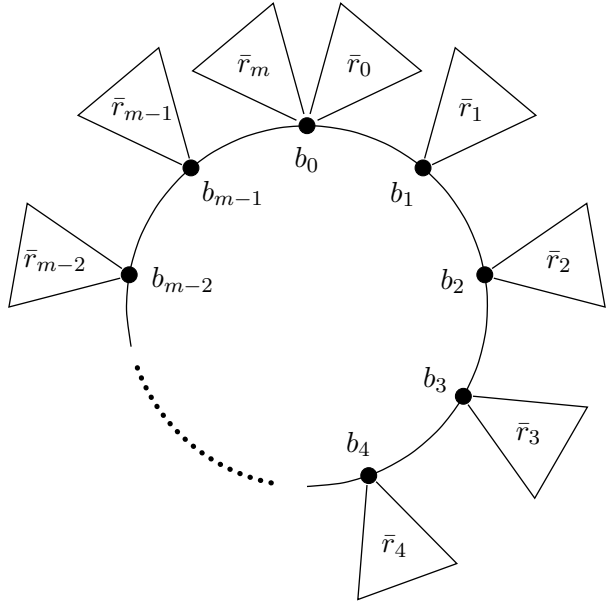
\begin{figure}[ht]
\centering
\begin{tikzpicture}

\def\r{2.4}

\foreach \i/\txt in {1/1,2/2,3/3,4/4,7/{m-2},8/{m-1}} {
  \pgfmathparse{\r*cos(90 - \i*40)} \let\x\pgfmathresult
  \pgfmathparse{\r*sin(90 - \i*40)} \let\y\pgfmathresult
  \vertex(\x, \y);
  \triangleRotate(\x, \y, \i*40);
  
  \pgfmathparse{(\r+1)*cos(90 - \i*40)} \let\xtxt\pgfmathresult
  \pgfmathparse{(\r+1)*sin(90 - \i*40)} \let\ytxt\pgfmathresult
  \node at (\xtxt, \ytxt) {$\br_{\txt}$};
  
  \pgfmathparse{(\r-.45)*cos(90 - \i*40)} \let\xb\pgfmathresult
  \pgfmathparse{(\r-.45)*sin(90 - \i*40)} \let\yb\pgfmathresult
  \node[right=-8pt] at (\xb, \yb) {$b_{\txt}$};  
}

\vertex(0, \r);
\triangleRotate(.02, \r + .03, 40);
\node at (.7, \r + .8) {$\br_0$};
\triangleRotate(-.02, \r + .03, -40);
\node at (-.7, \r + .8) {$\br_m$};
\node at (0, \r - .45) {$b_0$};

\draw[line width=.5pt, domain=-90:193] plot[smooth] ({\r*cos(\x)}, {\r*sin(\x)});
\draw[dotted, line width=2pt, line cap=round, dash pattern=on 0pt off 2\pgflinewidth, domain=200:260] plot[smooth] ({(\r*cos(\x)}, {\r*sin(\x)});

\end{tikzpicture}

\caption{Decomposition of an $\baf$-tour $(\ba^*, \bt^*)$ of length $\ell >0$, beginning in the database, in an
$\ell$-quasi-acyclic structure (Lemma~\ref{lem:dbcycle_decomp}).}
\label{fig:dbcycle}
\end{figure}

\section{Detection of Violations}
\label{sec:violations}

We now show how to detect violations of binary path-functional dependencies,
using the decompositions of tours developed in the previous section.
Recall from Section~\ref{sec:preliminaries} that a critical violation of a binary path-functional
dependency $\kappa$ in a structure $\gA$ is identified with an $\seq$-tour whose first three elements are distinct. Such a tour can be decomposed in various ways,
depending on which of these three elements belong to the active domain (i.e., the database). Altogether,
we isolate nine configurations of $\seq$-tours, illustrated in Fig.~\ref{fig:tours}, with the
circle representing a sub-tour in
the database (which is further decomposed as in Fig.~\ref{fig:dbcycle})
and triangles representing acyclic sub-tours. The task of this section is to establish that 
the nine configurations of Fig.~\ref{fig:tours} exhaust the possible $\seq$-tours. 
In Section~\ref{sec:tour_predicates}, we show how to use $\GC$-formulas together with 
conditions on the database to rule out each of these configurations, and thus to guarantee that $\kappa$
is not critically violated.
In the following sequence of lemmas, we fix $\kappa = \PFD[\baf f, \bag]$, where
$\bag \neq \epsilon$, 
as usual writing $\seq$ for $f f^{-1} \baf^{-1} \bag \bag^{-1} \baf$. Notice that, by construction,
$\sizeOf{\seq} \geq 4$.

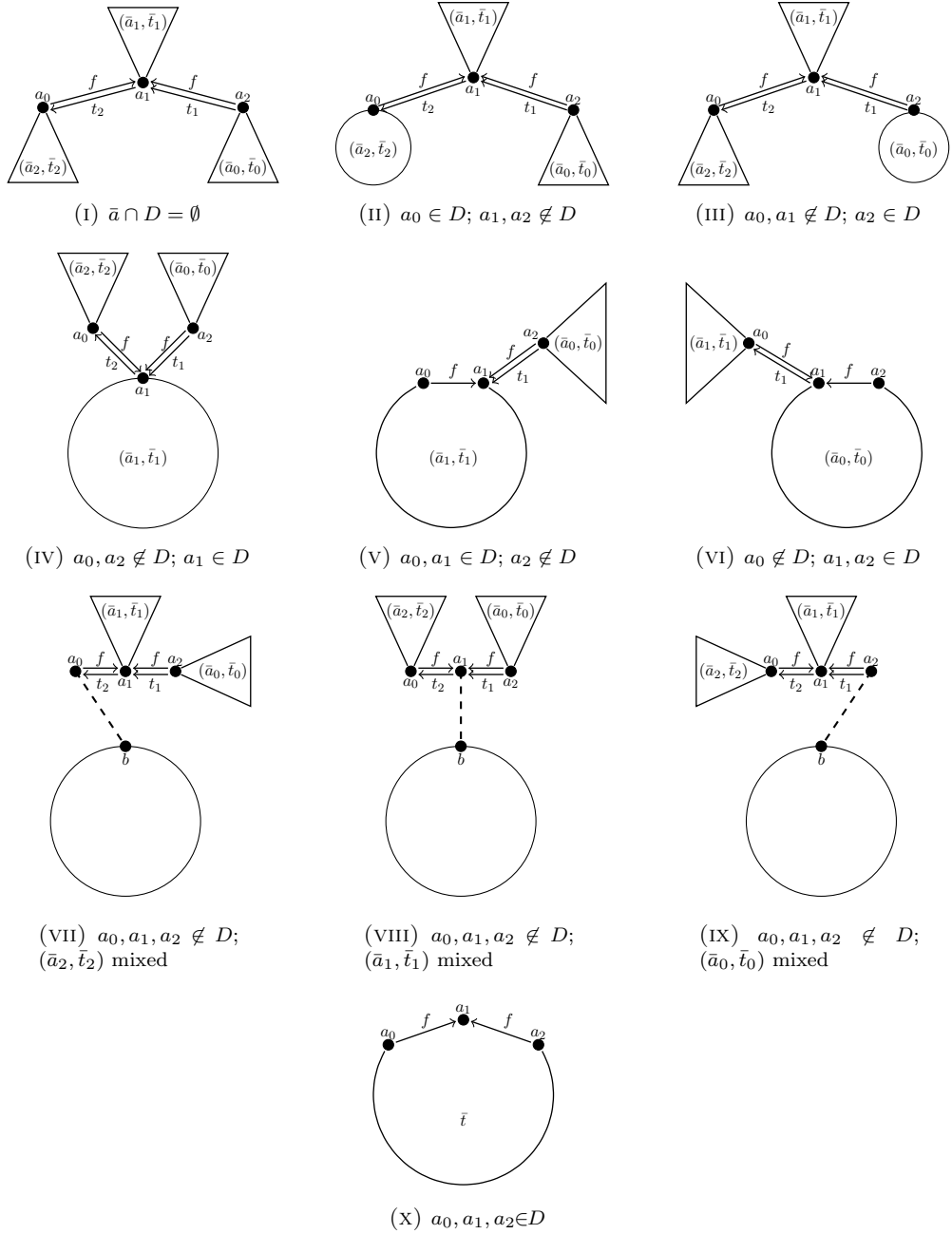
\begin{figure}
\centering
\begin{tabular}{c@{\hskip -20pt} c@{\hskip -20pt} c}

\subfloat[\scriptsize{$\ba \cap D = \emptyset$}]{  
\centering
\begin{tikzpicture}[scale=.7, every node/.style={scale=0.6}]

\vertex(0, 0);
\vertex(2, .5);
\vertex(4, 0);

\triangleDown(0, 0);
\triangleUp(2, .5);
\triangleDown(4, 0);

\draw[line width=.5pt, ->] (.15, .07) -- (1.85, .52);
\draw[line width=.5pt, <-] (.15, -.05) -- (1.85, .4);
\draw[line width=.5pt, ->] (3.85, .07) -- (2.15, .52);
\draw[line width=.5pt, ->] (3.85, -.05) -- (2.15, .39);

\node[above] at (0, 0) {$a_0$};
\node[below=2pt] at (2, .5) {$a_1$};
\node[above] at (4, 0) {$a_2$};

\node at (0, -1.2) {$(\ba_2,\bt_2)$};
\node at (2, 1.7) {$(\ba_1,\bt_1)$};
\node at (4, -1.2) {$(\ba_0,\bt_0)$};

\node at (1, .5) {$f$};
\node at (3, .5) {$f$};
\node at (1.1, -.05) {$t_2$};
\node at (3, -.05) {$t_1$};

\draw[line width=.5pt, white] (-1, -1.65) -- (5, -1.65);

\end{tikzpicture}}  

&  

\subfloat[\scriptsize{$a_0 \in D$; $a_1, a_2 \not \in D$}]{  
\centering
\begin{tikzpicture}[scale=.7, every node/.style={scale=0.6}]

\draw (0, 1.25) circle (.75cm);

\vertex(0, 2);
\vertex(2, 2.65);
\vertex(4, 2);

\node[above] at (0, 2) {$a_0$};
\node[below=2pt] at (2, 2.65) {$a_1$};
\node[above] at (4, 2) {$a_2$};

\triangleUp(2, 2.65);
\triangleDown(4, 2);

\node at (0, 1.2) {$(\ba_2,\bt_2)$};
\node at (2, 3.85) {$(\ba_1,\bt_1)$};
\node at (4, 0.8) {$(\ba_0,\bt_0)$};

\draw[line width=.5pt, ->] (.1, 2.08) -- (1.83, 2.67);
\draw[line width=.5pt, <-] (.15, 2) -- (1.85, 2.58);
\draw[line width=.5pt, ->] (3.87, 2.1) -- (2.15, 2.68);
\draw[line width=.5pt, ->] (3.87, 2.0) -- (2.15, 2.58);

\node at (1, 2.6) {$f$};
\node at (3.2, 2.6) {$f$};
\node at (1.1, 2.06) {$t_2$};
\node at (3.15, 2.04) {$t_1$};

\draw[line width=.5pt, white] (-1, .4) -- (5, .4);

\end{tikzpicture}}  

&  

\subfloat[\scriptsize{$a_0, a_1 \not \in D$; $a_2 \in D$}]{  
\centering
\begin{tikzpicture}[scale=.7, every node/.style={scale=0.6}]

\draw (4, 1.25) circle (.7cm);

\vertex(0, 2);
\vertex(2, 2.65);
\vertex(4, 2);

\node[above] at (0, 2) {$a_0$};
\node[below=2pt] at (2, 2.65) {$a_1$};
\node[above] at (4, 2) {$a_2$};

\triangleUp(2, 2.65);
\triangleDown(0, 2);

\node at (0, 0.8) {$(\ba_2,\bt_2)$};
\node at (2, 3.85) {$(\ba_1,\bt_1)$};
\node at (4, 1.2) {$(\ba_0,\bt_0)$};

\draw[line width=.5pt, ->] (.1, 2.08) -- (1.83, 2.67);
\draw[line width=.5pt, <-] (.15, 2) -- (1.85, 2.58);
\draw[line width=.5pt, ->] (3.87, 2.1) -- (2.15, 2.68);
\draw[line width=.5pt, ->] (3.87, 2.0) -- (2.15, 2.58);

\node at (1, 2.6) {$f$};
\node at (3.2, 2.6) {$f$};
\node at (1.1, 2.06) {$t_2$};
\node at (3.15, 2.04) {$t_1$};

\draw[line width=.5pt, white] (-1.5, .4) -- (5.5, .4);

\end{tikzpicture}}  

\\  

\subfloat[\scriptsize{$a_0, a_2 \not \in  D$;  $a_1 \in D$}]{  
\centering
\begin{tikzpicture}[scale=.7, every node/.style={scale=0.6}]

\vertex(0, 0);
\vertex(-1, 1);
\vertex(1, 1);

\draw (0, -1.5) circle (1.5cm);

\triangleUp(-1, 1);
\triangleUp(1, 1);

\draw[line width=.5pt, transform canvas={yshift=.1ex, xshift=.25ex}, ->] (-.9, .9) -- (-.1, .1);
\draw[line width=.5pt, transform canvas={yshift=-.1ex, xshift=-.25ex}, <-] (-.9, .9) -- (-.1, .1);
\draw[line width=.5pt, transform canvas={yshift=.22ex, xshift=-.2ex}, ->] (.9, .9) -- (.1, .1);
\draw[line width=.5pt, transform canvas={yshift=-.22ex, xshift=.1ex}, ->] (.9, .9) -- (.1, .1);

\node at (.45, .7) {$f$};
\node at (.72, .3) {$t_1$};
\node at (-.35, .7) {$f$};
\node at (-.62, .3) {$t_2$};

\node at (-1, 2.2) {$(\ba_2,\bt_2)$};
\node at (0, -1.6) {$(\ba_1,\bt_1)$};
\node at (1, 2.2) {$(\ba_0,\bt_0)$};

\node[below left=-1pt] at (-1, 1) {$a_0$};
\node[below right=-1pt] at (1, 1) {$a_2$};
\node[below=2pt] at (0, 0) {$a_1$};

\draw[line width=.5pt, white] (-3.1, -3.1) -- (3.1, -3.1);

\end{tikzpicture}}  

&  

\subfloat[\scriptsize{$a_0, a_1 \in D$; $a_2 \not \in D$}]{  
\centering
\begin{tikzpicture}[scale=.7, every node/.style={scale=0.6}]

\vertex(0, 0);
\vertex(1.2, 0);

\node[above] at (0, 0) {$a_0$};
\node[above=1pt] at (1.2, 0) {$a_1\;$};
\node[above left=-1pt] at (2.4, .8) {$a_2$};

\draw[line width=.5pt, ->] (0.15, 0) -- (1.05, 0);
\node at (.6, .25) {$f$};

\vertex(2.4, .8);
\draw[line width=.5pt, transform canvas={yshift=.25ex, xshift=-.15ex}, <-] (1.35, 0.05) -- (2.28, .73);
\draw[line width=.5pt, transform canvas={yshift=-.2ex, xshift=.1ex}, <-] (1.35, 0.05) -- (2.28, .73);

\node at (1.75, .63) {$f$};
\node at (2, .13) {$t_1$};

\triangleBigRight(2.4, .8);
\node at (3.1, .8) {$(\ba_0,\bt_0)$};

\def\r{1.5};
\draw[line width=.5pt, domain=-242:60, xshift=.56cm, yshift=-1.39cm] plot[smooth] ({\r*cos(\x)}, {\r*sin(\x)});
\node at (.6, -1.5) {$(\ba_1,\bt_1)$};

\draw[line width=.5pt, white] (-2, -3) -- (4, -3);

\end{tikzpicture}}  

&  

\subfloat[\scriptsize{$a_0 \not \in D$; $a_1, a_2 \in D$}]{  
\centering
\begin{tikzpicture}[scale=.7, every node/.style={scale=0.6}]

\vertex(0, 0);
\vertex(1.2, 0);

\node[above] at (0, 0) {$a_1$};
\node[above] at (1.2, 0) {$a_2$};

\draw[line width=.5pt, <-] (0.15, 0) -- (1.05, 0);
\node at (.6, .25) {$f$};

\vertex(-1.4, .8);
\draw[line width=.5pt, transform canvas={yshift=.2ex, xshift=.1ex}, <-] (-.15, 0.05) -- (-1.28, .73);
\draw[line width=.5pt, transform canvas={yshift=-.2ex, xshift=-.1ex}, ->] (-.15, 0.05) -- (-1.28, .73);

\node at (-.65, .65) {$f$};
\node at (-.75, .12) {$t_1$};

\triangleBigLeft(-1.4, .8);
\node at (-2.1, .8) {$(\ba_1,\bt_1)$};
\node[above right=-1pt] at (-1.4, .8) {$a_0$};

\def\r{1.5};
\draw[line width=.5pt, domain=-242:60, xshift=.56cm, yshift=-1.39cm] plot[smooth] ({\r*cos(\x)}, {\r*sin(\x)});
\node at (.6, -1.5) {$(\ba_0,\bt_0)$};

\draw[line width=.5pt, white] (-3.2, -3) -- (3, -3);

\end{tikzpicture}}  

\\  

\subfloat[\scriptsize{$a_0, a_1, a_2 \not \in D$; $(\ba_2,\bt_2)$ mixed}]{  
\centering
\begin{tikzpicture}[scale=.7, every node/.style={scale=0.6}]

\node[above] at (0, 0) {$a_0$};
\node[below=1pt] at (1, 0) {$a_1$};
\node[above] at (2, 0) {$a_2$};

\vertex(0, 0);
\vertex(1, 0);
\vertex(2, 0);
\vertex(1, -1.5);

\draw (1, -3) circle (1.5cm);

\triangleUp(1, 0);
\begin{scope}[rotate around={-90:(2,0)}]
  \triangleUp(2, 0);
\end{scope}

\node at (2.95, 0) {$(\ba_0, \bt_0)$};
\node at (1, 1.2) {$(\ba_1, \bt_1)$};

\draw[line width=.5pt, transform canvas={yshift=.27ex}, ->] (.15, 0) -- (.85, 0);
\draw[line width=.5pt, transform canvas={yshift=-.27ex}, <-] (.15, 0) -- (.85, 0);
\draw[line width=.5pt, transform canvas={yshift=.27ex}, <-] (1.15, 0) -- (1.85, 0);
\draw[line width=.5pt, transform canvas={yshift=-.27ex}, <-] (1.15, 0) -- (1.85, 0);

\draw[thick, dashed] (0, 0) -- (1, -1.5);
\node[below=1pt] at (1, -1.5) {$b$};

\node at (.5, .25) {$f$};
\node at (.55, -.28) {$t_2$};
\node at (1.6, .25) {$f$};
\node at (1.6, -.28) {$t_1$};

\draw[line width=.5pt, white] (-2.3, -4.8) -- (5, -4.8);

\end{tikzpicture}}  

&  

\subfloat[\scriptsize{$a_0, a_1, a_2 \not \in D$; $(\ba_1,\bt_1)$ mixed}]{  
\centering
\begin{tikzpicture}[scale=.7, every node/.style={scale=0.6}]

\node[below=1pt] at (0, 0) {$a_0$};
\node[above] at (1, 0) {$a_1$};
\node[below=1pt] at (2, 0) {$a_2$};

\vertex(0, 0);
\vertex(1, 0);
\vertex(2, 0);
\vertex(1, -1.5);

\draw (1, -3) circle (1.5cm);

\triangleUp(0, 0);
\triangleUp(2, 0);

\node at (0, 1.2) {$(\ba_2, \bt_2)$};
\node at (2, 1.2) {$(\ba_0, \bt_0)$};

\draw[line width=.5pt, transform canvas={yshift=.27ex}, ->] (.15, 0) -- (.85, 0);
\draw[line width=.5pt, transform canvas={yshift=-.27ex}, <-] (.15, 0) -- (.85, 0);
\draw[line width=.5pt, transform canvas={yshift=.27ex}, <-] (1.15, 0) -- (1.85, 0);
\draw[line width=.5pt, transform canvas={yshift=-.27ex}, <-] (1.15, 0) -- (1.85, 0);

\draw[thick, dashed] (1, 0) -- (1, -1.5);
\node[below=1pt] at (1, -1.5) {$b$};

\node at (.5, .25) {$f$};
\node at (.55, -.28) {$t_2$};
\node at (1.6, .25) {$f$};
\node at (1.55, -.28) {$t_1$};

\draw[line width=.5pt, white] (-2.5, -4.8) -- (5, -4.8);

\end{tikzpicture}}  

&  

\subfloat[\scriptsize{$a_0, a_1, a_2 \not \in D$; $(\ba_0,\bt_0)$ mixed}]{  
\centering
\begin{tikzpicture}[scale=.7, every node/.style={scale=0.6}]

\node[above] at (0, 0) {$a_0$};
\node[below=1pt] at (1, 0) {$a_1$};
\node[above] at (2, 0) {$a_2$};

\vertex(0, 0);
\vertex(1, 0);
\vertex(2, 0);
\vertex(1, -1.5);

\draw (1, -3) circle (1.5cm);

\begin{scope}[rotate around={90:(0,0)}]
  \triangleUp(0, 0);
\end{scope}
\triangleUp(1, 0);

\node at (-0.95, 0) {$(\ba_2,\bt_2)$};
\node at (1, 1.2) {$(\ba_1, \bt_1)$};

\draw[line width=.5pt, transform canvas={yshift=.27ex}, ->] (.15, 0) -- (.85, 0);
\draw[line width=.5pt, transform canvas={yshift=-.27ex}, <-] (.15, 0) -- (.85, 0);
\draw[line width=.5pt, transform canvas={yshift=.27ex}, <-] (1.15, 0) -- (1.85, 0);
\draw[line width=.5pt, transform canvas={yshift=-.27ex}, <-] (1.15, 0) -- (1.85, 0);

\draw[thick, dashed] (2, 0) -- (1, -1.5);
\node[below=1pt] at (1, -1.5) {$b$};

\node at (.5, .25) {$f$};
\node at (.5, -.28) {$t_2$};
\node at (1.6, .25) {$f$};
\node at (1.5, -.28) {$t_1$};

\draw[line width=.5pt, white] (-3, -4.8) -- (4.7, -4.8);

\end{tikzpicture}}  

\\  


&  

\subfloat[\scriptsize{$a_0, a_1, a_2 \hspace{-1mm} \in \hspace{-1mm} D$}]{ 
\centering
\begin{tikzpicture}[scale=.7, every node/.style={scale=0.6}]

\vertex(0, 0);
\vertex(1.5, .5);
\vertex(3, 0);

\draw[line width=.5pt, ->] (0.15, .05) -- (1.35, .47);
\node at (.75, .48) {$f$};

\draw[line width=.5pt, <-] (1.65, .47) -- (2.85, .05);
\node at (2.4, .48) {$f$};

\def\r{1.8};
\draw[line width=.5pt, domain=-209:29, xshift=1.5cm, yshift=-1cm] plot[smooth] ({\r*cos(\x)}, {\r*sin(\x)});
\node at (1.5, -1.5) {$\bt$};

\node[above] at (0, 0) {$a_0$};
\node[above] at (1.5, .5) {$a_1$};
\node[above] at (3, 0) {$a_2$};

\draw[line width=.5pt, white] (-2, -3) -- (5.4, -3);

\end{tikzpicture}}  

&  


\end{tabular}
\caption{Possible configurations of a violating $\seq$-tour $\bar{a}$ with initial
(distinct) elements $a_0, a_1$ and $a_2$, in a 
$\sizeOf{\seq}$-quasi-acyclic
structure with active domain $D$. Triangles indicate acyclic subtours; the
circle indicates a subtour involving elements of $D$; see Lemmas~\ref{lem:cond_i}--\ref{lem:cond_x}.}
\label{fig:tours}

\end{figure}

\begin{lem}
\label{lem:cond_i}
Suppose $\gA$ is a $\sizeOf{\seq}$-quasi-acyclic structure containing a passive tour, $(\ba,\seq)$, whose first three elements, $a_0$, $a_1$ and $a_2$, are
distinct. Then there exist decompositions 
$\ba = a_0 a_1 \ba_0 a_2 \ba_1 a_1 \ba_2$ and $\seq = f \frev \bt_0 t_1 \bt_1 t_2 \bt_2$, such that 
(i) $(\ba_0,\bt_0)$ is an acyclic tour in $\gA$ starting at $a_2$;
(ii) $(\ba_1,\bt_1)$ is an acyclic tour in $\gA$ starting at $a_1$;
(iii) $(\ba_2,\bt_2)$ is an acyclic tour in $\gA$ starting at $a_0$;
(iv) $t_2(a_1,a_0)$ and $t_1(a_2,a_1)$ are true  in $\gA$ (Fig.~\ref{fig:tours}(\textsc{i})).
\end{lem}

\begin{proof}
Since $(\ba,\seq)$ is passive, it is acyclic.
Write  $\ba = a_0 \cdots a_{\ell - 1}$, $\br = r_0 \cdots r_{\ell - 1}$ and $a_\ell = a_0$; thus, $r_0 = f$ and $r_1 = \frev$. Let $i$ be the largest index ($2 \leq i < \ell$) such that $a_i = a_2$, let
$\ba_0 = a_2 \cdots a_{i-1}$
and $\bt_0 = r_2 \cdots r_{i-1}$. Thus, $(\ba_0,\bt_0)$ is an acyclic tour starting at $a_2$. By Lemma~\ref{lma:retrace}, $a_{i+1} = a_1$; write $t_1 = r_i$. Similarly, let $j$ be the largest index
($i+1 \leq j < \ell$) such that $a_j = a_1$, and let $\ba_1 = a_{i+1} \cdots a_{j-1}$
and $\bt_1 = r_{i+1} \cdots r_{j-1}$; thus, $(\ba_1,\bt_1)$ is an acyclic tour starting at $a_1$. 
By Lemma~\ref{lma:retrace} again, $a_{j+1} = a_0$; write $t_2 = r_j$. Let $\ba_2 = a_{j+1}, \ldots, a_{\ell-1}$ and $\bt_2 = r_{j+1}, \ldots, r_{\ell-1}$. This completes the tour (going back to $a_0$) and
$(\ba_2,\bt_2)$ is acyclic.
\end{proof}
\begin{lem}
\label{lem:cond_ii}
Suppose $\gA$ is a $\sizeOf{\seq}$-quasi-acyclic structure containing a mixed tour, $(\ba,\seq)$, whose first three elements, $a_0$, $a_1$ and $a_2$, are
distinct. Suppose further that $a_0$ is active, but $a_1$ and $a_2$ are passive. 
Then there exist decompositions 
$\ba = a_0 a_1 \ba_0 a_2 \ba_1 a_1 \ba_2$ and $\seq = f \frev \bt_0 t_1 \bt_1 t_2 \bt_2$,
such that 
(i) $(\ba_0,\bt_0)$ is an acyclic tour in $\gA$ starting at $a_2$;
(ii) $(\ba_1,\bt_1)$ is an acyclic tour in $\gA$ starting at $a_1$;
(iii) $(\ba_2,\bt_2)$ is a tour in $\gA$ starting at $a_0$;
(iv) $t_2(a_1,a_0)$ and $t_1(a_2,a_1)$ are true  in $\gA$ (Fig.~\ref{fig:tours}(\textsc{ii})).
\end{lem}

\begin{proof}
Write  $\ba = a_0 \cdots a_{\ell - 1}$, $\br = r_0 \cdots r_{\ell - 1}$ and $a_\ell = a_0$; thus, $r_0 = f$ and $r_1 = \frev$. Let $j$ be the smallest index
($2 \leq j < \ell$) such that $a_{j+1}$ is active. Since $a_0$ is active and $a_1$ passive, it follows by Lemma~\ref{lem:simple_db1} that $a_j = a_1$ and $a_{j+1}= a_0$. Thus, $\ba'= a_1 a_2 \cdots a_{j-1}$ is a passive---and
hence acyclic---tour. Considering the tour $\ba'$, let $i$ be the largest index ($2 \leq i < j$) such that $a_{i} = a_2$.
By Lemma~\ref{lma:retrace}, $a_{i+1} = a_1$.
Let $\ba_0 = a_2 \cdots a_{i-1}$ and $\bt_0 = r_2 \cdots r_{i-1}$;
let $\ba_1 = a_{i+1} \cdots a_{j-1}$
and $\bt_1 = r_{i+1} \cdots r_{j-1}$; and let $\ba_2 = a_{j+1} \cdots a_{\ell-1}$ and $\bt_2 = r_{j+1} \cdots r_{\ell-1}$.
In addition, write $t_1 = r_i$ and $t_2 = r_j$. 
Thus, $\ba = a_0 a_1 \ba_0 a_2 \ba_1 a_1 \ba_2$ and $\seq = f \frev \bt_0 t_1 \bt_1 t_2 \bt_2$. Moreover,
$(\ba_0,\bt_0)$  and $(\ba_1,\bt_1)$ are acyclic tours, and  $(\ba_2,\bt_2)$ is a tour.
\end{proof}

\begin{lem}
\label{lem:cond_iii}
Suppose $\gA$ is a $\sizeOf{\seq}$-quasi-acyclic structure containing a mixed tour, $(\ba,\seq)$, whose first three elements, $a_0$, $a_1$ and $a_2$, are
distinct. Suppose further that $a_2$ is active, but $a_0$ and $a_1$ are passive. Then there exist decompositions 
$\ba = a_0 a_1 \ba_0 a_2 \ba_1 a_1 \ba_2$ and $\seq = f \frev \bt_0 t_1 \bt_1 t_2 \bt_2$,
such that 
(i) $(\ba_0,\bt_0)$ is a tour in $\gA$ starting at $a_2$;
(ii) $(\ba_1,\bt_1)$ is an acyclic tour in $\gA$ starting at $a_1$;
(iii) $(\ba_2,\bt_2)$ is an acyclic tour in $\gA$ starting at $a_0$;
(iv) $t_2(a_1,a_0)$ and $t_1(a_2,a_1)$ are true  in $\gA$ (Fig.~\ref{fig:tours}(\textsc{iii})).
\end{lem}

\begin{proof}
Analogous to the proof of Lemma~\ref{lem:cond_ii}.
\end{proof}

\begin{lem}
\label{lem:cond_iv}
Suppose $\gA$ is a $\sizeOf{\seq}$-quasi-acyclic structure containing a mixed tour, $(\ba,\seq)$, whose first three elements, $a_0$, $a_1$ and $a_2$, are
distinct. Suppose further that $a_1$ is active, but $a_0$ and $a_2$ are passive. Then there exist decompositions $\ba = a_0 a_1 \ba_0 a_2 \ba_1 a_1 \ba_2$ and $\seq = f \frev \bt_0 t_1 \bt_1 t_2 \bt_2$,
such that 
(i) $(\ba_0,\bt_0)$ is an acyclic tour in $\gA$ starting at $a_2$;
(ii) $(\ba_1,\bt_1)$ is a tour in $\gA$ starting at $a_1$;
(iii) $(\ba_2,\bt_2)$ is an acyclic tour in $\gA$ starting at $a_0$;
(iv) $t_2(a_1,a_0)$ and $t_1(a_2,a_1)$ are true  in $\gA$ (Fig.~\ref{fig:tours}(\textsc{iv})).
\end{lem}

\begin{proof}
Analogous to the proof of Lemma~\ref{lem:cond_ii}.
\end{proof}

\begin{lem}
\label{lem:cond_v}
Suppose $\gA$ is a $\sizeOf{\seq}$-quasi-acyclic structure containing a mixed tour, $(\ba,\seq)$, whose first three elements, $a_0$, $a_1$ and $a_2$, are
distinct. Suppose further that $a_0$ and $a_1$ are active, but $a_2$ is passive. Then there exist decompositions $\ba = a_0 a_1 \ba_0 a_2 \ba_1$ and $\seq = f \frev \bt_0 t_1 \bt_1$,
such that (i) $(\ba_0,\bt_0)$ is an acyclic tour in $\gA$ starting at $a_2$;
(ii) 
$(a_0 \ba_1,f \bt_1)$ a tour in $\gA$ where $\ba_1$ is a non-empty sequence starting with $a_1$; (iii) $t_1(a_2,a_1)$ is true in $\gA$ (Fig.~\ref{fig:tours}(\textsc{v})).
\end{lem}

\begin{proof}
Write  $\ba = a_0 \cdots a_{\ell - 1}$, $\br = r_0 \cdots r_{\ell - 1}$ and $a_\ell = a_0$; thus, $r_0 = f$ and $r_1 = \frev$. Let $i$ be the smallest index ($2 \leq i < \ell$) such that $a_{i+1}$ is active.
Since $a_\ell$ is active, $i$ exists. By Lemma~\ref{lem:simple_db1}, 
$a_i = a_2$ and $a_{i+1} = a_1$. Let 
$\ba_0 = a_2 \cdots a_{i-1}$ and $\bt_0 = r_2 \dots r_{i-1}$. By construction, $(\ba_0,\bt_0)$ is a passive---hence acyclic---tour beginning at $a_2$. Now let $t_1 = r_i$, 
$\ba_1 = a_{i+1} \cdots a_{\ell-1}$ and $\bt_0 = r_{i+1} \cdots r_{\ell-1}$.
\end{proof}

\begin{lem}
\label{lem:cond_vi}
Suppose $\gA$ is a $\sizeOf{\seq}$-quasi-acyclic structure containing a mixed tour, $(\ba,\seq)$, whose first three elements, $a_0$, $a_1$ and $a_2$, are
distinct. Suppose further that $a_1$ and $a_2$ are active, but $a_0$ is passive. Then there exist decompositions 
$\ba = a_0 a_1 \ba_0 a_1 \ba_1$ and $\seq = f \frev \bt_0 t_1 \bt_1$,
such that 
(i) $(a_1 \ba_0,f^{-1} \bt_0)$ a tour in $\gA$ where $\ba_0$ is a non-empty sequence starting with $a_2$; 
(ii) $(\ba_1, \bt_1)$ is an acyclic tour in $\gA$ starting at $a_0$;
(iii) $t_1(a_1,a_0)$ is true in $\gA$ (Fig.~\ref{fig:tours}(\textsc{vi})).
\end{lem}

\begin{proof}
Analogous to the proof of Lemma~\ref{lem:cond_iv}.
\end{proof}

We remark at this point that, if $\gA$ is a $\sizeOf{\seq}$-quasi-acyclic structure containing a mixed tour $(\ba;\seq)$ whose first three elements, $a_0$, $a_1$ and $a_2$, are distinct, then it cannot be the case that $a_0$ and $a_2$ are active, but $a_1$ passive, as $a_0a_1a_2$ would then form a cycle in the graph of $\gA$.
In the statement of the next lemma, all arithmetic performed on indices is assumed to be modulo 3.
\begin{lem}
\label{lem:cond_vii_viii_ix}
Suppose $\gA$ is a $\sizeOf{\seq}$-quasi-acyclic structure containing a mixed tour, $(\ba,\seq)$, whose first three elements, $a_0$, $a_1$ and $a_2$, are
distinct and all passive. Then there exist decompositions $\ba = a_0 a_1 \ba_0 a_2 \ba_1 a_1 \ba_2$ and
$\seq = f \frev \bt_0 t_1 \bt_1 t_2 \bt_2$, 
such that: (i) for all $j$ ($0 \leq j < 3$) $(\ba_{j},\bt_{j})$ is a tour in $\gA$ starting at $a_{2-j}$;
(ii) two of these tours are acyclic and the third is mixed;
(iii) $t_2(a_1,a_0)$ and $t_1(a_2,a_1)$ are true  in $\gA$ (Fig.~\ref{fig:tours}(\textsc{vii})--(\textsc{ix})).
\end{lem}
\begin{proof}
Write  $\ba = a_0 \cdots a_{\ell - 1}$, $\br = r_0 \cdots r_{\ell - 1}$ and $a_\ell = a_0$; thus, $r_0 = f$ and $r_1 = \frev$. 
Let 
$h$ be the smallest index $2 \leq h  < \ell$ such that $a_{h+1}$ is active, and $k$ the 
largest index $h < k < \ell$ such that $a_k$ is active and
$a_{k+1}$ is passive. By Lemma~\ref{lem:simple_db2}, 
$a_h = a_{k+1}$, whence $a_2 \cdots a_{h} a_{k+2} \cdots a_{\ell-1}$ is a passive tour. 
Suppose first that $a_0$ is encountered along the walk $a_2 \cdots a_{h}$, i.e.~there exists $j$ ($2 < j < h$) such that $a_{j+1} = a_0$ (Fig.~\ref{fig:tours}(\textsc{vii})). Denote by $\ell'$ the smallest such value of $j$. It follows
that $\bb= a_0 a_1 a_2 \cdots a_{\ell'}$ is a passive tour. Repeating the reasoning of Lemma~\ref{lem:cond_i},
let $i$ be the largest index ($2 \leq i < \ell'$) such that $a_i = a_2$, let
$\ba_0 = a_2 \cdots a_{i-1}$
and $\bt_0 = r_2 \cdots r_{i-1}$. Thus, $(\ba_0,\bt_0)$ is a passive---and hence acyclic---tour starting at $a_2$. By Lemma~\ref{lma:retrace} (applied to $\bb$), $a_{i+1} = a_1$; write $t_1 = r_i$. Now let $m$ be the largest index 
($i+1 \leq m < \ell'$) such that $a_m = a_1$. 
Let $\ba_1 = a_{i+1} \cdots a_{m-1}$
and $\bt_1 = r_{i+1} \cdots r_{m-1}$; thus, $(\ba_1,\bt_1)$ is an acyclic tour starting at $a_1$. 
By Lemma~\ref{lma:retrace} again, $a_{m+1} = a_0$, whence, in 
fact $m+1 = \ell'$; write $t_2 = r_m$.
Let $\ba_2 = a_{\ell'} \cdots a_{\ell-1}$ and $\bt_2 = r_{\ell'} \cdots r_{\ell-1}$; 
then, $(\ba_2,\bt_2)$ is a mixed tour starting at $a_0$.

Now suppose that $a_0$ is not encountered again along the walk $a_2 \cdots a_{h}$, but $a_1$ is, i.e.~there exists $j$ 
($2 \leq j < h$) such that $a_{j+1} = a_1$ (Fig.~\ref{fig:tours}(\textsc{viii})).
Denote by $\ell'$ the smallest such value of $j$. 
Again, let $i$ be the largest index ($2 \leq i < \ell'$) such that $a_i = a_2$, let
$\ba_0 = a_2 \cdots a_{i-1}$
and $\bt_0 = r_2 \cdots r_{i-1}$. Thus, $(\ba_0,\bt_0)$ is an acyclic tour starting at $a_2$. 
By Lemma~\ref{lma:retrace}, $a_{i+1} = a_1$; write $t_1 = r_i$. Since $a_{i+1} = a_1$ and
$a_{h} = a_{k+1}$, 
we know that $\bb = a_0 a_{i+1} \cdots a_h a_{k+2} \dots a_{\ell-1}$ is a passive tour. Let $m$ be the largest
index $1 \leq m < \ell$ such that $a_m$ occurs in $\bb$ and $a_m= a_1$. 
By Lemma~\ref{lma:retrace} (applied to $\bb$), the next element in the tour $\bb$ must be $a_0$. Since, by assumption, $a_0$ does not occur among $a_2 \cdots a_h$,
we have $k+1 \leq m < \ell$. 
Thus, letting  $\ba_1 = a_{i+1} \cdots a_{m-1}$ and $\bt_1 = r_{i+1} \cdots r_{m-1}$, 
$(\ba_1,\bt_1)$ is a mixed tour starting at $a_1$. Let $t_2 = r_m$. 
Let $\ba_2 = a_{m+1} \cdots a_{\ell-1}$ and $\bt_1 = r_{m+1} \cdots r_{\ell-1}$; 
thus, $(\ba_2,\bt_2)$ is a passive---and hence acyclic---tour starting at $a_0$.

Finally, suppose that neither $a_0$ nor $a_1$ is encountered along the walk $a_2 \cdots a_{h}$
(Fig.~\ref{fig:tours}(\textsc{ix})). 
Since $a_{h} = a_{k+1}$, we know that $\bb = a_0 a_1  \cdots a_h a_{k+2} \dots a_{\ell-1}$ is a passive---and
hence acyclic---tour. Let
$i$ be the largest index ($2 \leq i < \ell$) such that $a_i$ occurs in $\bb$ and $a_i = a_2$.
By Lemma~\ref{lma:retrace}, the next element in the tour $\bb$ must be $a_1$. Since, by assumption, $a_1$ does not occur among $a_2 \cdots a_h$,
we have $k+1 \leq i < \ell$, and hence $a_{i+1} = a_1$. Let
$\ba_0 = a_2 \cdots a_{i-1}$ and $\bt_0 = r_2 \cdots r_{i-1}$. Thus, $(\ba_0,\bt_0)$ is a mixed tour starting at $a_2$.
Write $t_1 = r_i$. Since $a_{i+1} = a_1$ and $k+1  \leq i$, we know that $a_0 a_{i+1} \cdots a_{\ell-1}$ is a passive tour. Let $j$ be the largest
index $i+1 \leq j < \ell$ such that $a_j = a_1$, 
let $\ba_1 = a_{i+1} \cdots a_{j-1}$, and let $\bt_1 = r_{i+1} \cdots r_{j-1}$.
Thus, $(\ba_1,\bt_1)$ is a passive---and hence acyclic---tour starting at $a_1$.
Lemma~\ref{lma:retrace}, $a_{j+1} = a_0$; let $t_2 = r_j$. 
Let $\ba_2 = a_{j+1} \cdots a_{\ell-1}$, $\bt_2 = r_{j+1} \cdots r_{\ell-1}$. 
Thus  $(\ba_2,\bt_2)$ is a passive---and hence acyclic---tour starting at $a_0$. 
\end{proof}

\begin{lem}
\label{lem:cond_x}
Suppose $\gA$ is a $\sizeOf{\seq}$-quasi-acyclic structure containing a mixed tour, $(\ba,\seq)$, whose first three elements, $a_0$, $a_1$ and $a_2$, are
distinct. Suppose further that $a_0$, $a_1$ and $a_2$ are active. Then, writing
$\ba = a_0 a_1 \bb$ and $\seq = f \frev \bt$, the pair $[\bb a_0, \bt]$ is a walk in $\gA$ starting at $a_2$ (Fig.~\ref{fig:tours}(\textsc{x})).
\end{lem}

\begin{proof}
Completely trivial.
\end{proof}

Let us sum up what we have discussed so far. Fix a path-functional dependency $\kappa = \PFD[\baf f, \bag]$,
and let $\ell$ be the length of $\seq$.
Any violation of $\PFD[\baf f, \bag]$ in a structure $\gA$
can be identified with a $\seq$-tour in $\gA$ whose first three elements are distinct.
Thus, when determining the satisfiability of a $\GCDK$-formula $\phi \wedge \Delta \wedge K$, 
our goal is to determine whether there exists a model of $\phi \wedge \Delta$ 
containing no such tours. Indeed,
by Lemma~\ref{lem:big_cycles}, we may confine attention to $\ell$-quasi-acyclic models of $\phi \wedge \Delta$.
Now suppose $\ba$ is a sequence of elements in some such model suspected of being a
$\seq$-tour in $\gA$ whose first three elements are distinct.
Depending on which (if any) of these three initial elements
of an $\seq$-tour $\ba$ belong to the database, $\ba$ can be decomposed in different ways, as seen
in Fig.~\ref{fig:tours}. All these decompositions, except for that of
Fig.~\ref{fig:tours}(\textsc{i}), involve elements of the database.
The configuration of Fig.~\ref{fig:tours}(\textsc{i}),
it turns out, can be ruled out by means of a
$\GC$-formula. In addition,
by introducing additional predicates to our original signature, and writing formulas 
ensuring that these predicates are satisfied by database elements occurring at certain critical points in the various configurations depicted in
Fig.~\ref{fig:tours}(\textsc{ii})--(\textsc{ix}), we can reduce the problem of determining the existence of 
such violating configurations to a simple database check.
The next section addresses the task of establishing the required interpretations
of these additional predicates.

\section{Encoding critical violations in $\GC$}
\label{sec:tour_predicates}
Let $\kappa = \PFD[\baf f , \bag]$ be a binary path-functional dependency, and
let $\seq =  f f^{-1} \baf^{-1} \bag \bag^{-1}\baf$ as above.
In this section, we present three technical devices consisting of sets of $\GC$-formulas denoted, respectively,
$F_\kappa$, $B_\kappa$ and $I_\kappa$.
The first device enables us to describe acyclic tours in structures, the second to characterize certain sorts of ternary branching structures, and the third to characterize elements connected to each other by a  pair of walks forming an acyclic 
tour. In the next section, we shall employ these devices to characterize the ten cases depicted in Fig.~\ref{fig:tours},
and thus to rule out the presence of violations of $\kappa$ in a structure using $\GC$-formulas.
We write $t \in \seq$ to indicate that $t$ occurs in $\seq$, and
$\bt \vartriangleleft \seq$ to indicate that $\bt$ is a sub-word (i.e.~a contiguous sub-sequence) of $\seq$.  
Observe that the number of sub-words of $\seq$ is $|\seq|(|\seq|+1)/2$.  

Let the signature $\sigma_\kappa^1$ consist of $\sigma$ together with the unary predicates $\fan{\bt}$, one for each 
sub-word $\bt$ of $\seq$. The intention is that $\fan{\bt}$ will be satisfied by an element $a$ in certain sorts of structures 
just in case $a$ is the start of an acyclic $\bt$-tour. Accordingly, let $F_\kappa$ be the set 
consisting of the formula $\forall x \, \fan{\epsilon}(x)$ together with all the formulas
\begin{align*}
& \forall x \, 
   \big( \fan{\bt}(x) \leftrightarrow
           \bigvee_{\substack{r, \br, s, \bs:\\ \bt = r \br s \bs}}
             \exists y \, (x \neq y \land r(x, y) \land \fan{\br}(y)
             \land s(y, x) \land \fan{\bs}(x) ) \big), 
\end{align*}
where $\bt \vartriangleleft \seq$ is non-empty. Evidently the size of $F_\kappa$ is polynomial in the size of $\seq$.
\begin{lem}
\label{lem:fan}
Suppose $\gA$ is a $\sigma$-structure such that $\gA \models F_\kappa$. Let 
$a$ be an element of $A$ and $\bt$ a sub-word of $\seq$. If $\gA \models \fan{\bt}(a)$, then $a$ is the start of a
$\bt$-tour in $\gA$; conversely, if $a$ is the start of an acyclic
$\bt$-tour in $\gA$, then  $\gA \models \fan{\bt}(a)$.

\end{lem}

\begin{proof}
We prove by induction on the length of $\bt$ that if $\gA \models \fan{\bt}(a)$, then
$a$ is the start of a $\bt$-tour in $\gA$. If $\bt = \epsilon$, then the result is evident
(recall that every single element is the start of an $\epsilon$-tour).
Now, if $\bt \neq \epsilon$, since $\gA \models F_\kappa$, 
we can write $\bt = r \br s \bs$, for some $\br$ and $\bs$, such that there exists an
$a' \in A$ with $a \neq a'$, $\gA \models r(a, a')$, $\gA \models s(a', a)$, $\gA \models \fan{\br}(a')$ and $\gA \models \fan{\bs}(a)$. Thus, by inductive hypothesis, $a'$ is the start
of an $\br$-tour and $a$ is the start of an $\bs$-tour. Clearly, then, $a$
is the start of the $\bt$-tour.

For the converse, suppose that $a$ is the start of an acyclic
$\bt$-tour $\ba$ in $\gA$. We prove, by induction on the length of $\bt$, that $\gA \models \fan{\bt}(a_0)$. Write $\ba = a_0, \ldots, a_{\ell-1}$ and $\bt = t_0, \ldots, t_{\ell-1}$,
as usual, letting $a_\ell = a_0$.
If $\bt = \epsilon$ (i.e.~$\ell = 0$) then, since $\gA \models F_\kappa$ 
we have $\gA \models \fan{\epsilon}(a)$.
Now, suppose that $\bah \neq \epsilon$ (i.e.~$\ell > 0$).  
Then $\gA \models t_0(a_0, a_1)$. Let $j$ be the largest index ($1 \leq j < \ell$)
such that $a_j = a_1$. By Lemma~\ref{lma:retrace}, $a_{j+1} = a_0$, whence $\gA \models t_j(a_1,a_0)$. Write $r = t_0$,
$\br = t_1 \cdots t_{j-1}$, $s= t_{j}$ and
$\bs = t_{j+1} \cdots t_{\ell-1}$. 
Thus, $\bt = r \br s \bs$, $a_1$ is the start of an acyclic
$\br$-tour, and $a_0$ the start of an acyclic $\bs$-tour. By inductive hypothesis, $\gA \models \fan{\br}(a_1)$ and $\gA \models \fan{\bs}(a_0)$. Since $\gA \models F_\kappa$, we have $\gA \models \fan{\bt}(a_0)$.
\end{proof}

Now for our second device, which enables us to identify ternary-branching structures of the kind illustrated in Fig.~\ref{fig:tours}(\textsc{ii}). 
\begin{figure}
\begin{tikzpicture} [every node/.style={scale=0.8}]

\vertex(0, 0);
\vertex(2, .5);
\vertex(4, 0);

\triangleUp(2, .5);
\triangleDown(4, 0);

\draw[line width=.5pt, ->] (.15, .07) -- (1.85, .52);
\draw[line width=.5pt, <-] (.15, -.05) -- (1.85, .4);
\draw[line width=.5pt, ->] (3.85, .07) -- (2.15, .52);
\draw[line width=.5pt, ->] (3.85, -.05) -- (2.15, .39);

\node[above] at (0, 0.1) {$a_0$};
\node[below=2pt] at (2, .4) {$a_1$};
\node[above] at (4, 0.1) {$a_2$};

\node at (2, 1.7) {$(\ba_1,\bt_1)$};
\node at (4, -1.2) {$(\ba_0,\bt_0)$};

\node at (1, .5) {$g$};
\node at (3, .5) {$g$};
\node at (1.1, -.05) {$t_2$};
\node at (3, -.05) {$t_1$};

\draw[line width=.5pt, white] (-1, -1.65) -- (5, -1.65);

\end{tikzpicture}  
\caption{Satisfaction of the predicate $\branch{g,t_1,t_2,\bt_0,\bt_1}{1}$ at $a_0$.}
\label{fig:lastOneHopefully}
\end{figure}
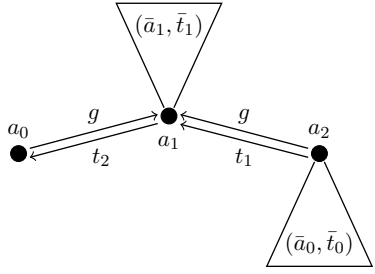
Let the signature $\sigma_\kappa^2$ consist of $\sigma_\kappa^1$ together with the unary predicates $\branch{g,t_1,t_2,\bt_0,\bt_1}{1}$, 
where $g$, $t_1$ and $t_2$ are letters in $\seq$, and
$\bt_0$, $\bt_1$ are sub-words of $\seq$. The intention is that $\branch{g,t_1,t_2,\bt_0,\bt_1}{1}$ should be satisfied 
by an element $a_0$ just in case there exists a $(g g^{-1} \bt_0 t_1 \bt_1 t_2)$-tour $a_0 a_1 \ba_0 a_2 \ba_1 a_1$, where $\ba_0$ is a $\bt_0$-tour starting at $a_2$, and
$\ba_1$ a $\bt_1$-tour starting at $a_1$, as depicted in
Fig.~\ref{fig:lastOneHopefully}. Let $B_\kappa^1$ be the set of all formulas
\begin{multline*}
\forall x (\branch{g, t_1, t_2,\bt_0, \bt_1}{1}(x) \leftrightarrow \\
\exists y \exists z (x \neq y \land  y \neq z \land  x \neq z  \, \land \hspace{6cm} \\ 
      (g(x,y) \land t_2(y, x) \land \fan{\bt_1}(y)) \land
      (g(z,y) \land t_1(z, y) \land \fan{\bt_0}(z)))),
\end{multline*}
where $g, t_1, t_2 \in \seq$, $\bt_0 \vartriangleleft \seq$ and $\bt_1 \vartriangleleft \seq$. 
\begin{lem}
\label{lma:branchSem1}
Suppose $\gA$ is a structure such that $\gA \models F_\kappa \cup B_\kappa^1$. Then $a_0 \in A$ satisfies $\branch{g,t_1,t_2,\bt_0,\bt_1}{1}$ if and only if there exist elements $a_1, a_2 \in A$ such that 
$a_0$, $a_1$, $a_2$ are distinct, 
$\gA \models g(a_0,a_1)$, $\gA \models g^{-1}(a_1,a_2)$, 
$\gA \models t_1(a_2,a_1)$, $\gA \models t_2(a_1,a_0)$, 
$\gA \models \fan{\bt_{0}}(a_2)$ and $\gA \models \fan{\bt_{1}}(a_1)$.
\end{lem}
\begin{proof}
	Immediate.
\end{proof}
We also introduce an additional family of unary predicates $\branch{g,t_1,t_2,\bt_0,\bt_2}{2}$,
satisfied by elements in configurations such as that of $a_1$ in Fig.~\ref{fig:tours}(\textsc{iv}), where $g = f$.
To this end, we define $B_\kappa^2$ to be the set of formulas
\begin{multline*}
\forall y (\branch{g,t_1,t_2,\bt_0,\bt_2}{2}(y) \leftrightarrow \\
\exists x \exists z (x \neq y \land  y \neq z \land  x \neq z \, \land \hspace{6cm} \\                 
      (g(x,y) \land t_2(y, x) \land \fan{\bt_2}(x)) \land            
      (g(z,y) \land t_1(z, y) \land \fan{\bt_0}(z)))),
\end{multline*}
where $g, t_1, t_2, \bt_0, \bt_1$ have the same range as for $B_1$.
\begin{lem}
\label{lma:branchSem2}
Suppose $\gA$ is a structure such that $\gA \models F_\kappa \cup B_\kappa^2$. Then $a_1 \in A$ satisfies $\branch{g,t_1,t_2,\bt_0,\bt_1}{2}$ if and only if there exist elements $a_0, a_2 \in A$ such that 
$a_0$, $a_1$, $a_2$ are distinct, 
$\gA \models g(a_0,a_1)$, $\gA \models g^{-1}(a_1,a_2)$, $\gA \models t_1(a_2,a_1)$, $\gA \models t_2(a_1,a_0)$, 
$\gA \models \fan{\bt_0}(a_2)$ and $\gA \models \fan{\bt_2}(a_0)$.
\end{lem}
Of course, the formulas in $B_\kappa^1$ and $B_\kappa^2$ are not in $\GC$, as they feature three variables. However, Lemma~\ref{lem:rewrite} ensures that they are equivalent to $\GC$-formulas. In the sequel, therefore, we shall treat $B_\kappa^1$ and $B_\kappa^2$ as sets of
$\GC$-formulas, understanding them to be replaced by their (harder-to-read) $\GC$-equivalents. We write
$B_\kappa = B_\kappa^1 \cup B_\kappa^2$.

Our third device enables us to identify elements at the start of structures of the kind illustrated in Fig.~\ref{fig:isthmus}, whose final elements satisfy a given unary predicate $p$ of $\sigma_\kappa^2$.
\begin{figure}

\centering
\begin{tikzpicture}

\node[left=2pt] at (-1, 0) {$b$};
\node[right=2pt] at (7, 0) {$a$};

\sandglass[(-1, 0), $\br_0$, $\bs_k$];
\sandglass[(1, 0), $\br_1$, $\bs_{k-1}$];
\sandglass[(3, 0), $\br_2$, $\bs_{k-2}$];
\sandglass[(7, 0), $\br_k$, $\bs_0$];

\doubleArrow[(-1, 0), (1, 0), $r_0$, $s_{k-1}$];
\doubleArrow[(1, 0), (3, 0), $r_1$, $s_{k-2}$];
\doubleArrow[(3, 0), (4.5, 0), $r_2$, $s_{k-3}$];
\doubleArrow[((5.5, 0), (7, 0), $r_{k-1}$, $s_0$];

\draw[line width=.5pt, <-] (7.5, 0) -- (8.5, 0);
\draw (8.5, -.35) rectangle (9.3, .35);
\node at (8.9, -.05) {$p$};

\draw[dotted, line width=2pt, line cap=round, dash pattern=on 0pt off 2\pgflinewidth] (4.6, 0) -- (5.4, 0);

\end{tikzpicture}
\caption{An isthmus starting at $a$ and terminating in an element $b$ satisfying a unary predicate $p$; 
here, $\br = \br_0 r_0 \cdots \br_k r_k$ and $\bs = \bs_0 s_0 \cdots \bs_k s_k$.}
\label{fig:isthmus}
\end{figure}
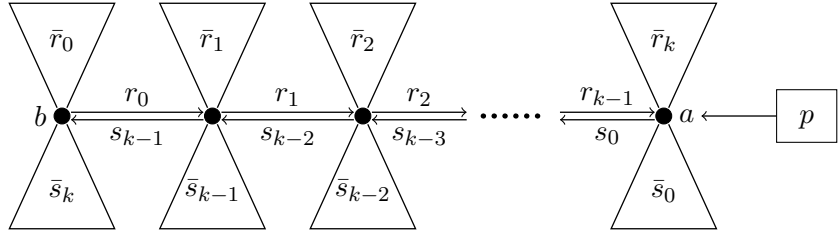
Let the signature $\sigma_\kappa^3$ 
consist of $\sigma_\kappa^2$ together with the unary predicates $\isth{\br, p, \bs}$, where $\br$ and $\bs$ are sub-words of $\seq$, and $p$ is any unary predicate of $\sigma_\kappa^2$. (Note that $\sigma_\kappa^2$ contains predicates of the forms
$\branch{g,t_1,t_2,\bt_0,\bt_1}{1}$ and $\branch{g,t_1,t_2,\bt_0,\bt_2}{2}$.) The intention is that 
$\isth{\br, p, \bs}$ will be satisfied by an element $b$ in certain sorts of structures just in case there exists an element $a$ satisfying $p$, as well as an $\br$-walk from
$b$ to $a$ and an $\bs$-walk from $a$ to $b$, which together form acyclic $\br\bs$-tour. We refer to such a tour, informally, as an {\em isthmus}. Accordingly, let $I_\kappa$ be the set of all formulas
\begin{multline*}
\forall x \big[ \isth{\br, p, \bs}(x) \leftrightarrow
                         \big((\fan{\br}(x) \land p(x) \land \fan{\bs}(x)) \lor\\
           \bigvee_{\substack{\br = \br' r \br'', \\ \bs = \bs'' s \bs'}}
             \hspace{-5mm} \big(\fan{\br'}(x) \land \fan{\bs'}(x) \land
 \exists y(x \neq y \land r(x, y) \land s(y, x) \land \isth{\br'', p, \bs''}(y) ) \big) \big) \big],
\end{multline*}
where
$\br \vartriangleleft \seq$, $\bs \vartriangleleft \seq$ and $p$ is a unary predicate of $\sigma_\kappa^2$. Again, the sizes of the signature $\sigma_\kappa^3$ and the set of formulas $I_\kappa$ are polynomially bounded in $|\seq|$.
\begin{lem}
\label{lem:isth}
Let $\br' \vartriangleleft \seq$, $\br'' \vartriangleleft \seq$ and
$p \in \sigma_\kappa^2$. 
Suppose $\gA \models F_\kappa \cup I_\kappa$. If $\gA \models \isth{\br, p, \bs}(b)$, then
there exists an element $a \in A$ such that $\gA \models p(a)$, a walk $[\bb a, \br]$ from $b$ to $a$, and a walk $[\bc b, \bs]$ from $a$ to $b$. Conversely, if such $a$, $\bb$, $\br$, $\bc$, $\bs$ exist, with $(\bb \bc, \br \bs)$ additionally an acyclic tour, then $\gA \models \isth{\br, p, \bs}(b)$.
\end{lem}

\begin{proof}
Suppose that $\gA \models \isth{\br, p, \bs}(b)$. To 
construct the required $a$, $\bb$, $\br$, $\bc$ and $\bs$, 
we proceed by 
induction on the length of the combined sequence $\br \bs$.
For the case $|\br \bs| = 0$,
if $\gA \models
\fan{\br}(b) \land p(b) \land \fan{\bs}(b)$, then by Lemma~\ref{lem:fan}, there
exists an $\br$-tour $\bb$ and a
$\bs$-tour $\bc$, both starting at $b$. Now set $a = b$. Thus,
$[\bb a,\br]$ is a walk from $b$ to itself, and $[\bc b,\bs]$ a walk from $a$ to itself, 
as required.
For the case $|\br \bs| > 0$, since $\gA \models  I_\kappa$,
we may write $\br = \br' r \br''$ and $\bs'' = \bs'' s \bs'$,
such that $\gA \models \fan{\br'}(b) \land \fan{\bs'}(b)$ and there exists 
$b' \in A$ with $\gA \models r(b, b')$, $\gA \models s(b', b)$ and
$\gA \models \isth{\br'', p, \bs''}(b')$. By inductive hypothesis, 
there exist an element $a \in A$ and walks $[\bb'' a, \br'']$, $[\bc'' b', \bs'']$, with $b'$ the
first element of $\bb''$, $a$ the
first element of $\bc''$, and $\gA \models p(a)$. On the other hand, from Lemma~\ref{lem:fan}, $b$ is the start of an $\br'$-tour, say $\bb'$, and of an $\bs'$-tour, say $\bc'$. Writing $\bb= 
\bb' b' \bb''$ and $\bc= 
\bc'' b' \bc'$, we see that $\bb a$ is an  
$\br$-walk from $b$ to $a$, and
$\bc b$ is an $\bs$-walk from $a$ to $b$, as required.

For the converse, suppose that there exist an element $a \in A$ and 
walks $[\bb a, \br]$, $[\bc b, \bs]$, with $b$ the
first element of $\bb$ and $a$ the
first element of $\bc$, such that 
$(\bb \bc, \br \bs)$ is an acyclic tour, and $\gA \models p(b)$.
We again proceed by 
induction on $|\br \bs|$, supposing that the result holds for sequences of smaller combined length.
Write $\bb b = b_0 \cdots b_\ell$ and $\bc a = c_0 \cdots c_m$. Thus, $b = b_0 = c_m$ and $a = b_\ell = c_0$.
Let $i$ be the largest index ($0 \leq i \leq \ell$) such that $b_{i} = b_0$. If $i = m$, then $a = b$. In that
case, $a$ satisfies $p$; moreover, $(\bb,\br)$ and $(\bc,\bs)$
are acyclic tours, and so by Lemma~\ref{lem:fan}, $b$ satisfies $\fan{\br}$ and $\fan{\bs}$.
Since $\gA \models I_\kappa$, $a$ also satisfies $\isth{\br, p, \bs}$. 
Suppose, then that $i < \ell$, so that $b_i \neq b_{i+1}$. Then $\bdd = b_i \cdots b_{\ell -1} c_0 \cdots c_{m-1}$ is an acyclic tour starting at $b$. Let $c$ be the last element of $\bdd$ equal to $b_{i+1}$.  
Applying Lemma~\ref{lma:retrace} to $\bdd$, the next element after $c$ must be $b_i =b_0 = b$; and since $c_0 \neq b_0$ and
by construction $b_0$ does not occur in $b_{i+1} \cdots b_{\ell-1}$, we have $c = c_j$ for some $j$ $(1 \leq j < m)$. Thus,
$b = b_i = c_{j+1}$. Now let $\bb' = b_0 \cdots b_{i-1}$, 
$\bb'' = b_{i+1} \cdots b_{\ell-1}$, $\bc'' = c_0 \cdots c_{j-1}$, 
$\bc' = c_{j+1} \cdots c_{m-1}$; define the sequences $\br'$, $\br''$, $\bs''$, $\bs'$ correspondingly.
Write $b' = b_{i+1} = c_j$, $r = r_i$ and $s = s_j$. Then $\gA \models r(b,b')$
and $\gA \models s(b',b)$. Moreover, $(\bb',\br')$  and $(\bc',\bs')$ are acyclic tours beginning at $b$, 
$[\bb'' a,\bs'']$ is a walk from $b'$ to $a$ and $[\bc'' b',\bs'']$ a walk from $a$ to $b'$ such that 
the tour $(\bb'' \bc'',\br'' \bs'')$ is acyclic.
By Lemma~\ref{lma:retrace}, $b$ satisfies $\fan{\br'}$ and $\fan{\bs'}$; and
by inductive hypothesis, $b'$ satisfies $\gA \models \isth{\br'', p, \bs''}$.  
Since $\gA \models I_\kappa$, $b$ satisfies $\isth{\br, p, \bs}$.
\end{proof}


\section{Complexity of $\GCDK$}
\label{sec:gcdk}
With these tools at our disposal, we may return to the task of reducing the
satisfiability and finite satisfiability problems for $\GCDK$ to the corresponding
problems for $\GCD$. As before, let $\kappa = \PFD[\baf f , \bag]$ be a binary path-functional dependency, and define $\seq =  f f^{-1} \baf^{-1} \bag \bag^{-1}\baf$. We have observed that $\kappa$ is critically violated in a structure $\gA$ if and only if $\gA$ contains an $\seq$-tour $\ba$ whose first three elements are distinct. Furthermore, we have shown that, in such a case, $\ba$ must fall under exactly one of the ten cases depicted in Fig.~\ref{fig:tours}. 

Consider first the case illustrated in Fig.~\ref{fig:tours}(\textsc{i}): $\ba$ is passive. 
Define $V_{\kappa,{\rm (i)}}$ to be the $\GC$-sentence
\begin{equation*}
\exists x \bigvee \big\{\fan{\bt_2}(x) \land \branch{f, t_1, t_2, \bt_0,\bt_1}{1}(x)
\mid \seq = f \frev \bt_0 t_1 \bt_1 t_2 \bt_2 \big\}.
\end{equation*}
The following lemma states that we can now rule out such critical violations of $\kappa$ by writing the 
$\GC$-formula $\neg V_{\kappa,{\rm (i)}}$.
\begin{lem}
\label{lem:viol_cond_i}
Suppose $\gA \models F_\kappa \wedge B_\kappa$.
If $\gA$ is $\sizeOf{\seq}$-quasi-acyclic, but contains a critical violation of $\kappa$ having the form of a passive $\seq$-tour, then $\gA \models V_{\kappa,{\rm (i)}}$. Conversely,
if $\gA \models V_{\kappa,{\rm (i)}}$, then $\gA$ contains a critical violation of $\kappa$.
\end{lem}
\begin{proof}
Suppose $\gA$ contains a passive $\seq$-tour whose first three elements, $a_0$, $a_1$ and $a_2$ are distinct.
By Lemma~\ref{lem:cond_i}, there is a decomposition $f \frev \bt_0 t_1 \bt_1 t_2 \bt_2$ of $\seq$ and
distinct $a_0, a_1, a_2 \in A$ such that 
$f(a_0, a_1)$, $f^{-1}(a_1, a_2)$, $t_1(a_2, a_1)$ and $t_2(a_1, a_0)$ all hold, and each
$a_i$ ($0 \leq i < 3$) is the start of a $\bt_{2-i}$ tour.  
Since $\gA$ is $\sizeOf{\seq}$-quasi-acyclic, these tours must be acyclic.
By three applications of Lemma~\ref{lem:fan},
each $a_i$ satisfies $\fan{\bt_{2-i}}$. By Lemma~\ref{lma:branchSem1} 
and construction of $V_{\kappa,{\rm (i)}}$, $\gA \models V_{\kappa,{\rm (i)}}$.

Suppose, conversely, $\gA \models V_{\kappa,{\rm (i)}}$. 
By  construction of $V_{\kappa,{\rm (i)}}$ and Lemma~\ref{lma:branchSem1}, 
there is a decomposition
$f \frev \bt_0 t_1 \bt_1 t_2 \bt_2$ of $\seq$ and three 
distinct elements $a_0$, $a_1$, $a_2$ such that $f(a_0, a_1)$, $f^{-1}(a_1, a_2)$, $t_1(a_2, a_1)$, $t_2(a_1, a_0)$
all hold, and each $a_i$ ($0 \leq i < 3$) satisfies $\fan{\bt_{2-i}}$.
By Lemma~\ref{lem:fan}, each $a_i$ the start of a $\bt_{2-i}$ tour. 
Thus, $a_0$ is the start of an $\seq$-tour in $\gA$ whose first elements are
distinct, whence $\gA$ contains a critical violation of $\kappa$.
\end{proof}

Now consider the case where $(\ba;\seq)$ is as in Fig.~\ref{fig:tours}(\textsc{ii}): $a_0$ is active, but $a_1$ and $a_2$ are passive. The following lemma states that we can now rule out such critical violations of $\kappa$ by checking some conditions on $\Delta$. Note that, if a database $\Delta$ is complete with respect to some signature, and $p$ is a  predicate in that signature, then, for $\gA \models \Delta$ and $\ba$ a tuple of active elements of $\gA$ of the
same arity as $p$, the conditions $\gA \models p(\ba)$
and $p(\ba) \in \Delta$ are equivalent.
\begin{lem}
\label{lem:viol_cond_ii}
Suppose $\gA \models F_\kappa \wedge B_\kappa\wedge \Delta$, where $\Delta$ is complete with respect to $\sigma^3_\kappa$.
If $\gA$ is $\sizeOf{\seq}$-quasi-acyclic, but contains an $\seq$-tour $a_0 \cdots a_{\ell-1}$ such that
$a_0$, $a_1$, $a_2$ are distinct with $a_0$ active, but $a_1$ and $a_2$ passive, then there exists a 
decomposition 
\begin{equation*}
\seq = f f^{-1} \bt_0 t_1 \bt_1 t_2 \{\br_i r_i\}_{i=0}^{m-1} \br_m
\end{equation*}
and a sequence of database elements $b_0 \cdots b_m$ with $a_0 = b_0 = b_m$, such that:\linebreak
(i) $\branch{f,t_1,t_2,\bt_0,\bt_1}{1}(b_0) \in \Delta$;
(ii) $r_i(b_i,b_{i+1}) \in \Delta$ for all $i$ ($0 \leq i < m$); and
(iii) $\fan{\br_i}(b_i) \in \Delta$ for all $i$ ($0 \leq i \leq m$). Conversely,
if such a decomposition and sequence of database elements exists, then $\gA$ contains a critical violation of $\kappa$.
\end{lem}
\begin{proof}
Suppose $\gA$ contains a $\seq$-tour $a_0 \cdots a_{\ell-1}$ such that
	$a_0$ is active, but $a_1$ and $a_2$ are passive. By Lemma~\ref{lem:cond_ii} there exist decompositions 
	$\ba = a_0 a_1 \ba_0 a_2 \ba_1 a_1 \ba_2$ and $\seq = f \frev \bt_0 t_1 \bt_1 t_2 \bt_2$,
	such that $f(a_0,a_1)$, $t_2(a_1,a_0)$, $f(a_2,a_1)$ and $t_1(a_2,a_1)$ are all true  in $\gA$,
	$(\ba_{2-i},\bt_{2-i})$ ($i = 1, 2$) is an acyclic tour starting at $a_i$, and $(\ba_2,\bt_2)$
	is a tour in $\gA$ starting at $a_0$.
	By two applications of Lemma~\ref{lem:fan} and Lemma~\ref{lma:branchSem1}, 
	$\branch{f,t_1,t_2,\bt_0,\bt_1}{1}(a_0) \in \Delta$.
	Furthermore, by Lemma~\ref{lem:dbcycle_decomp}, the tour
	$(\ba_2, \bt_2)$ starting at $a_0$ can be further decomposed
	as 
    $(\{\bb_i b_i\}_{i=0}^{m-1} \bb_m, \{\br_i r_i\}_{i=0}^{m-1} \br_m)$
    such that $(b_0 \cdots b_{m-1}, r_0 \cdots r_{m-1})$ is an active tour, and 
    each $(\bb_i, \br_i)$ is an acyclic tour beginning at $b_i$ ($1 \leq i \leq m$).
    Thus, we have $r_i(b_i,b_{i+1}) \in \Delta$ for all $i$ ($0 \leq i < m$).
    Finally, by 
    repeated applications of Lemma~\ref{lem:fan}, we have $\fan{\br_i}(b_i) \in \Delta$ all $i$ ($0 \leq i \leq m$).

    Conversely, suppose the decomposition $f f^{-1} \bt_0 t_1 \bt_1 t_2 \{\br_i r_i\}_{i=0}^{m-1} \br_m$
	and
	\mbox{database} elements $b_0 \cdots b_{m}$ exist with the advertised properties.
	By
 Lemma~\ref{lma:branchSem1} and two applications of Lemma~\ref{lem:fan}, $a_0 = b_0$ is the start of an $f f^{-1} \bt_0 t_1 \bt_1 t_2$-tour, say $\bdd$, whose first three elements are distinct. Moreover, by Lemma~\ref{lem:fan}, each $b_i$ ($0 \leq i \leq m$) is the start of an $\br_{i}$ tour, say $\bb_i$.
	Concatenating these tours, $a_0$ is the start of a $\seq$-tour, namely $\bdd \{\bb_i b_i\}_{i=0}^{m-1}\bb_m$, whose first three elements are
	distinct, whence $\gA$ contains a critical violation of $\kappa$.
\end{proof}

Violations of $\kappa$ having the forms shown in Fig.~\ref{fig:tours}(\textsc{iii})--(\textsc{vi}) can be ruled out in essentially the same way as for Fig.~\ref{fig:tours}(\textsc{ii}). Fig.~\ref{fig:tours}(\textsc{iii}) is symmetric to
Fig.~\ref{fig:tours}(\textsc{ii}); Fig.~\ref{fig:tours}(\textsc{iv}) is similar to Fig.~\ref{fig:tours}(\textsc{ii}), but using
the predicate $\text{branch}_2$ correspondingly. Fig.~\ref{fig:tours}(\textsc{v}) is, again, similar to Fig.~\ref{fig:tours}(\textsc{ii}), where the predicate $\branch{f,t_1,t_2,\bt_0,\bt_1}{1}$ is satisfied by
$a_0 = b_0$,
but with the extra requirement that $a_1 = b_1$, $f(b_0, b_1) \in \Delta$, and $\bt_1 = \epsilon$.
(Note, in the previous sentence, that the element $b_1$
is unique because $f$ is functional.)
Fig.~\ref{fig:tours}(\textsc{vi}) is symmetric to Fig.~\ref{fig:tours}(\textsc{v}).
Consequently, we turn our attention to the case of Fig.~\ref{fig:tours}(\textsc{vii})--(\textsc{ix}).
\begin{lem}
\label{lem:viol_cond_vii}
Suppose 
$\gA \models F_\kappa \wedge B_\kappa \land I_\kappa \wedge \Delta$, where $\Delta$ is complete with respect to $\sigma^3_\kappa$.
If $\gA$ is $\sizeOf{\seq}$-quasi-acyclic, but contains a mixed $\seq$-tour $\ba$ whose first three elements are all passive, then 
there exists a decomposition $\seq = f f^{-1} \bt_0 t_1 \bt_1 t_2 \bt_2$ and 
a sequence of active elements $b_0 \cdots b_m$, with $b_m = b_0$,
such that the following hold:
\begin{enumerate}[noitemsep, label=(\roman*)]
\item one of the three conditions
\begin{enumerate}[noitemsep, label=(\alph*)]
\item 
$\bt_2= \bs \{\br_i r_i\}_{i=0}^{m-1} \br_m \br$,
$\isth{\br, \branch{f,t_1,t_2,\bt_0,\bt_1}{1},\bs} (b_0) \in \Delta$
\item 
$\bt_1= \bs \{\br_i r_i\}_{i=0}^{m-1} \br_m \br$,
$\isth{\br, \branch{f,t_1,t_2,\bt_2,\bt_0}{2},\bs} (b_0) \in \Delta$
\item 
$\bt_0= \bs \{\br_i r_i\}_{i=0}^{m-1} \br_m \br$,
$\isth{\br, \branch{f,t^{-1}_2,t_1^{-1},\bt_1,\bt_2}{1},\bs} (b_0) \in \Delta$
\end{enumerate}
obtains;
\item $r_i(b_i,b_{i+1}) \in \Delta$ for all $i$ ($0 \leq i < m$); and
\item $\fan{\br_i}(b_i) \in \Delta$ for all $i$ ($0 \leq i \leq m$).
\end{enumerate}
Conversely,
if any such decomposition and sequence of database elements exists, then $\gA$ contains a critical violation of $\kappa$.
\end{lem}
\begin{proof}
Suppose $\gA$
contains a mixed $\seq$-tour $\ba$ whose first three elements are all passive. By Lemma~\ref{lem:cond_vii_viii_ix}, there exist decompositions $\ba = a_0 a_1 \ba_0 a_2 \ba_1 a_1 \ba_2$ and
$\seq = f \frev \bt_0 t_1 \bt_1 t_2 \bt_2$, 
such that:
(i) for all $j$ ($0 \leq j < 3$) $(\ba_{j},\bt_{j})$ is a tour in $\gA$ starting at $a_{2-j}$;
(ii) two of these tours are acyclic and the third is mixed;
(iii) $t_2(a_1,a_0)$ and $t_1(a_2,a_1)$ are true  in $\gA$.

Assume, for definiteness, that $(\ba_{0},\bt_{0})$ and $(\ba_{1},\bt_{1})$ are acyclic, and $(\ba_{2},\bt_{2})$ is mixed. By two applications of Lemma~\ref{lem:fan} and Lemma~\ref{lma:branchSem1}, 
$\branch{f,t_1,t_2,\bt_0,\bt_1}{1}(a_0) \in \Delta$. Now consider the mixed tour $(\ba_{2},\bt_{2})$, starting
at the passive element $a_0$. Write $a = a_0$, the initial element of this tour. By 
Lemma~\ref{lem:isthmus}, 
there exists an element $b$ such that we can write $\ba_{2} = \bc \ba^* \bb$ and $\bt_2 = \bs \bt^* \br$,  where 
$(\bb a,\br)$ is a walk from $b$ to $a$, $(\bc b,\bs)$ is a walk from $a$ to $b$, and  $(\ba^*,\bt^*)$ is a tour starting
at $b$. Write $b_0 = b$. Now applying Lemma~\ref{lem:isth}, we see that  
$\isth{\br, \branch{f,t_1,t_2,\bt_0,\bt_1}{1},\bs}(b_0) \in \Delta$. Further, applying Lemma~\ref{lem:dbcycle_decomp} to the tour
$(\ba^*,\bt^*)$, we may write 
$\ba^* = \{\bb_i b_i\}_{i=0}^{m-1} \bb_m$ and
$\bt^* = \{\br_i r_i\}_{i=0}^{m-1} \br_m$
such that $r_i(b_i,b_{i+1}) \in \Delta$ for all $i$ ($0 \leq i < m$), and
$(\bb_i, \br_i)$ is an acyclic tour starting at $b_i$ for all $i$ ($0 \leq i \leq m$). By 
repeated applications of Lemma~\ref{lem:fan}, $\fan{\br_i}(b_i) \in \Delta$ for all $i$ ($0 \leq i \leq m$).
This yields the conditions required in the lemma, with disjunct (a) realized in condition (i). 
If $(\ba_{1},\bt_{1})$ is mixed, then we obtain disjunct (b), and 
if $(\ba_{0},\bt_{0})$ is mixed, disjunct (c), by similar arguments.

Conversely, suppose that 
there exists a decomposition as described by the lemma, with disjunct (a) realized in condition (i).
By (i)(a) and Lemma~\ref{lem:isth}, there exist an element $a \in A$, a walk $[\bb a, \br]$ from $b_0$ to $a$, and a walk $[\bc b_0, \bs]$ from $a$ to $b_0$, such that 
$(\bb \bc; \br \bs)$ is an acyclic tour, and
$\gA \models \branch{f,t_1,t_2,\bt_0,\bt_1}{1}(a)$.
Write $a_0 = a$, so that, by Lemma~\ref{lma:branchSem1},
there exist elements $a_1, a_2 \in A$ such that 
$a_0$, $a_1$, $a_2$ are distinct, 
$\gA \models f(a_0,a_1)$, $\gA \models f(a_2,a_1)$, 
$\gA \models t_1(a_2,a_1)$, $\gA \models t_2(a_1,a_0)$ and for $i = 1, 2$, 
$\gA \models \fan{\bt_{2-i}}(a_i)$.  By Lemma~\ref{lem:fan}, 
there is an
$\bt_{2-i}$-tour starting at $a_i$ for $i = 1, 2$.  
Starting at the element $ a= a_0$, and concatenating these various
tours and single steps, we obtain the tour
\begin{equation*}
(a_0 a_1 \ba_0 a_2 \ba_1 a_1 \bc \bb,\     f f^{-1} \bt_0 t_1 \bt_1 t_2 \bs \br)
\end{equation*}
passing through $b_0$. Note that $b_0$ is the initial element of $\bb$. 
From condition (ii) and repeated applications of Lemma~\ref{lem:fan},
there exists a tour $(\bb_j,\br_j)$ starting at $b_j$ for all $j$ ($0 \leq j < m$)
and a tour $(\bb_m,\br_m)$ starting at $b_0$. Thus, from condition (iii), there exists
a tour
\begin{equation*}
(\bb_0 b_0  \dots \bb_{m-1} b_{m-1} \bb_m,\ r_0 \br_0 \dots \br_{m-1} r_{m-1} \br_m)
\end{equation*}
starting at $b_0$. Inserting the second of these tours into the first at the element $b_0$, we obtain the tour 
\begin{equation*}
(a_0 a_1 \ba_0 a_2 \ba_1 a_1 \bc \{\bb_i b_i\}_{i=0}^{m-1} \bb_m  \bb,\     
f f^{-1} \bt_0 t_1 \bt_1 t_2 \bs \{\br_i r_i\}_{i=0}^{m-1} \br_m \br).
\end{equation*}
But this is just an
$\seq$-tour whose first three elements, $a_0, a_1, a_2$,  are distinct, so that $\kappa$ has a critical violation.
The disjuncts (b) and (c) in condition (i) are dealt with similarly.
\end{proof}

This leaves just the final case of Fig.~\ref{fig:tours} to cover.
\begin{lem}
\label{lem:viol_cond_x}
Suppose $\gA \models F_\kappa \cup B_\kappa \cup I_\kappa$.
If $\gA$ is $\sizeOf{\seq}$-quasi-acyclic, but contains a mixed $\seq$-tour  whose three initial  elements are all active, then 
there exist a decomposition $\seq = r_0 r_1 \{\br_i r_i\}_{i=2}^m \br_m$
and a sequence of active elements $b_0 \cdots b_m$ with $b_0, b_1, b_2$ distinct and $b_m = b_0$,
such that
\begin{enumerate}[noitemsep, label=(\roman*)]
\item $r_i(b_i,b_{i+1}) \in \Delta$ for all $i$ ($2 \leq i < m$); and
\item $\fan{\br_i}(b_i) \in \Delta$ for all $i$ ($0 \leq i \leq m$).
\end{enumerate}
Conversely,
if such a decomposition exists, then $\gA$ contains a critical violation of $\kappa$.
\end{lem}
\begin{proof}
Easy using by-now familiar reasoning from Lemmas~\ref{lem:dbcycle_decomp}, \ref{lem:cond_x} and~\ref{lem:fan}.
\end{proof}

Lemmas~\ref{lem:viol_cond_i}--\ref{lem:viol_cond_x} yield a reduction from the satisfiability and finite satisfiability problems
for $\GCDK$ to the corresponding problems for $\GCD$. For definiteness, we consider
satisfiability; the reduction for finite satisfiability is identical.
By Lemma~\ref{lma:splitting}, we may confine attention to $\GC$-formulas of the form $\psi \wedge \Gamma \wedge K$,
where $\psi$ is a $\GC$-formula, $\Gamma$ a database and $K$ a set of unary and binary path-functional dependencies.
We may assume that $\Gamma$ is consistent, since if not, $\psi$ is certainly unsatisfiable.
Furthermore, we may suppose that every $\kappa \in K$ is binary, since unary path-functional dependencies can be easily
eliminated. Let $F_K$ be the conjunction of
all the $F_\kappa$ for $\kappa \in K$; similarly for $B_K$ and $I_K$. Let $N_K$ be the conjunction of
all the formulas $\neg V_{\kappa}$, for $\kappa \in K$. Finally,
let the signature $\sigma_K^3$ be the union of
all the $\sigma_\kappa^3$ for $\kappa \in K$. 
Consider each of the exponentially many
consistent completions $\Delta \supseteq \Gamma$ over $\sigma_K^3$, in turn. For each such
$\Delta$,  
check the $\GCD$-formula
\begin{equation*}
\phi \wedge B_K \wedge F_K \wedge I_K \wedge N_K,
\end{equation*}
is satisfiable, and
that none of the ten types of violations of $\kappa$ depicted in Fig.~\ref{fig:tours}(\textsc{ii})--(\textsc{x}) arises, 
by checking that $\Delta$ satisfies the relevant conditions in 
Lemmas~\ref{lem:viol_cond_ii}--\ref{lem:viol_cond_x}. 
(Note that the truth of formulas $N_K$ rules out violations of $\kappa$ of the type
depicted in Fig.~\ref{fig:tours}(\textsc{i}), by Lemma~\ref{lem:viol_cond_i}.)
If, for some $\Delta$, these checks succeed,
we may report that $\phi$ is satisfiable. Otherwise, we may report that it is unsatisfiable.

Anticipating Theorems~\ref{thm:main} and~\ref{thm:mainGeneral}, 
stating that the satisfiability and finite satisfiability
problems for $\GCD$ are in $\dexp$, we have our main result:
\begin{thm}
\label{thm:main_key}
The satisfiability and finite satisfiability problems for $\GCDK$ are in $\dexp$.
\end{thm}

\section{Complexity of $\GCD$: preliminaries}
\label{sec:types}

It remains only to establish that the satisfiability and finite satisfiability 
problems for $\GCD$ are in $\dexp$ (Theorems~\ref{thm:main} and~\ref{thm:mainGeneral}). The proof given here is a
modification of the proof in~\cite{pratt07} that the finite satisfiability
problem for $\GC$ is in $\dexp$. Unfortunately, guarded formulas cannot straightforwardly 
be used to express databases, and the authors know of no natural reduction from (finite) $\GCD$-satisfiability to (finite) $\GC$-satisfiability: the proof, even though its broad outlines are the same, 
has to be re-written from scratch. Many of the following lemmas are taken---modulo inessential reformulation---directly from~\cite{pratt07}.
We begin with a normal-form lemma for $\GC$-formulas.
\begin{lem}[\cite{pratt07}, Lemma~2]
	\label{lem:gc2_nf}
	Let $\psi$ be a $\GC$-formula. We can compute in time polynomial in $\|\psi\|$, a sentence
	\begin{equation}
	\begin{split}
	\phi \: := \: \forall x \, \alpha \land \bigwedge_{1 \leq j \leq n} \forall x \forall y (e_j(x,y) \rightarrow (\beta_j \lor x = y)) \ \ \land  \\
	\hspace{3cm} \bigwedge_{1 \leq i \leq m} \forall x \exists_{=C_i} \, y (o_i(x, y) \land x \neq y),
	\end{split}
	\label{eq:nf}
	\end{equation}
	such that:
	(i) $\alpha$ is a quantifier- and equality-free formula in one variable, $x$;
	(ii) $n$, $m$ are positive integers;
	(iii) $e_j$ is a binary predicate different from $=$, for all $j$ ($1 \leq j \leq n$);
	(iv) $\beta_j$ is a quantifier- and equality-free formula;
	(v) $C_i$ is a positive integer;
	(vi) $o_i$ is a binary predicate other than =, for all $i$ ($1 \leq i \leq m$);
	(vii) $\phi$ entails $\psi$; and
	(viii) any model of $\psi$ can be expanded to a model of $\phi$.
\end{lem}

In the sequel, we fix a $\GCD$-formula $\phi \wedge \Delta$, where $\phi$ is of the form~\eqref{eq:nf},
and $\Delta$ a database, over some signature $\sigma$. As 
we are employing the unique names assumption, we continue to write $c$ instead of $c^\gA$ where the interpretation $\gA$ is clear from context.
We assume that $\sigma$ contains, for each individual constant $c$, a unary predicate 
with the same name, and that $\Delta$ contains
all the formulas $c(c)$ and $\neg c(d)$ where $c$ and $d$ are distinct individual constants. We refer to
these unary predicates as {\em naming predicates} and to the formulas added to $\Delta$ as {\em naming formulas}. Elements realizing a naming predicate $c(\cdot)$ are to be viewed as candidates for the
constant $c$. When building a model of $\phi \land \Delta$, one of these candidates is
chosen to realize the actual constant.
Note that the addition of naming predicates and formulas requires at
most a polynomial increase in the size of $\sigma$ and $\Delta$. We further assume, again for technical reasons, that $\sigma$ contains
$\lceil \log((mC)^2+1) \rceil$ unary predicates not occurring in $\phi$ or $\Delta$. 
We refer to these unary predicates as {\em spare predicates}.
Notice that the number of spare predicates is bounded by a polynomial function of $\|\phi\|$.

A $1$-\textit{type} is a maximal consistent set of non-ground, equality-free literals (over $\sigma$) with 
$x$ as the only free variable; 
a $2$-\textit{type} is a maximal consistent set of non-ground, equality-free  literals (over $\sigma$) with $x$ and $y$ as the only free variables. 
In this article, expressions involving equality do not feature in 1- or 2-types: this includes, in particular, formulas of the forms $x = c$ and $x \neq c$. On the other hand, every individual constant $c$ in $\sigma$ is also a naming predicate, and literals of the forms $c(x)$ and $\neg c(x)$ do occur in 1- and 2-types. Since $\Delta$ is assumed to contain all the naming formulas, if $\gA \models \phi \wedge \Delta$, then $c$ will be the unique database element of $\gA$ whose 1-type contains the literal $c(x)$. We stress however that, if $\gA \models \phi \wedge \Delta$, 
there can be multiple other (non-database) elements of $\gA$ whose 1-type contains the literal $c(x)$.
(These elements can be viewed as candidates that where not chosen to realize $c$ when $\gA$
was built.)
We use the letters $\pi$ and $\tau$, possibly with subscripts,
to range over 1-~and 2-types respectively. Where convenient, we treat a $k$-type ($k=1,2$) as the conjunction of the formulas constituting it. 

For a given structure $\gA$ interpreting $\sigma$ and element $a \in A$, we denote by
$\tpA{a}$ the unique 1-type $\tau$ such that $\gA \models \tau(a)$. In this case, we say that $\tau$ is \textit{realized} by $a$ or that $a$ \textit{realizes} $\tau$. If 
$a$ and $b$ are distinct elements of $A$, we speak similarly of the 2-type $\tpA{a,b}$
\textit{realized} by $a$ and $b$.
Let $\tau$ be a 2-type; we denote by $\tau^{-1}$ the 2-type which is the result of transposing the variables $x$ and $y$ in $\tau$.
Further, we denote by $\tpStart{\tau}$ the result of deleting all literals from $\tau$ involving the variable $y$, and we define $\tpEnd{\tau} = \tpStart{\tau^{-1}}$. We are to think of $\tpStart{\tau}$ as
the `starting point' of $\tau$ and of $\tpEnd{\tau}$ as the `endpoint' of $\tau$. 
Evidently, if $\tau$ is a 2-type, $\gA$ a structure, and $a, b$ distinct elements of $A$ such that $\tpA{a,b} = \tau$, then $\tpA{b,a} = \tau^{-1}$, $\tpA{a} = \tpStart{\tau} $ and
$\tpA{b} = \tpEnd{\tau}$.

By a \emph{vector} we understand an $m$-dimensional vector over $\mathbb{N}$.
We denote the vector $(C_1, \ldots, C_m)$ of counting subscripts occurring in~\eqref{eq:nf} 
by $\bC$ and the $m$-dimensional zero vector by $\bO$. We use 
$\leq$ for the pointwise ordering on vectors: $\bv \leq \bw$ if every component of $\bv$ is less than or equal
to the corresponding component of $\bw$, and $\bv < \bw$ if $\bv \leq \bw$ and $\bv \neq \bw$; similarly for $\geq$ and $>$. Let
$C = \max_{1 \leq i \leq m}\{ C_i \}$. The number of vectors $\bu$ such that $\bu \leq \bC$ is evidently
bounded by $(C + 1)^m$, and thus by an exponential function of $\| \phi \|$.

The following notions make specific reference to our fixed formula $\phi$ given in~\eqref{eq:nf}, and in
particular to the binary predicates $q_1, \dots, q_m$ appearing in it.
Let $\tau$ be a 2-type over $\sigma$. We say that $\tau$ is a \emph{message-type} if $\models \tau \rightarrow o_i(x, y)$, for some $i$ $(1 \leq i \leq m)$. For $\tau$ a message-type, if $\tau^{-1}$ is also a message-type, we
say that $\tau$ is \emph{invertible}; otherwise, $\tau$ is \emph{non-invertible}. Finally, if $\tau$
is a 2-type such that neither $\tau$ nor $\tau^{-1}$ is a message-type, we say that $\tau$ is \emph{silent}.

Let $\gA$ be a structure interpreting $\sigma$. If $a$, $b$ are distinct elements of $A$ such that
$\tpA{a,b} = \tau$ is a message-type, we say, informally, that $a$ {\em sends a message} ({\em of type $\tau$}) {\em to} $b$. We call $\gA$ \emph{chromatic} if the following are true:
\begin{itemize}
  \item For all $a, b \in \gA$ such that $a \neq b$ and $\tpA{a, b}$ is an invertible message-type, we have $\tpA{a} \neq \tpA{b}$.
  \item For all pairwise distinct $a, b, c \in \gA$ such that $\tpA{a, b}$ and $\tpA{b, c}$ are invertible message-types, we have
        $\tpA{a} \neq \tpA{c}$. 
\end{itemize}
In other words, a structure $\gA$ is chromatic if no element sends an invertible message
to an element with the same 1-type as itself, and no element sends invertible messages
to two different elements with the same 1-type as each other.  The following lemma shows that for our fixed formula $\phi \wedge \Delta$ and signature $\sigma$, we may restrict attention
to chromatic models. This lemma relies on the fact that $\sigma$ contains the $\lceil \log((mC)^2+1) \rceil$ 
`spare' predicates not occurring in $\phi$ or $\Delta$.
\begin{lemC}[{\cite[Lemma~3]{pratt07}}]
If $\phi \wedge \Delta$ has a model interpreting $\sigma$, then it has a chromatic model interpreting $\sigma$ over the same domain.
\end{lemC}

Assume, then, that $\gA$ is a chromatic model of $\phi \wedge \Delta$ over $\sigma$. By adding
to $\Delta$ all ground literals over $\sigma$ that are true in $\gA$, we may assume that $\Delta$ is complete over $\sigma$. Notice that, for a complete database $\Delta$, all models of $\Delta$ assign the same 1-type to $c$, and so we may denote this 1-type by $\tp{\Delta}{c}$. If $c$ and $d$ are  distinct individual constants, we define $\tp{\Delta}{c,d}$ analogously.

Let $\Pi = \pi_0, \ldots, \pi_{P-1}$ be an enumeration of the 1-types over $\sigma$. It is clear that $P$ is
a power of 2, thus $p = \log P$ is an integer. We will need to index certain sub-sequences of $\Pi$ using bit-strings. Let $\epsilon$ denote the empty string and define $\Pi_\epsilon = \Pi$
to be the whole sequence. 
Now suppose $\Pi_\s = \pi_j, \dots \pi_{k-1}$ has been defined, where
$\s$ is a bit-string of length less than $p$, and $k-j$ is a power of 2. We define
$\Pi_\sO$ and $\Pi_\sI$ to be the first
and second halves of $\Pi_\s$, respectively:
$$\Pi_\sO = \pi_j, \ldots,  \pi_{(j+k)/2-1}
\quad \text{ and } \quad \Pi_\sI = \pi_{(j+k)/2}, \ldots, \pi_{k-1}.$$ 
When $| \s |= p$, it is clear that $\Pi_\s$
corresponds to exactly one 1-type $\pi_j$: we sometimes denote this type by $\pi_\s$. Note that,
in this case, $\s$ is the binary representation of the subscript $j$.

We now index (sequences of) invertible message-types according to their terminal 1-types using bit-strings as
follows. Let $\Lambda$ be the set of all invertible message-types (over $\sigma$).
Fix any 1-type $\pi$, let $\s$ be any bit-string with $0 \leq | \s | \leq p$, and define
$$\lps = \{\lambda \in \Lambda \mid \tpStart{\lambda} = \pi \text{ and } \tpEnd{\lambda} \in \Pi_\s \},$$
i.e.~$\lps$ is the set
of all invertible message-types `starting' from an element of 1-type $\pi$ and `ending' at an element whose 1-type is in $\Pi_\s$.
Notice that each of the sets $\lps$ will usually contain more than one 2-type. This is true even when $| \s | = p$, as there
are several ways one could `connect' an element of 1-type $\pi$ with another element of 1-type
$\pi'$. For a \emph{chromatic} model, $\gA$, however, we are guaranteed
that no element sends an invertible message-type to an element of the same 1-type and any element $a$ of 1-type $\pi$ can send an
invertible message to no more than one element $b$ of type $\pi' \, (\neq \pi)$. Thus, when $| \s | = p$, there can be
at most one element $b \in A \setminus \{ a \}$ such that $\tpA{a, b} \in \lps$.

In a similar way, we use bit-strings to index sequences of 2-types that are not invertible message-types.
Fix any 1-type $\pi$ and let  $$M_\pi \: = \: \mu_{\pi, 0}, \ldots, \mu_{\pi, Q-1}$$ be an enumeration of
all 2-types $\tau$ with $\tpStart{\tau} = \pi$ that are either non-invertible message-types or silent 2-types. 
In other words, $M_\pi$ is an enumeration of all
2-types $\tau$ with $\tpStart{\tau} = \pi$ such that $\tau^{-1}$ is not a message-type. It follows, then,
that $Q$ is a power of 2,
thus $q = \log Q$ is an integer. Define $M_{\pi, \epsilon} = M_\pi$.
Now suppose $M_{\pi, \t} = \mu_j, \dots \mu_{k-1}$ has been defined, where
$\t$ is a bit-string of length less than $q$, and $k-j$ is a power of 2. We recursively define  $$M_{\pi, \tO} = \mu_j, \ldots, \mu_{(j+k)/2-1}\quad
\text{ and } \quad M_{\pi, \tI} = \mu_{(j+k)/2}, \ldots, \mu_{k-1}.$$ 
Note that if $| \t | = q$, then
$M_{\pi, \t}$ contains a single 2-type $\mu_{\pi, j}$, which we often write as $\mu_{\pi, \t}$ for
convenience.

\section{Reduction to linear programming}
\label{sec:transform}

In this section, we show how to transform our given $\GCD$-formula, $\phi \wedge \Delta$,
into a system of exponentially many linear inequalities. The variables in these inequalities 
represent the frequency with which
certain (local) configurations are realized in a putative model. We shall show that this system of inequalities has
a non-negative integer solution if and only if $\phi \wedge \Delta$ is finitely satisfiable.

Let $\gA$ be a model of $\phi$ interpreting the signature $\sigma$. Let  $a \in A$ be an element with $\tpA{a}= \pi$,
and $\s$ a bit-string of length at most $p$. Define the $\s$-\emph{spectrum} of $a$ in $\gA$, denoted 
$\spA{\s}{a}$, to be the $m$-element vector $\bv = (v_1, \ldots, v_m)$ where, for $1 \leq i \leq m$,
$$v_i = |\{ b \in A \setminus \{ a \} : \gA \models o_i(a, b) \text{ and } \tpA{a, b} \in \lps \}|.$$
Likewise, if $\t$ is a bit-string with $| \t | < q$, define the $\t$-\emph{tally} of $a$, denoted 
$\tlA{\t}{a}$, to be the $m$-element vector $\bw = (w_1, \ldots, w_m)$ where, for $1 \leq i \leq m$,
$$w_i = |\{ b \in A \setminus \{ a \} : \gA \models o_i(a, b) \text{ and } \tpA{a, b} \in \mpt \}|.$$
Informally, $\spA{\s}{a}$ records the number of $o_i$-messages ($1 \leq i \leq m$)  sent
by the element $a$, and whose types belong to $\lps$. (Note that these message-types are all invertible.)
In particular, 
$\spA{\epsilon}{a}$ records the number of invertible $o_i$-messages ($1 \leq i \leq m$)  sent
by the element $a$.
Similarly, $\tlA{\t}{a}$
records the number of $o_i$-messages ($1 \leq i \leq m$)  sent
by the element $a$, and whose types belong to $\mpt$. (Note that these message-types are all non-invertible.)
In particular, 
$\tlA{\epsilon}{a}$ records the number of non-invertible $o_i$-messages ($1 \leq i \leq m$)  sent
by the element $a$.

Let $\s$ be any
bit-string with $| \s | < p$ and fix a 1-type $\pi$. For a given structure $\gA$, each vector $\bv$ specifies a 
set of elements of $\gA$, i.e.~the set of elements of type $\pi$ that have
$\s$-spectrum $\bv$. The following lemma encapsulates the observation that this set can also be characterized
as the union of sets of elements of type $\pi$ whose $\sO$- and $\sI$-spectra add up to $\bv$. The same idea
applies to tallies. 

\begin{lem}
\label{spec_lem}
Let $\gA$ be a finite model of $\phi$, and let $a \in A$ with $\tpA{a} = \pi$. Let $\s$, $\t$ be any bit-strings with $| \s | < p$ and $| \t | < q$. Then, the
following equations hold:
\begin{align}
  \label{spec_lem_eq1}
  \spA{\epsilon}{a} + \tlA{\epsilon}{a} &= \bC \\
  \label{spec_lem_eq2}
  \spA{\sO}{a} + \spA{\s1}{a} &= \spA{\s}{a} \\
  \label{spec_lem_eq3}
  \tlA{\tO}{a} + \tlA{\tI}{a} &= \tlA{\t}{a}
\end{align}
\end{lem}

\begin{proof}
Equation~(\ref{spec_lem_eq1}) is immediate from the definition of spectra and tallies and the normal form in Lemma~\ref{lem:gc2_nf}. For Equation~(\ref{spec_lem_eq2}),
notice that the set $\lps$ can be partitioned into two subsets $\Lambda_{\pi, \sO}$ and $\Lambda_{\pi, \sI}$,
and this induces a partition of the outgoing message-types from $a$ ($\im$); 
Equation~(\ref{spec_lem_eq2}) is then evident. Likewise for Equation~(\ref{spec_lem_eq3}).
\end{proof}

Let $\tau$ be any 2-type. With $\tau$ we associate an $m$-dimensional vector $\bC_\tau$,
whose $i$th component is given by $$(\bC_\tau)_i =
\left\{\begin{array}{l l}
  1 & \text{if } \models \tau \rightarrow o_i(x,y), \\
  0 & \text{otherwise}.
\end{array}\right.$$
Let  $\gA$ be a finite model of $\phi$, $a \in A$, and
$\t$ a bit-string with $| \t | = q$. Now, consider the (only)
2-type $\mu_{\pi, \t}$ in $M_{\pi, \t}$ and, if $\mu_{\pi, \t}$ is non-silent, let $n$ be the number of messages of type $\mu_{\pi, \t}$ sent
by $a$, i.e.~$n$ is the number of elements
$b \in A \setminus \{ a \}$ such that $\tpA{a, b} = \mu_{\pi, \t}$. It is clear, then, that $\tlA{\t}{a} = n \cdot \bC_{\mu_{\pi, \t}}$. On
the other hand, if $\mu_{\pi, \t}$ is silent, we have $\tlA{\t}{a} = \bC_{\mu_{\pi, \t}} = \bO$.
Let $\s$ be a bit-string with $| \s | = p$, and suppose further that $\gA$ is chromatic.
For each element $a \in A$ with 1-type $\pi$ and non-zero $\s$-spectrum, there is a unique $\lambda \in \lps$ (depending on $a$)
such that $a$ sends a message of type $\lambda$, and hence has $\s$-spectrum
$\bC_\lambda$.

\begin{lem}
\label{tmp_lem1}
Let $\gA$ be a chromatic model of $\phi$, $a \in A$, $\pi$ a 1-type, and $\s$ a bit-string with $| \s | = p$. If $\tpA{a} = \pi$
and $\spA{\s}{a} \neq \bO$, then there exists $\lambda \in \lps$ with $\spA{\s}{a} = \bC_\lambda$ such that $a$ sends a message of type
$\lambda$ to some $b \in A \setminus \{ a \}$. Conversely, if there exists $\lambda \in \lps$ such that $a$ sends a message of type
$\lambda$ to some $b \in A \setminus \{ a \}$, then $\tpA{a} = \pi$ and $\spA{\s}{a} = \bC_\lambda$.
\end{lem}

\begin{proof}
Suppose $\tpA{a} = \pi$ and $\spA{\s}{a} \neq \bO$. As discussed previously, there exists a unique $b \in A \setminus \{ a \}$ such
that $\lambda = \tpA{a, b} \in \lps$. Clearly, then, $\spA{\s}{a} = \bC_\lambda$. Conversely, suppose $a$ sends a message of type
$\lambda \in \lps$ to some element $b \in A \setminus \{ a \}$. Evidently, then, $\tpA{a} = \pi$ and $b$ is unique, thus $\spA{\s}{a}
= \bC_\lambda$.
\end{proof}

Henceforth, given a structure $\gA$,
we will denote the set of elements in the universe $A$ of $\gA$ having 1-type $\pi$ by $A_\pi$, that is $A_\pi = \{ a \in A \mid \tpA{a} = \pi \}$. 
We now proceed to transform the question of whether $\phi \wedge \Delta$ is finitely 
satisfiable into the question of whether a certain
system of linear equations/inequalities has a solution over $\mathbb{N}$.
The solutions of these equations count how often
various local configurations appear in a model. These configurations are:
\begin{itemize}[noitemsep]
  \item realizations of each invertible message-type $\lambda$;
  \item elements of 1-type $\pi$ having $\s$-spectrum $\bu$, for all vectors $\bu \leq \bC$;
  \item elements of 1-type $\pi$ having $\t$-tally $\bu$, for all vectors $\bu \leq \bC$;
  \item elements of 1-type $\pi$ whose $\s$-spectrum $\bu$ is obtained as the sum of an $\sO$-spectrum $\bv$ and an $\sI$-spectrum $\bw$, for $\bv, \bu, \bw \leq \bC$ and for all $\s$ with $| \s | < p$;
  \item elements of 1-type $\pi$ whose $\t$-tally $\bu$ is obtained as the sum of a $\tO$-tally $\bv$ and a $\tI$-tally $\bw$, for $\bv, \bu, \bw \leq \bC$ and for all $\t$ with $| \t | < q$.
\end{itemize}
To each of those configurations we associate a variable which is intended to capture how many times it appears in a model. These variables and
the properties that they capture can be seen in Table~\ref{table:var_meanings}. Unless specified otherwise, the ranges of the subscripts of these
variables are as follows: $\pi$ ranges over all 1-types in $\Pi$,
$\lambda$ ranges over all invertible message-types, $\s$ ranges over all bit-strings
with $| \s | \leq p$, $\t$ ranges over all bit-strings with $| \t | \leq q$, and
$\bu, \bv, \bw$ vary over all vectors $\leq \bC$. (Similarly for their primed
counterparts $\pi', \lambda', \s', \t', \bu', \bv'$ and $\bw'$.) That is,
we have one variable $x_\lambda$ for each invertible message-type $\lambda$, one variable $y_{\pi, \s, \bu}$ for each possible combination of $\pi \in \Pi$, $\s$ with $| \s | \leq p$ and $\bu \leq \bC$, and so on.

\begin{table}
\def\arraystretch{1.4}  
\centering
\begin{tabularx}{.98\textwidth}{ | X | l | }  
  \hline
    Variable                    &   Intended meaning of its value \\
  \hline
    $x_\lambda$                 &   $|\{ a \in A :$ for some $b \in A \setminus \{ a \}$,~$\tpA{a, b} = \lambda \}|$  \\
    $y_{\pi, \s, \bu}$          &   $|\{ a \in A_\pi : \spA{\s}{a} = \bu \}|$ \\
    $z_{\pi, \t, \bu}$          &   $|\{ a \in A_\pi : \tlA{\s}{a} = \bu \}|$ \\
    $\hy_{\pi, \s, \bv, \bw}$   &   $|\{ a \in A_\pi : \spA{\sO}{a} = \bv$ and $\spA{\sI}{a} = \bw$, whenever $| \s | < p \}|$ \\
    $\hz_{\pi, \t, \bv, \bw}$   &   $|\{ a \in A_\pi : \tlA{\tO}{a} = \bv$ and $\tlA{\tI}{a} = \bw$, whenever $| \t | < q \}|$ \\
  \hline
\end{tabularx}
\vspace{8pt}
\caption{Variables and their intended meanings, for a finite model $\gA$
of $\Delta \cup \{ \phi \}$. Recall, $A_\pi = \{ a \in A \mid \tpA{a} = \pi \}$.}
\label{table:var_meanings}
\end{table}

We now write a collection of constraints that a given structure $\gA$ has to satisfy for it to be a model of $\phi \wedge \Delta$. For ease of reading,
we divide these constraints into four groups: the first three are imposed by the formula $\phi$,
and the fourth, by the database $\Delta$.
Let $\cE_1$ be the following set of constraints, where  indices vary
over their standard ranges, but with $\s$ and $\t$ subject to the additional constraint 
that $| \s | < p$ and $| \t | < q$:
\begin{align}
  \label{e1_c0}  y_{\pi, \epsilon, \bu}  \: &=    \: z_{\pi, \epsilon, \bC - \bu} \\
  \label{e1_c1}  y_{\pi, \s, \bu}   \: &= \:  \sum \{ \hy_{\pi, \s, \bv', \bw'} \mid \bv' + \bw' = \bu \} \\
  \label{e1_c2}  z_{\pi, \t, \bu}   \: &= \:  \sum \{ \hz_{\pi, \t, \bv', \bw'} \mid \bv' + \bw' = \bu \} \\
  \label{e1_c3}  y_{\pi, \sO, \bv}  \: &= \:  \sum \{ \hy_{\pi, \s, \bv, \bw'} \mid \bv + \bw' \leq \bC \} \\
  \label{e1_c4}  y_{\pi, \sI, \bw}  \: &= \:  \sum \{ \hy_{\pi, \s, \bv', \bw} \mid \bv' + \bw \leq \bC \} \\
  \label{e1_c5}  z_{\pi, \tO, \bv}  \: &= \:  \sum \{ \hz_{\pi, \t, \bv, \bw'} \mid \bv + \bw' \leq \bC \} \\
  \label{e1_c6}  z_{\pi, \tI, \bw}  \: &= \:  \sum \{ \hz_{\pi, \t, \bv', \bw} \mid \bv' + \bw \leq \bC \}  \\
  \label{e1_c7}  1                 \: &\leq \: \sum \{ y_{\pi', \epsilon, \bu'} \mid \pi' \text{ a 1-type, } \bu' \leq \bC \}
\end{align}

\begin{lem}[\cite{pratt07}, Lemma 7]
\label{lem:E1}
Let $\gA$ be a finite model of $\phi$. The constraints $\cE_1$ are satisfied when the variables take the values specified in Table~\ref{table:var_meanings}.
\end{lem}
\begin{proof}
For~\eqref{e1_c0}, we see from~\eqref{spec_lem_eq1} that every element contributing to the count $y_{\pi, \epsilon, \bu}$ contributes to $z_{\pi, \epsilon, \bC - \bu}$, and vice versa. For~\eqref{e1_c1},
we see from~\eqref{spec_lem_eq2} that every element contributing to the count $y_{\pi, \s, \bu}$ 
must contribute to $\hy_{\pi, \s, \bv', \bw'}$ for some pair of vectors $\bv'$, $\bw'$ summing to $\bu$,  and vice versa.
The other equations are similar. The inequality~\eqref{e1_c7}
in effect states that $A$ is non-empty.
\end{proof}

Recalling our fixed formula $\phi$ given in~\eqref{eq:nf},
let $\tau$ be any 2-type. We say that $\tau$ is \emph{forbidden} if the following formula is unsatisfiable:
\begin{gather*}
\label{forbidden}
\bigwedge \tau \land \alpha(x) \land \alpha(y) \,\land \hspace{15em} \\
\hspace{3em}
\bigwedge_{1 \leq j \leq n} \big((e_j(x,y) \rightarrow \beta_j(x, y)) \land (e_j(y,x) \rightarrow \beta_j(y, x))\big).
\end{gather*}
Evidently, no forbidden 2-type can be realized in any model of $\phi$.

Let $\cE_2$ be the following set of constraints, where indices
vary
over their standard ranges, but with $\s$ and $\t$ subject to the additional constraint 
that $| \s | = p$ and $| \t | = q$:
\begin{align}
  \label{e2_c1}  y_{\pi, \s, \bu}   \: &=  \:  \sum \{ x_{\lambda'} \mid \lambda' \in \lps \text{ and } \bC_{\lambda'} = \bu \} \\
  \label{e2_c2}  x_{\lambda^{-1}}   \: &=  \:  x_\lambda  \\
  \label{e2_c3}  x_\lambda          \: &=  \: 0, \quad \text{whenever } \tpStart{\lambda} = \tpEnd{\lambda} \\
  \label{e2_c4}  x_\lambda          \: &=  \: 0, \quad \text{whenever } \lambda \text{ is forbidden} \\
  \label{e2_c5}  z_{\pi, \t, \bu}   \: &=  \: 0, \quad \text{whenever } \mu_{\pi, \t} \text{ is forbidden} \\
  \label{e2_c6}  z_{\pi, \t, \bu}   \: &=  \: 0, \quad \text{whenever } \bu \text{ is not a scalar multiple of } \bC_{\mu_{\pi, \t}}
\end{align}

\begin{lem}[\cite{pratt07}, Lemma 8]
\label{lem:E2}
Let $\gA$ be a finite, chromatic model of $\phi$. The constraints $\cE_2$ are satisfied when the variables take the values specified in Table~\ref{table:var_meanings}.
\end{lem}
\begin{proof}
For~\eqref{e2_c1}, we observed in Lemma~\ref{tmp_lem1} that every element $a$ contributing to the count 
$y_{\pi, \s, \bu}$
emits exactly one invertible message whose type $\lambda$ lies in $\lps$, and that $\bu = \spA{\s}{a} = \bC_\lambda$. The other equations are obvious.
\end{proof}

Let $\cE_3$ be the following set of constraints, where  indices vary
over their standard ranges, but with $\t$ and $\bu$ subject to the additional constraint 
that $| \t | = q$, and $\bu \neq \bO$:
\begin{equation}
  \label{e3_c1}  z_{\pi, \t, \bu} > 0 \: \Rightarrow \:  \sum \{ y_{\pi', \epsilon, \bu'} \mid \pi' = \tpEnd{\mu_{\pi, \t}}
                                                            \text{ and } \bu' \leq \bC \} > 0
\end{equation}
In effect, this inequality states that, if a non-invertible message-type is realized in $\gA$, then
it must have a `landing-site' of the appropriate 1-type. Thus, we have:
\begin{lemC}[{\cite[Lemma 9]{pratt07}}]
\label{lem:E3}
Let $\gA$ be a finite model of $\phi$. The constraints $\cE_3$ are satisfied when the variables take the values specified in Table~\ref{table:var_meanings}.
\end{lemC}

We now turn our attention to the constraints that $\Delta$ enforces. (Note that these constraints do not occur at all in~\cite{pratt07}.)
These constraints feature the spectra and tallies of our database elements. The difficulty is that, since we do not know how elements in the database are related to those outside it in some putative model $\gA$, these quantities are unknown. Moreover, 
although the number of database elements is bounded by the size of $\Delta$, the number of values
$\sp{\s}{\gA}{c}$ and $\tl{\t}{\gA}{c}$ (for each constant $c \in D$)
is exponential in the size of $\phi$, and thus too large to simply guess in order to obtain the desired $\dexp$-complexity bound. We therefore need to perform this guessing judiciously.

Suppose $\gA \models \phi \wedge \Delta$ is chromatic:
for each constant $c \in D$, let
\begin{itemize}
\item $S_c$ be the set of strings $\s$ of length $p$ such that $c$ sends an invertible message in $\gA$ to some database element $b$ of 1-type $\pi_\s$; and
\item $T_c$ be the set of strings $\t$ of length $q$ such that $c$ sends a non-invertible message in $\gA$ of type
$\mu_\t$ to some database element.
\end{itemize}
Thus, $S_c$ is the set of indices of (terminal 1-types of) invertible messages sent by $c$ within the database, and
$T_c$ is the set of indices of non-invertible messages sent by $c$ within the database. Obviously, the cardinalities of both sets are bounded by the size of the database.

If $\x$ is any bit-string we denote by $\mathrm{seg}(\x)$ the set of all proper initial segments of $\x$
(thus, $\epsilon \in \mathrm{seg}(\x)$, but $x \not \in \mathrm{seg}(\x)$), and we write
\begin{equation*}
\Ccom_\x = \{\epsilon\} \cup \{\y\bitb \mid \text{$\y \in \mathrm{seg}(\x)$ and } \bitb \in \{\xO,\xI\} \}.
\end{equation*}
Alternatively, $\Ccom_\x$ is the set of all proper initial segments of $\x$ together with their extensions by a single bit.
For example, 
\begin{equation}
\Ccom_{\xO\xI\xO\xI} = \{\epsilon, \xO, \xI, \xO\xO, \xO\xI, \xO\xI\xO, \xO\xI\xI, \xO\xI\xO\xO, \xO\xI\xO\xI \}.
\end{equation}
(The notation is an allusion to the similar notion of C-command in transformational linguistics.) Notice that
$|\Ccom_\x| = 2|\x|+1$. Now define, for any individual constant $c$,
\begin{align*}
\CcomS_c & = \{\epsilon\} \cup \bigcup \{ \Ccom_\s \mid \s \in S_c\}\\
\CcomT_c & = \{\epsilon\} \cup \bigcup \{ \Ccom_\t \mid \t \in T_c\}.
\end{align*}
The idea is simple: for an individual constant $c$ of type $\pi$, 
$\CcomS_c$ collects
every string $\s$ of length $p$ such that $c$ sends an invertible message of type $\lambda \in \lps$
to a database element, together with all of $\s$'s proper prefixes and the extensions of those proper prefixes by a single bit. Notice that $\CcomS_c$ is always taken to include the empty string
even if $c$ sends no invertible messages to any database element: this stipulation avoids some otherwise tedious special cases.
Similarly, $\CcomT_c$ collects
every string $\t$ of length $q$ such that $c$ sends a non-invertible message of type $\mu \in \mpt$ to a database element, together with all of $\t$'s proper prefixes and the extensions of those proper prefixes by a single bit; again, we always add the string $\epsilon$.
Both of these sets are polynomially bounded in the size of $\phi \wedge \Delta$; we shall need to guess $\sp{\s}{\gA}{c}$
only for $\s \in \CcomS_c$ and $\tl{\t}{\gA}{c}$
only for $\t \in \CcomT_c$.  We remark that, by construction, for any $\s$ with $0 \leq |\s| <p$,
$\sO \in \CcomS_c$ if and only if $\sI \in \CcomS_c$, and $\sO \in \CcomS_c$ implies $\s \in \CcomS_c$; similarly 
with $\CcomT_c$.

Keeping our chromatic structure $\gA$ fixed for the moment, for any individual constant $c$, with $\tp{\gA}{c} = \pi$, and any
$\s \in \CcomS_c$,
$\t \in \CcomT_c$, define:
\begin{align}
\gamma^c_{\pi,\s} & = \spA{\s}{c} & 
\delta^c_{\pi,\t} & = \tlA{\t}{c}.
\label{eq:IansEq}
\end{align}
That is, $\gamma^c_{\pi,\s}$ is the $\s$-spectrum in $\gA$ of the database element named $c$, and 
$\delta^c_{\pi,\t}$ is its $\t$-tally. As a special case of~\eqref{spec_lem_eq1}--\eqref{spec_lem_eq3},
for all bit-strings $\s$, $\t$ with $\sO \in \CcomS_c$ and $\tO \in \CcomT_c$:\\
\begin{align}
\label{eq:gammadelta} \gamma^c_{\pi, \epsilon} + \delta^c_{\pi, \epsilon} & = \bC\\
\label{eq:gamma} \gamma^c_{\pi, \sO} + \gamma^c_{\pi, \sI} & = \gamma^c_{\pi, \s}\\
\label{eq:delta} \delta^c_{\pi, \tO} + \delta^c_{\pi, \tI} & = \delta^c_{\pi, \t}.
\end{align}

A moment's thought shows that
$\Delta$ imposes additional constraints on the quantities $\gamma^c_{\pi,\s}$ $\delta^c_{\pi,\t}$.
For consider an individual constant $c$ and  bit-string $\s$ of length $p$, and suppose
there exists an individual constant $d$ such that
$\tp{\Delta}{c,d} = \lambda \in \Lambda_{\pi,\s}$. (Since $\gA$ is chromatic,
there can be at most one such $d$.) Then
$\gamma^c_{\pi,\s} = \spA{\s}{c} = \bC_\lambda$.
Similarly, if, for an individual constant $c$ and bit-string $\t$ of length $q$,
where $\mu= \mu_{\pi,\t}$ is a non-invertible message-type,
there exist $k$ individual constants $d$, distinct from $c$, with
$\tp{\Delta}{c,d} = \mu_{\pi,\t}$,  then, writing $\mu$ for $\mu_{\pi,\t}$, we have $\delta^c_{\pi,\t} = \tlA{\t}{c} \geq k \cdot \bC_{\mu}$.
As we might say, $\Delta$ has to be {\em compatible} with the systems $\gamma^c_{\pi,\s}$ and
$\delta^c_{\pi,\t}$, in the obvious sense of not requiring that $c$ sends more messages than these
vectors allow. 

We need one further group of constraints. 
For each invertible type $\lambda$, let $\eta_\lambda$ be the number of elements in
$\Delta$ that send an invertible message of type $\lambda$ to another database element. 
Of course, these numbers can be read directly from $\Delta$.
Now let $\cE_4$ be the following classes of constraints, where $\lambda$ 
varies as usual, $\pi$ ranges over the set of 1-types
realized in the database, $c$ ranges over the set of individual constants, 
$\s$ ranges over the set of strings $\s$ such that $\sO \in \CcomS_c$ 
and $\t$ over the set of strings $\t$ such that $\tO \in \CcomT_c$:
\begin{align}
  \label{e4_c1}  x_\lambda   \: &\geq  \:  \eta_\lambda \\
  \label{e4_c2}  y_{\pi, \epsilon, \bu}  \: &\geq  \:  1,
    \quad \text{when } \bu = \gamma^c_{\pi, \epsilon} \text{ and } \pi = \tp{\Delta}{c} \\
  \label{e4_cTmp}  z_{\pi, \epsilon, \bu}  \: &\geq  \:  1,
    \quad \text{when } \bu = \delta^c_{\pi, \epsilon} \text{ and } \pi = \tp{\Delta}{c} \\
  \label{e4_c3}  \hat{y}_{\pi, \s, \bv, \bw}   \: &\geq  \:  1,
     \quad \text{when } \bv = \gamma^c_{\pi, \sO}, \, \bw = \gamma^c_{\pi, \sI},
     \text{ and } \pi = \tp{\Delta}{c} \\
  \label{e4_c4}  \hat{z}_{\pi, \t, \bv, \bw}   \: &\geq  \:  1,
     \quad \text{when } \bv = \delta^c_{\pi, \tO}, \, \bw = \delta^c_{\pi, \tI},
     \text{ and } \pi = \tp{\Delta}{c}
\end{align}

\begin{lem}
\label{ortho}
Let $\gA$ be a finite model of $\phi \land \Delta$. The constraints $\cE_4$ are satisfied when the variables take the values specified in Table~\ref{table:var_meanings}.
\end{lem}

\begin{proof}
The constraints~(\ref{e4_c1}), (\ref{e4_c2}) and (\ref{e4_cTmp}) are evident.
For (\ref{e4_c3}), let $c \in \cO$ be a constant, with 1-type $\pi = \tp{\Delta}{c}$. 
Pick any $\s$ such that $\sO \in \CcomS_c$. (Hence, also, $\sI, \s \in \CcomS_c$.)
By definition, $\gamma^c_{\pi, \sO}$ and $\gamma^c_{\pi, \sI}$ are the $\sO$- and
$\sI$-spectrum (respectively) of the (database) element $c \in A$. But then,
referring to Table~\ref{table:var_meanings}, $c$ is one of the elements `recorded'
by the variable $\hat{y}_{\pi, \s, \bv, \bw}$, for $\bv = \gamma^c_{\pi, \sO}$ and
$\bw = \gamma^c_{\pi, \sI}$. Consequently, $\hat{y}_{\pi, \s, \bv, \bw} \geq 1$,
when $\bv = \gamma^c_{\pi, \sO}$ and $\bw = \gamma^c_{\pi, \sI}$. Thus, the
constraints~(\ref{e4_c3}) are satisfied. Similarly for~(\ref{e4_c4}).
\end{proof}

Let $\cE = \cE_1 \cup \cE_2 \cup \cE_3 \cup \cE_4$. Note that the size $\|\cE\|$ of $\cE$ is bounded above by an exponential function of $\| \phi \|$.
\begin{lem}
\label{lma:orthoPlus}
Let $\phi$ and $\Delta$ be as above.
If $\phi \wedge \Delta$ is finitely satisfiable, then $\cE$ has a solution over $\mathbb{N}$.
\end{lem}
\begin{proof}
Lemmas~\ref{lem:E1}--\ref{ortho}.
\end{proof}

\section{Obtaining a Model from the Solutions}
\label{sec:main_result}

In the previous section, we established a system of constants $\gamma^c_{\pi, \s}$, $\delta^c_{\pi, \t}$ and $\eta_\lambda$, and constructed a
system $\cE$ of linear equations and inequalities featuring these constants. We showed that, if $\phi \wedge \Delta$ has a finite model $\gA$,
then, by interpreting the constants $\gamma^c_{\pi, \s}$, $\delta^c_{\pi, \t}$ and $\eta_\lambda$,
with respect to $\Delta$ and $\gA$ as suggested, we find that the
$\gamma^c_{\pi, \s}$, $\delta^c_{\pi, \t}$ satisfy \eqref{eq:gammadelta}--\eqref{eq:delta}, and, moreover, that $\cE$ has a solution over $\mathbb{N}$. In this section, we establish a converse result. Assuming that the constants $\gamma^c_{\pi, \s}$ and $\delta^c_{\pi, \t}$ satisfy \eqref{eq:gammadelta}--\eqref{eq:delta},
and that the $\gamma^c_{\pi, \s}$, $\delta^c_{\pi, \t}$ and $\eta_\lambda$ are compatible with $\Delta$ as described,
we construct, from any solution of $\cE$ over $\mathbb{N}$, a finite model of $\phi \wedge \Delta$.

Our strategy is as follows. We start with sets of elements $A_\pi$ of the right cardinality, for each 1-type $\pi$, and
gradually build the message-types
that those elements `want' to send. Fix some 1-type $\pi$ and let $A_\pi$ be a set with cardinality
$$|A_\pi| = \sum \{ y_{\pi, \epsilon, \bu'} \mid \bu' \leq \bC \}.$$ Think of $A_\pi$ as
the set of elements that `want' to have 1-type $\pi$. We shall define functions $\bf_{\pi, \s}$ and $\bg_{\pi, \t}$ that give us the spectra and tallies
for each element $a \in A_\pi$. Think of $\bf_{\pi, \s}(a)$ as the $\s$-spectrum that $a$ `wants' to have and $\bg_{\pi, \t}(a)$ as the $\t$-tally
that $a$ `wants' to have (when a model is eventually built). For those functions to agree with the solutions of the previous system of constraints, we shall ensure that
\begin{align}
  \label{tmp1} | \bf^{-1}_{\pi, \s}(\bu)| &= y_{\pi, \s, \bu} \\
  \label{tmp2} | \bg^{-1}_{\pi, \t}(\bu)| &= z_{\pi, \t, \bu}.
\end{align}
Furthermore, we shall ensure that, for all $a \in A_\pi$,
\begin{align}
  \label{tmp3} \bf_{\pi, \epsilon}(a) + \bg_{\pi, \epsilon}(a) &= \bC \\
  \label{tmp4} \bf_{\pi, \sO}(a) + \bf_{\pi, \sI}(a) &= \bf_{\pi, \s}(a) \\
  \label{tmp5} \bg_{\pi, \tO}(a) + \bg_{\pi, \tI}(a) &= \bg_{\pi, \t}(a).
\end{align}
Finally, for each individual constant $c$, setting $\pi = \tp{\Delta}{c}$, we shall ensure the existence of
an element $b_c \in A_\pi$ such that, for all $\s \in \CcomS_c$ and all $\t \in \CcomT_c$,
\begin{align}
\label{tmpNew1} \bf_{\pi, \s}(b_c) = \gamma_{\pi, \s}^c\\
\label{tmpNew2} \bg_{\pi, \t}(b_c) = \delta_{\pi, \t}^c.
\end{align}
The idea is that such an element $b_c$ can be used to realize the constant $c$, when a model is
eventually built. 

The following lemma guarantees that the above requirements can be satisfied.
\begin{lem}
\label{tech_lem}
Suppose that $x_\lambda$, $y_{\pi, \s, \bu}$, $z_{\pi, \t, \bu}$, $\hy_{\pi, \s, \bv, \bw}$, $\hz_{\pi, \t, \bv, \bw}$ are (classes of) natural numbers satisfying 
the constraints $\cE$. Fix any 1-type $\pi \in \Pi$, 
and let $A_\pi$ be a set of cardinality $\sum \{ y_{\pi, \epsilon, \bu'} \mid \bu' \leq \bC \}$. Then
there exists a system of functions on $A_\pi$ $$\bf_{\pi, \s} : A_\pi \rightarrow \{ \bu \mid \bu \leq \bC \} \quad \text{and} \quad \bg_{\pi, \t} : A_\pi
\rightarrow \{ \bu \mid \bu \leq \bC \},$$ for each bit-string $\s$ with $| \s | \leq p$ and $\t$ with $| \t | \leq q$,   
satisfying the following conditions:
\begin{enumerate}[label=(\roman{enumi}), noitemsep]
  \item Equations (\ref{tmp1}) and (\ref{tmp2}) hold for all vectors $\bu \leq \bC$;
  \item Equations (\ref{tmp3})--(\ref{tmp5}) hold for all $a \in A_\pi$ and all $\s$, $\t$ with
  $| \s | < p$ and $| \t | < q$.
\end{enumerate}
Furthermore, if $\pi = \tp{\Delta}{c}$ for some individual constant $c$, then
there exists in addition an element $b_c \in A_\pi$ such that
\begin{enumerate}[label=(\roman{enumi}), noitemsep]
\setcounter{enumi}{2}
  \item Equations (\ref{tmpNew1})--(\ref{tmpNew2}) hold for all $\s \in \CcomS_c$
  and $\t \in \CcomT_c$.
\end{enumerate}
\end{lem}
\begin{proof}
Decompose $A_\pi$ into pairwise disjoint sets $A_{\bu}$ of cardinality $|A_{\bu}| = y_{\pi, \epsilon, \bu}$, for each vector $\bu \leq \bC$. Note that, since $y_{\pi, \epsilon, \bu}$ might be zero, some of these sets may be empty.

We construct the functions $\bf_{\pi, \s}$, where $0 < |\s| \leq p$ by induction on $\s$.
Suppose $\s = \epsilon$; for each $\bu \leq \bC$, and for all $a \in A_{\bu}$, set
$\bf_{\pi, \epsilon}(a) = \bu$ and $\bg_{\pi, \epsilon}(a) = \bC - \bu$. These assignments
clearly satisfy (\ref{tmp1}) and (\ref{tmp2}), keeping in mind the constraints~(\ref{e1_c0}).
Now suppose $\pi = \tp{\Delta}{c}$, for some individual constant $c$.
By \eqref{e4_c2}, and setting $\bu$ to be the vector $\gamma_{\pi, \epsilon}^c$, we have $A_{\bu} \neq \emptyset$. Choose $b_c$ to be any element in $A_{\bu}$. This immediately secures~\eqref{tmpNew1} for
$\s = \epsilon$, and given the condition~\eqref{eq:gammadelta} relating $\gamma_{\pi, \epsilon}^c$ and
$\delta_{\pi, \epsilon}^c$, also~\eqref{tmpNew2}.

Now, assume that $\bf_{\pi, \s}$ ($0 \leq |\s| < p$) has been defined and satisfies~(\ref{tmp1}) and~\eqref{tmpNew1}. For every vector $\bu \leq \bC$, decompose the set $\bf_{\pi, \s}^{-1}(\bu)$ into subsets
$A_{\bv, \bw}$ with cardinality $|A_{\bv, \bw}| = \hy_{\pi, \s, \bv, \bw}$, where
$\bv$, $\bw$ range over all vectors $\leq \bC$ such that $\bv + \bw = \bu$.
This is possible from the constraints~(\ref{e1_c1}) and Equation~(\ref{tmp1}).
Set
\begin{equation*}
\bf_{\pi, \sO}(a) = \bv \quad \text{and} \quad \bf_{\pi, \sI}(a) = \bw,
\end{equation*} 
for all $a \in A_{\bv, \bw}$. Notice that Equation~(\ref{tmp4}) holds as required.

To see that $\bf_{\pi, \sO}$ and $\bf_{\pi,\sI}$ both satisfy
Equation~\eqref{tmp1}, note that $\bf_{\pi,\sO}(a) = \bv$ if and
only if, for some vector $\bw'$ such that $\bv +
\bw' \leq \bC$, $a \in A_{\bv,
\bw'}$. Similarly, $\bf_{\pi,\sI}(a) = \bw$ if and only if,
for some vector $\bv'$ such that $\bv' + \bw \leq
\bC$, $a \in A_{\bv', \bw}$. That is,
\begin{eqnarray*}
\bf_{\pi,\sO}^{-1}(\bv) & = & 
  \bigcup \{ A_{\bv,\bw'} \mid 
     \bv + \bw' \leq \bC \} \\
\bf_{\pi,\sI}^{-1}(\bw) & = & 
  \bigcup \{ A_{\bv',\bw} \mid 
     \bv' + \bw \leq \bC \},
\end{eqnarray*}
with the collections of sets on the respective right-hand sides being
pairwise disjoint.  By the constraints~\eqref{e1_c3}--\eqref{e1_c4},
together with the fact that $|A_{\bv,\bw}| =
\hat{y}_{\pi,s,\bv,\bw}$ for all $\bv,\bw$, we
have:
\begin{eqnarray*}
| \bf_{\pi,s0}^{-1}(\bv) | & = & y_{\pi,s0,\bv}\\
| \bf_{\pi,s1}^{-1}(\bw) | & = & y_{\pi,s1,\bw},
\end{eqnarray*}
which establishes \eqref{tmp1} for the functions
$\bf_{\pi,\sO}$ and $\bf_{\pi,\sI}$.

It remains only to secure~\eqref{tmpNew1}
for the functions $\bf_{\pi,\sO}$ and $\bf_{\pi,\sI}$ in the case where $\pi = \tp{\Delta}{c}$, for some individual constant $c$, and
$\sO$ (and hence $\sI$ and $\s$) is in $\CcomS_c$. By inductive hypothesis, $\bf_{\pi, \s}(b_c) = \gamma^c_{\pi,\s}$, i.e., 
$b_c \in \bf_{\pi, \s}^{-1}(\bu)$ where
$\bu$ is the vector $\gamma^c_{\pi,\s}$.
Let $\bv = \gamma^c_{\pi,\sO}$ and $\bw = \gamma^c_{\pi,\sI}$. Then, by
\eqref{eq:gamma}, these vectors satisfy the condition $\bv + \bw = \bu$, so that, in the decomposition 
of $\bf_{\pi, \s}^{-1}(\bu)$, the set $A_{\bv, \bw}$ will have cardinality $\hat{y}_{\pi,\bv,\bw}$.
Thus, by ~\eqref{e4_c3}, $A_{\bv, \bw} \neq \emptyset$, so that we may ensure that $b_c$ is contained in this
set when we perform the decomposition of $\bf_{\pi, \s}^{-1}(\bu)$. But then we will have set 
$\bf_{\pi, \sO}(b_c) = \bv = \gamma^c_{\pi,\sO}$ and $\bf_{\pi, \sI}(b_c) = \bw = \gamma^c_{\pi,\sI}$, so that~\eqref{tmpNew1} holds with $\s$ replaced by $\sO$ and $\sI$.
This completes the induction. 

The construction of the functions $\bg_{\pi, \t}$ is completely analogous.
\end{proof}

\begin{lemC}[{\cite[Lemma 12]{pratt07}}]
\label{lem:total_count}
Let the functions $\bf$ and $\bg$ be constructed as in Lemma~\ref{tech_lem}.
Then, for all $a \in A_\pi$, we have
$$\sum \{ \bf_{\pi, \s'}(a) : |\s'| = p \} + \sum \{ \bg_{\pi, \t'}(a) : |\t'| = q \} = \bC.$$
\end{lemC}

\begin{proof}
We prove the stronger result that, for all $a \in A_\pi$, $j$ ($0 \leq
j \leq p$) and $k$ ($0 \leq k \leq q$),
\begin{eqnarray}
\sum \{\bf_{\pi,s'}(a) : |s'| = j \} + 
\sum \{\bg_{\pi,t'}(a) : |t'| = k \} & = & \bC,
\label{suff:2i}
\end{eqnarray}
using a double induction on $j$ and $k$.  If $j = k = 0$, then the
left-hand side of~\eqref{suff:2i} is simply $\bf_{\pi,\epsilon}(a) +
\bg_{\pi,\epsilon}(a)$, which is equal to $\bC$ by~\eqref{tmp3}.
Suppose now that the result holds for the pair $j$, $k$, with $j < p$.
Then
\begin{eqnarray*}
\ & \ & \sum \{\bf_{\pi,s'}(a) : |s'| = (j+1) \} + 
        \sum \{\bg_{\pi,t'}(a) : |t'| = k \} \\
\ & = & \sum \{\bf_{\pi,s'0}(a) + \bf_{\pi,s'1}(a)  : |s'| = j \} + 
        \sum \{\bg_{\pi,t'}(a) : |t'| = k \} \\
\ & = & \sum \{\bf_{\pi,s'}(a) : |s'| = j \} + 
        \sum \{\bg_{\pi,t'}(a) : |t'| = k \} \qquad \text{ by~\eqref{tmp4} } \\
\ & = & \bC \qquad \text{ by inductive hypothesis.} 
\end{eqnarray*}
This establishes the result for the pair $j+1,k$.  An analogous
argument using~\eqref{tmp5} applies when $k < m$, completing the
induction.
\end{proof}
We are now ready to prove the converse of Lemma~\ref{lma:orthoPlus}.

\begin{lem}
\label{antistropho}
Let $\Delta$, $\phi$ and $\cE$ be as given above. If $\cE$ has a solution over $\mathbb{N}$, then $\phi \wedge \Delta$ is finitely satisfiable.
\end{lem}

\begin{proof}
Let $x_\lambda$, $y_{\pi, \s, \bu}$, $z_{\pi, \t, \bu}$, $\hy_{\pi, \s, \bv, \bw}$, $\hz_{\pi, \t, \bv, \bw}$ be natural numbers satisfying $\cE$ (with the indices $\pi$, $\s$, $\t$, $\bu$, $\bv$, $\bw$ 
varying as usual). Notice that, for all positive integers $k$,
the (sets of) natural numbers $kx_\lambda$, $ky_{\pi, \s, \bu}$, $kz_{\pi, \t, \bu}$, $k\hy_{\pi, \s, \bv, \bw}$, $k\hz_{\pi, \t, \bv, \bw}$ also satisfy $\cE$.
Thus, we may assume that all values in the sought-after solution are either 0 or $\geq 3mC$.

We start by defining the universe $A$ of our model $\gA$.
Let $$A = \bigcup \{ A_\pi \mid \pi \text{ is any 1-type over } \sigma \},$$ where
each set $A_\pi$ has cardinality $$|A_\pi| = \sum \{ y_{\pi, \epsilon, \bu'} \mid \bu' \leq \bC \}$$ and the
sets $A_\pi$ are pairwise disjoint. Think of $A_\pi$ as the elements of $A$ that `want' to have 1-type $\pi$. Note that $A \neq \emptyset$, by the constraint~(\ref{e1_c7}).

Let the functions $\bf_{\pi, \s}$, $\bg_{\pi, \t}$ and elements $b_c$ be as constructed in Lemma~\ref{tech_lem}.
Think of $\bf_{\pi, \s}(a)$ as the $\s$-spectrum that $a$ `wants' to
have, $\bg_{\pi, \t}(a)$ as the $\t$-tally that $a$ `wants' to have, and $b_c$ the database element that
`wants' to be the denotation of $c$.
We are only interested in the values of these
functions when $|\s| = p$ and $|\t| = q$. Fix some such $\s$.
Decompose each $A_\pi$ into sets $\bf_{\pi, \s}^{-1}(\bu)$ (with $\bu$ varying),
for each vector $\bu$ with $\bO < \bu \leq \bC$. 
By the constraints~(\ref{e2_c1}) and Equation~(\ref{tmp1}),
decompose each of those $\bf_{\pi, \s}^{-1}(\bu)$
into pairwise disjoint (possibly empty) sets $A_\lambda$ with $|A_\lambda| = x_\lambda$, for all invertible message-types $\lambda \in \lps$ with $\bC_\lambda = \bu$. When performing this decomposition, for each individual constant $c$, if $\Delta$ specifies that $c$ sends an invertible message of type $\lambda$
to some other database element, then make sure that 
$b_c$ is assigned to the set $A_\lambda$: by~(\ref{e4_c1}), this set is at least as numerous as the number
of database elements sending a message of type $\lambda$, so that we never run out of elements.
Think of $A_\lambda$ as the set of elements in $A_\pi$ that `want' to send a single
invertible message of type $\lambda$.
The above process is repeated for all possible different
values of $\s$ (with $|\s| = p$), and each decomposition should be thought of as independent of each
other. Analogously, for each $\t$ with $|\t| = q$, $A_\pi$ is decomposed into pairwise disjoint sets $\bg_{\pi, \t}^{-1}(\bu)$ (with $\bu$ varying); and, again, those decompositions should be thought of as independent of each other.

Based on the above decompositions, we specify for each $a \in A_\pi$ a `mosaic piece' and show how to
assemble these pieces into a model of $\phi$. A mosaic piece is, informally, a collection of the messages
that $a$ `wants' to send. This collection might contain more than one message of
each type (or zero for that matter). Let $a \in A_\pi$; the mosaic piece corresponding to $a$ contains:
\begin{enumerate}[label=(\roman{enumi})]
  \item a single message labelled $\lambda_{a, \s}$ for each bit-string $\s$ with $|\s| = p$ if 
        $\bf_{\pi, \s}(a) \neq \bO$, where $\lambda_{a, \s}$ is the (unique) 2-type $\lambda$ for which
        $a \in A_\lambda$;
  
  \item $n_{a, \t}$ messages labelled $\mu_{\pi, \t}$ for each bit-string $\t$ with $|\t| = q$, where
        $n_{a, \t}$ is the (unique) natural number such that
        $\bg_{\pi, \t}(a) = n_{a, \t} \cdot \bC_{\mu_{\pi, \t}}$. Note that if $\bg_{\pi, \t}(a) = 0$
        then $n_{a, \t} = 0$, otherwise $n_{a, \t}$ exists by the constraints~(\ref{e2_c6}) and
        Equation~(\ref{tmp2}).
\end{enumerate}
The mosaic piece corresponding to $a$ is depicted in Fig.~\ref{fig:messages}. Here, for legibility,
we have identified the
strings $\s$ of length $p$ with the numbers $0, \dots, P-1$, and the 
strings $\t$ of length $q$ for which $\mu_{\pi,\t}$ is a (non-invertible) message-type with the
numbers $0, \dots, R-1$, for suitable $P$, $R$. The dashed radial arrows 
indicate that $a$ may or may not send a message of invertible type $\lambda_{a,\s} \in \Lambda_{\pi,\s}$;
the solid radial arrows 
indicate that $a$ sends $n_{a,\t}$ (possibly zero) messages of non-invertible type $\mu_{\pi,\t}$.

\def\arrowDashedAngleRadius(#1, #2, #3); {  
  \pgfmathparse{(#2)*cos(#1)} \let\x\pgfmathresult;
  \pgfmathparse{(#2)*sin(#1)} \let\y\pgfmathresult;
  \vertex(\x, \y);
  
  \pgfmathparse{(#2-.13)*cos(#1)} \let\z\pgfmathresult;
  \pgfmathparse{(#2-.13)*sin(#1)} \let\w\pgfmathresult;
  \draw[line width=.5pt, densely dashed, ->] (0, 0) -- (\z, \w);
  
  \pgfmathparse{(#2-.9)*cos(#1 - 8)} \let\a\pgfmathresult;
  \pgfmathparse{(#2-.9)*sin(#1 - 8)} \let\b\pgfmathresult;
  \node at (\a, \b) {#3};
}

\def\arrowNormalAngleRadius(#1, #2, #3); {  
  \pgfmathparse{(#2)*cos(#1)} \let\x\pgfmathresult;
  \pgfmathparse{(#2)*sin(#1)} \let\y\pgfmathresult;
  \vertex(\x, \y);
  
  \pgfmathparse{(.16)*cos(#1)} \let\zz\pgfmathresult;
  \pgfmathparse{(.16)*sin(#1)} \let\ww\pgfmathresult;
  \pgfmathparse{(#2-.13)*cos(#1)} \let\z\pgfmathresult;
  \pgfmathparse{(#2-.13)*sin(#1)} \let\w\pgfmathresult;
  \draw[line width=.5pt, ->] (\zz, \ww) -- (\z, \w);
  
  \pgfmathparse{(#2-.9)*cos(#1 - 8)} \let\a\pgfmathresult;
  \pgfmathparse{(#2-.9)*sin(#1 - 8)} \let\b\pgfmathresult;
  \node at (\a, \b) {#3};
}

\def\lineDashedAngleRadius(#1, #2); {  
  \pgfmathparse{(#2)*cos(#1)} \let\z\pgfmathresult;
  \pgfmathparse{(#2)*sin(#1)} \let\w\pgfmathresult;
  \draw[line width=.5pt, densely dashed] (0, 0) -- (\z, \w);
}

\def\lineNormalAngleRadius(#1, #2); {  
  \pgfmathparse{(.16)*cos(#1)} \let\zz\pgfmathresult;
  \pgfmathparse{(.16)*sin(#1)} \let\ww\pgfmathresult;
  \pgfmathparse{(#2)*cos(#1)} \let\z\pgfmathresult;
  \pgfmathparse{(#2)*sin(#1)} \let\w\pgfmathresult;
  \draw[line width=.5pt] (\zz, \ww) -- (\z, \w);
}

\def\arc(#1, #2, #3, #4); {  
  \draw[line width=.5pt, dashed, <->] (#1:#3) arc[radius=#3, start angle=#1, end angle=#2];
  \pgfmathparse{(#3+.3)*cos((#1 + #2)/2)} \let\x\pgfmathresult;
  \pgfmathparse{(#3+.3)*sin((#1 + #2)/2)} \let\y\pgfmathresult;
  \node at (\x, \y) {#4};
}

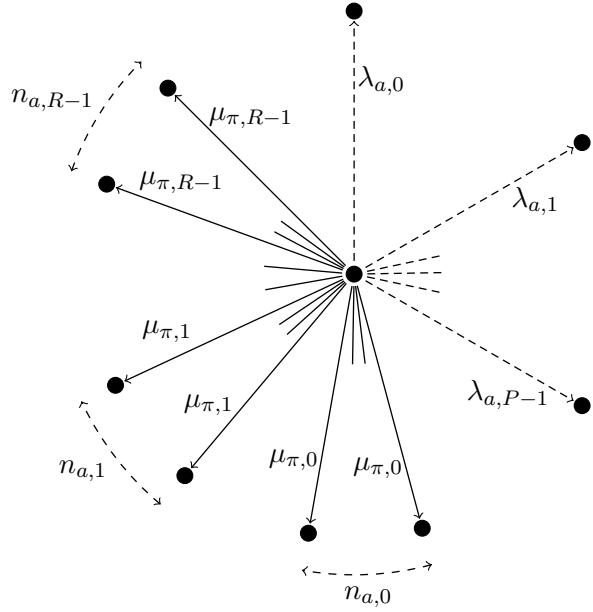
\begin{figure}
\centering
\begin{tikzpicture}

\vertex(0, 0);
\arrowDashedAngleRadius(90, 3.5, $\lambda_{a,0}$);
\arrowDashedAngleRadius(30, 3.5, $\lambda_{a,1}$);

\lineDashedAngleRadius(12, 1.2);
\lineDashedAngleRadius(1, 1.2);
\lineDashedAngleRadius(-12, 1.2);

\arrowDashedAngleRadius(-30, 3.5, $\lambda_{a,P-1}$);
\arrowNormalAngleRadius(-75, 3.5, $\mu_{\pi,0}$);

\lineNormalAngleRadius(-83, 1.2);
\lineNormalAngleRadius(-91, 1.2);
\arc(-75, -100, 4, $n_{a,0}$);

\arrowNormalAngleRadius(-100, 3.5, $\mu_{\pi,0}$);
\arrowNormalAngleRadius(-130, 3.5, $\mu_{\pi,1}$);

\lineNormalAngleRadius(-138, 1.2);
\lineNormalAngleRadius(-146, 1.2);
\arc(-130, -155, 4, $n_{a,1}\quad$);

\arrowNormalAngleRadius(-155, 3.5, $\mu_{\pi,1}$);
\lineNormalAngleRadius(-170, 1.2);
\lineNormalAngleRadius(-185, 1.2);
\arrowNormalAngleRadius(-200, 3.5, $\mu_{\pi,R-1}$);

\lineNormalAngleRadius(-208, 1.2);
\lineNormalAngleRadius(-216, 1.2);
\arc(-200, -225, 4, $n_{a,R-1}\qquad$);

\arrowNormalAngleRadius(-225, 3.5, $\quad\mu_{\pi,R-1}$);

\end{tikzpicture}
\caption{The messages sent by $a \in A_\pi$.}
\label{fig:messages}
\end{figure}

Let $a \in A$ and define $\bC_a$ to be the vector $(C_{a, 1}, \ldots, C_{a, m})$ whose $i$th
coordinate $C_{a, i}$ records the number of messages in the mosaic piece of $a$ containing
an outgoing $o_i$ arrow---i.e.~messages having label $\nu$ for which $o_i(x, y) \in \nu$.
Clearly (see Fig.~\ref{fig:messages}), $$\bC_a = \sum \{ \bf_{\pi, \s'}(a) : |\s'| = p \}
+ \sum \{ \bg_{\pi, \t'} : |\t'| = q \}$$ and, by Lemma~\ref{lem:total_count},
\begin{equation}
\label{tmp10}
\bC_a = \bC.
\end{equation}

We now build $\gA$ in four steps as follows.

\vspace{1em}
\noindent\textbf{Step 1} (Fixing the 1-types) \hspace{.5em} For all 1-types $\pi$ and all
$a \in A_\pi$, set $\tpA{a} = \pi$. Since 
 1-types do not contain equality literals such as $x = c$ or $x \neq d$, this is meaningful.
Moreover, since
the sets $A_\pi$ are pairwise disjoint, no clashes arise.
For each individual constant $c$, let $c^\gA = b_c$. Since $\Delta$ contains naming formulas
$c(c)$ and $\neg c(d)$, where $c$ and $d$ are distinct, it is obvious that the unique names assumption
is respected: that is, $c^\gA \neq d^\gA$ for $c$ and $d$ distinct. 
Therefore, we may, as before, write $c$ instead of the more correct $c^\gA$, since no confusion arises.

\vspace{1em}
\noindent\textbf{Step 2} (Fixing the invertible message-types) \hspace{.5em} We first assign the invertible
message-types for all pairs $c, d \in A$, where $c$ and $d$ are distinct individual constants,
as dictated by $\Delta$. To see that this is possible, suppose $\Delta \models \lambda(c,d)$, where
$\lambda \in \lps$, say. Then by Lemma~\ref{tech_lem}, $c$ will `want' to have $\s$-spectrum $\gamma^c_{\pi,\s}$, which, by the assumed compatibility of $\Delta$ with these constants, implies
that $\gamma^c_{\pi,\s} = \bC_\lambda$. That is: the mosaic piece for $c$ will send an invertible message of type $\lambda$. Similarly,  
the mosaic piece for $d$ will send an invertible message of type $\lambda^{-1}$; and these two messages may be paired up with
each other. We then put, by the constraints~(\ref{e2_c2}), all other $\lambda$-labelled messages and all $\lambda^{-1}$-labelled messages in one-to-one correspondence, for each invertible message of type $\lambda$.
Thus, if $a$ sends a $\lambda$-labelled message and $b$ `wants to receive it' (i.e.~sends a
$\lambda^{-1}$-labelled message), we set $\tpA{a, b} = \lambda$. 
To ensure that each assignment $\tpA{a, b}$ ($a, b \in A$) is valid, we
need only check that $a$ and $b$ are distinct. But, since $x_\lambda > 0$, by the
constraints~(\ref{e2_c3}) we must have $\tpA{a} \neq \tpA{b}$ hence, by construction, $a$ and $b$
belong to the disjoint sets $A_{\tpA{a}}$ and $A_{\tpA{b}}$. Moreover, since every element
sends at most one invertible message of each type, no conflicts with the present assignment will
arise in future assignments. 

\vspace{1em}
\noindent\textbf{Step 3} (Fixing the non-invertible message-types) \hspace{.5em}
Start by decomposing
each non-empty set $A_\pi$ into three pairwise disjoint (possibly empty) sets $A_{\pi, 0}$, $A_{\pi, 1}$ and
$A_{\pi, 2}$ having at least $mC$ elements each, since $|A_\pi| \geq 3mC$, and with the following
restriction: if $c^\gA \in A_\pi$, for some constant $c\in \cO$ (i.e.~$\pi = \tp{\Delta}{c}$),
pick these three sets such that $c^\gA \in A_{\pi, 0}$. This is possible by our choise of solution
of $\cE$.

Let $\mu_{\pi, \t}$ be any
non-invertible message type, with $\pi = \tpStart{\mu_{\pi, \t}}$ and $\rho = \tpEnd{\mu_{\pi, \t}}$
being its starting and terminal 1-types. Let $a \in A$ be an element that sends $n_{a, \t} > 0$ messages
of type $\mu_{\pi, \t}$. Clearly, then, $a \in A_\pi$ and there is a vector $\bu > \bO$ such that
$\bg_{\pi, \t}(a) = \bu$, hence $\bg_{\pi, \t}^{-1}(\bu)$ is non-empty. As a result, $z_{\pi, \t, \bu} =
|\bg_{\pi, \t}^{-1}(\bu)|$ is positive, thus, by the constraints~(\ref{e3_c1}), $\sum \{ y_{\rho, \epsilon,
\bu'} \mid \bu' \leq \bC \}$ is also positive. This implies that $A_\rho$ is non-empty since, clearly,
$|A_\rho| = \sum \{ y_{\rho, \epsilon, \bu'} \mid \bu' \leq \bC \}$ and hence has been partitioned
into three sets $A_{\rho, 0}$, $A_{\rho, 1}$ and $A_{\rho, 2}$ having at least $mC$ elements each.

Let us assume for the moment that $a$ is not equal to the interpretation of any constant. We employ 
the standard `circular witnessing' technique of ~\cite{gradel1997decision}.
The non-empty set $A_\pi$ has been partitioned into 
$A_{\pi, 0}$, $A_{\pi, 1}$ and $A_{\pi, 2}$; since $a \in A_\pi$, let $j$ be such that $a \in
A_{\pi, j}$, $0 \leq j \leq 2$. Now, let $k = j + 1 \, (\text{mod } 3)$ and select $n_{a, \t}$ elements
$b$ from $A_{\rho, k}$ that have not already been chosen to receive any messages (invertible or
non-invertible) and set, for each one of those, $\tpA{a, b} = \mu_{\pi, \t}$. Note that there are enough
elements in $A_{\rho, k}$ to choose from, as $a$ can send at most $mC$ messages (of invertible
or non-invertible type). Suppose, on the other hand, 
$a$ is (the interpretation of) some constant, say $c$, and that $\Delta$ specifies that
$c$ sends messages of type $\mu$ to $k$ other database elements.
Then by Lemma~\ref{tech_lem}, $c$ will `want' to have $\t$-spectrum $\delta^c_{\pi,\t}$, which, by the assumed compatibility of $\Delta$ with these constants, implies
that $\delta^c_{\pi,\t} \geq k \cdot \bC_\mu$. That is: the mosaic piece for $c$ will send at least $k$ non-invertible messages of type $\lambda$. Now set $\tpA{c, d} = \mu_{\pi, \t}$ for each of the $k$ elements as required by $\Delta$,
subtract $k$ from the value $n_{a, \t}$ to take account of the fact that these non-invertible message of type $\mu_{\pi, \t}$ have been dealt with, and proceed as 
before. Since $\mu$
is not an invertible 2-type, we can be sure that it has not already been set during Step 2. Moreover, `circular witnessing' ensures that no clashes arise during Step 3. 

\vspace{1em}
\noindent\textbf{Step 4} (Fixing the remaining 2-types) \hspace{.5em} 
If $c$ and $d$ are distinct individual constants for which $\tp{\gA}{c,d}$ has not been defined, set
$\tp{\gA}{c,d} = \tp{\Delta}{c,d}$; of course, this must a silent type, since all messages sent within the database have been accounted for.
If $\tpA{a, b}$ has not yet been defined, set it to be the 2-type $$\pi \cup \rho[y / x]
\cup \{ \lnot e \mid e \text{ is a guard-atom not involving} = \},$$ where $\pi = 
\tpA{a}$, $\rho = \tpA{b}$ and $\rho[y / x]$ is the result of replacing $x$ by $y$ in $\rho$. Note
that, since $C_1, \ldots, C_m$ are by assumption positive integers, $a$ and $b$ certainly send
some messages and, thus, the constraints~(\ref{e2_c4}) and~(\ref{e2_c5}) ensure that both
$\alpha \land \pi$ and $\alpha \land \rho$ are satisfiable.
Note also that this is the point in the proof where we make essential use of
of the fact that $\phi$ is guarded. In particular, the conjuncts $\forall x \forall y(e_j(x, y) \rightarrow (\ldots))$
in (\ref{eq:nf}) are satisfied by the assignments in this step, because the antecedents
are false. 

\vspace{1em}
This completes the definition of $\gA$ and we now show that $\gA \models \phi \land \Delta$. Referring
to the normal form in Lemma~\ref{lem:gc2_nf}, notice that none of the above steps violates the conjuncts
$$\forall x \, \alpha \land \bigwedge_{1 \leq j \leq n} \forall x \forall y (e_j(x,y) \rightarrow
(\beta_j \lor x = y)).$$ Furthermore, the conjuncts $$\bigwedge_{1 \leq i \leq m} \forall x \exists_{=C_i} \, y (o_i(x, y) \land x \neq y)$$ are all satisfied taking into account
Equation~(\ref{tmp10}) and
the fact that none of the 2-types assigned in Step 4 is a message type. The database $\Delta$ is evidently
satisfied by the above construction. 
\end{proof}

Now, observe that all the constraints in $\cE$ have the forms
\begin{align*}
  x_1 + \ldots + x_n \: &= \:  x, \\
  x_1 + \ldots + x_n \: &\geq \:  c, \\
                   x \: &= \:    0, \\
                x_1  \: &= \: x_2, \\
                  x \: &\geq \: c, \\
  x > 0 \: \Rightarrow \: x_1 + \ldots + x_n \: &> \: 0,
\end{align*}
where $x$, $x_1, \ldots x_n$ are variables and $c$ is a constant. Recall that the size $\|\cE\|$ of $\cE$
is exponential in $\sizeOf{\phi \wedge \Delta}$. Our goal is to find a solution
of $\cE$ in $\mathbb{N}$. The following lemma shows that we can transform a system of the
above form into an integer programming problem which, in turn, can be regarded as a linear
programming problem. This is important because linear programming
is in $\p$, whereas integer programming is in $\np$.
\begin{lem}[\cite{calvanese1996unrestricted}; see also \cite{pratt07}, Lemma 15]
\label{papadimitriou}
Let $\Delta$, $\phi$ and $\cE$ be as above. An algorithm exists to determine whether $\cE$ has a
solution over $\mathbb{N}$ in time bounded by a polynomial function of $\|\cE\|$.
\end{lem}

We have now established the required upper complexity-bound for finite satisfiability in $\GCD$. 
\begin{thm}
\label{thm:main}
The finite satisfiability problem in $\GCD$ is in $\dexp$.
\end{thm}

\begin{proof}
Let a $\GCD$-formula $\psi \wedge \Gamma$ and $\psi$ be given. By Lemma~\ref{lem:gc2_nf}, convert $\psi$ to a
formula $\phi$ of the form~\eqref{eq:nf}, and let $\sigma$ be the signature of $\phi$ together with additional
unary predicates:  
all the naming predicates for the individual constants in $\Delta$ and $\lceil \log((mC)^2+1) \rceil$ spare predicates. If $\Gamma$ is inconsistent, fail; otherwise, 
add all the naming formulas to $\Gamma$, and 
guess a consistent completion $\Delta \supseteq \Gamma$.
Having fixed $\phi$, $\sigma$  and $\Delta$, compute,
for each individual constant $c$, the sets of strings $\CcomS_c$ and $\CcomT_c$, and then guess
the 
vectors $\gamma_{\pi, \s}^c$, for
$\s \in \CcomS_c$
and $\delta_{\pi, \t}^c$, for 
$\t \in \CcomT_c$. Check that these vectors satisfy the 
conditions~\eqref{eq:gammadelta}--\eqref{eq:delta}, and that they are consistent with $\Delta$, failing if not. Compute the numbers $\eta_\lambda$ from $\Delta$.
Now write the system of equations $\cE$. Determine whether $\cE$ has a solution over $\mathbb{N}$.
By Lemma~\ref{papadimitriou}, this can be determined in exponential time. 
If $\cE$ has a solution for any of these guesses, then succeed; otherwise fail.
By Lemmas~\ref{ortho} and~\ref{antistropho}, the above procedure has a successful run (for some guess) 
if and only if $\psi \wedge \Gamma$ has a finite model. All guesses involve only a polynomial amount of
data, and so may be exhaustively searched in exponential time.
\end{proof}

Having dealt with the finite satisfiability problem for $\GC$, we turn now to the corresponding
satisfiability problem. We employ `extended arithmetic', over the set $\mathbb{N} \cup \{ \aleph_0 \}$,
where addition and multiplication are extended in the obvious way, i.e.~$\aleph_0 + \aleph_0 =
\aleph_0 \cdot \aleph_0 = \aleph_0$, $n + \aleph_0 = \aleph_0 + n = \aleph_0$, for all
$n \in \mathbb{N}$, etc. After all the gruel we have just chomped our way through, Theorem~\ref{thm:mainGeneral} is dessert.
\begin{thm}
\label{thm:mainGeneral}
The satisfiability problem in $\GCD$ is in $\dexp$.
\end{thm}

\begin{proof}[Sketch proof]
We proceed exactly as in the finite case, except that we seek solutions to $\cE$ over extended arithmetic. 
It is evident that, if $\bv$ is a solution of $\cE$, then so is $\aleph_0 \bv$. Thus, we may confine attention
to solutions over the 2-element set $\{0, \aleph_0\}$. Such a system is
essentially Boolean, so its constraints can be viewed as formulas of propositional logic; and,
with a little care, can be written as Horn clauses. 
(See \cite{pratt07} for more details.) 
This establishes membership in
$\dexp$.
\end{proof}

It is well known that the satisfiability and finite satisfiability problems for $\GC$ are
$\dexp$-hard. Thus, all the complexity bounds given above are tight.

\bibliographystyle{alpha}

\end{document}